\documentclass[opre,nonblindrev]{informs4rad} 

\newcommand{\revcolor}[1]{{\color{black}#1}}
\newcommand{\oldcolor}[1]{{\color{black}#1}}
\newcommand{\revcolorCMTwo}[1]{{\color{black}#1}}
\newcommand{\revcolorCAMtwo}[1]{{\color{black}#1}}
\newif\ifMS
\MStrue
\OneAndAHalfSpacedXI

\newif\iffigdraft

\TheoremsNumberedThrough 
\EquationsNumberedThrough

\usepackage[T1]{fontenc}
\usepackage[utf8]{inputenc}
\usepackage{newtxtext}
\usepackage{bbm}

\usepackage[normalem]{ulem}
\usepackage{xfrac}
\usepackage{xcolor}
\usepackage{graphicx}
\usepackage{subfigure}

\usepackage{thm-restate}
\usepackage[ruled,vlined,linesnumbered]{algorithm2e}
\SetNoFillComment

\usepackage[breakable, skins]{tcolorbox}
        
\usepackage{csquotes}

\makeatletter
\renewenvironment*{displayquote}
  {\begingroup\setlength{\leftmargini}{0cm}\csq@getcargs{\csq@bdquote{}{}}}
  {\csq@edquote\endgroup}
\makeatother
\DeclareRobustCommand{\mybox}[2][gray!15]{
\begin{tcolorbox}[ 
        colback=white,      
        colframe=gray,  
        boxrule=0.2pt,      
        arc=2pt,outer arc=2pt,
        left=12pt,
        right=12pt,
        top=5pt,
        bottom=5pt,
        width=1.07\linewidth,
        enlarge left by=-0.55cm,
        before upper=\renewcommand{\baselinestretch}{1.3}\selectfont,
        after upper=\normalfont
        ]
 #2
 \end{tcolorbox}
}

\usepackage{fix-cm}

\usepackage{mathtools}
\definecolor{cornellred}{rgb}{0.7, 0.11, 0.11}
\definecolor{maroon}{rgb}{0.52, 0, 0}
\definecolor{dgreen}{rgb}{0.0, 0.5, 0.0}
\definecolor{ballblue}{rgb}{0.13, 0.67, 0.8}
\definecolor{royalblue(web)}{rgb}{0.25, 0.41, 0.88}
\definecolor{bleudefrance}{rgb}{0.19, 0.55, 0.91}
\definecolor{royalazure}{rgb}{0.0, 0.22, 0.66}
\usepackage[hypertexnames=false]{hyperref}
\hypersetup{
	colorlinks = true,
	linkcolor=royalazure,
	urlcolor=royalazure,
	citecolor=royalazure,
	linkbordercolor = {white}
}
\usepackage{cleveref}
\usepackage{enumitem}
\usepackage{multirow}

\usepackage{pgf,tikz,pgfplots}

\usetikzlibrary{decorations.pathreplacing,calc}

\pgfplotsset{compat=1.15}
\usetikzlibrary{arrows}
\usepackage{pgfplots}
\usetikzlibrary{patterns}
\usepgfplotslibrary{fillbetween}
\usetikzlibrary{intersections}

\usetikzlibrary{arrows, decorations.markings}

\tikzstyle{vecArrow} = [thick, decoration={markings,mark=at position
	1 with {\arrow[semithick]{open triangle 60}}},
	double distance=1.4pt, shorten >= 5.5pt,
	preaction = {decorate},
postaction = {draw,line width=1.4pt, white,shorten >= 4.5pt}]
\tikzstyle{innerWhite} = [semithick, white,line width=1.4pt, shorten >= 4.5pt]

{\theoremstyle{THkey}\newtheorem{informal}{Theorem (informal)}}

	\usepackage{accents}

	\usepackage[authoryear,round]{natbib}

	\newcommand{\reals}{\mathbb{R}}
	\newcommand{\naturals}{\mathbb{N}}

\ifMS
\else
	\DeclareMathOperator{\argmax}{argmax}
	
\fi

\newcommand{\allocTilde}{\tilde{\alloc}}
\newcommand{\edgeallocijp}{\edgealloc_{ij}^{\permu}}\newcommand{\pre}{\texttt{pre}}
\newcommand{\nxt}{\texttt{suc}}
\newcommand{\suc}{\nxt}
\newcommand{\randseed}{\eta}
\newcommand{\prob}[2][]{\text{\bf Pr}\ifthenelse{\not\equal{}{#1}}{_{#1}}{}\!\left[{\def\givenn{\middle|}#2}\right]}
\newcommand{\expect}[2][]{\text{\bf E}\ifthenelse{\not\equal{}{#1}}{_{#1}}{}\!\left[{\def\givenn{\middle|}#2}\right]}

\newcommand{\indicator}[1]{{\mathbbm{1}\left\{ #1 \right\}}}

\newcommand{\condition}{\,\mid\,}

    \DeclareMathOperator{\EX}{\mathbb{E}}
	\DeclareMathOperator{\OPT}{\textsc{OPT}}

	\DeclarePairedDelimiterX{\set}[1]\{\}{#1}
	\let\Pr\relax
	\DeclarePairedDelimiterXPP{\Pr}[1]{\mathbb{P}}[]{}{#1}
	\DeclarePairedDelimiterXPP{\Ex}[1]{\mathbb{E}}[]{}{#1}

\newcolumntype{P}[1]{>{\centering\arraybackslash}c{#1}}

\newcommand*{\rom}[1]{\expandafter\romannumeral #1}
\newcommand{\Rom}[1]{\uppercase\expandafter{\romannumeral #1\relax}}

\newcommand{\permu}{\pi}

\newcommand{\xhdr}[1]{\smallskip\noindent{\bf {#1}\ }}

\newcommand{\optCRgen}{\Gamma_{\textsc{gen}}(\buybackcost)}
\newcommand{\optCRdet}{\Gamma_{\textsc{det-int}}(\buybackcost)}

\newcommand{\optCRgenfBar}{\Gamma_{\textsc{gen}}(\Bar{\buybackcost})}
\newcommand{\optCRgenexf}{\Gamma_{\textsc{gen}}(\EX[\buybackcost])}

\newcommand{\dd}{{\mathrm d}}

\newcommand{\bipartitegraph}{G}
\newcommand{\onlinenodes}{U}
\newcommand{\offlinenodes}{V}
\newcommand{\Edge}{E}

\newcommand{\totalresource}{n}
\newcommand{\totaltime}{T}
\newcommand{\weightc}{\mathcal{W}}
\newcommand{\weight}{w}
\newcommand{\largenumber}{K}
\newcommand{\weightij}{\weight_{ij}}
\newcommand{\weightijConfig}{\weight_{ijC}}
\newcommand{\weighti}{\weight_i}
\newcommand{\buybackcost}{f}
\newcommand{\ALG}{\textsc{ALG}}
\newcommand{\CR}{\textsc{CR}}
\newcommand{\instance}{I}
\newcommand{\instances}{\mathcal{I}}
\newcommand{\approxratio}{\Gamma}
\newcommand{\alloc}{x}
\newcommand{\calloc}{y}
\newcommand{\allocj}{\alloc_j}
\newcommand{\callocj}{\calloc_j}

\newcommand{\continstances}{\instances_{\textrm{cont}}}

\newcommand{\onesup}{^{(1)}}

\newcommand{\calloconesup}{\calloc\onesup}

\newcommand{\buybackweight}{\underline{\weight}}
\newcommand{\buybackweightonesup}{\buybackweight\onesup}

\newcommand{\Lambert}{W_{-1}}
\newcommand{\Lambertterm}{\Lambert\left(\frac{-1}{e(1+\buybackcost)}\right)}
\newcommand{\pen}{\Psi}

\newcommand{\penscalar}{\tau}
\newcommand{\offlinedual}{\beta}
\newcommand{\onlinedual}{\alpha}
\newcommand{\diracdeltafunction}{\delta}
\newcommand{\demand}{d}
\newcommand{\demandi}{\demand_i}
\newcommand{\capacity}{s}

\newcommand{\buybackweightj}{\buybackweight_j}
\newcommand{\approxratiodet}{\hat{\approxratio}_{\penscalar}}
\newcommand{\detthreshold}{\frac{1}{3}}

\newcommand{\offlinedualj}{\offlinedual_j}
\newcommand{\onlineduali}{\onlinedual_i}
\newcommand{\expfuncparamone}{\lambda}
\newcommand{\expfuncparamtwo}{\tau}
\newcommand{\approxratioexp}{\approxratio_{(\expfuncparamone,\expfuncparamtwo)}}
\newcommand{\fracthreshold}{\frac{e-2}{2}}
\newcommand{\edgealloc}{z}
\newcommand{\edgeallocij}{\edgealloc_{ij}}
\newcommand{\edgeallocijConfig}{x_{ijC}}
\newcommand{\configallocij}{\edgealloc_{iC}}
\newcommand{\edgealloci}{\edgealloc_{i}}
\newcommand{\weightipj}{\weight_{i'j}}
\newcommand{\weightip}{\weight_{i'}}
\newcommand{\penderivative}{\psi}

\newcommand{\weightmax}{\weight_{\max}}

\newcommand{\primed}{^\dagger}
\newcommand{\weightprimed}{\weight\primed}
\newcommand{\callocwpsup}{\calloc^{(\weightprimed)}}
\newcommand{\buybackweightwpsup}{\buybackweight^{(\weightprimed)}}

\newcommand{\thresholdweight}{{\hat\weight}}
\newcommand{\canalloc}{\hat{\alloc}}

\newcommand{\price}{p}

\newcommand{\alloconesup}{\alloc\onesup}
\newcommand{\largeconstant}{\largenumber}
\newcommand{\allocwpsup}{\alloc^{(\weightprimed)}}

\iffigdraft
\usepackage[allfiguresdraft]{draftfigure}
\fi

\newcommand{\inventory}{s}
\newcommand{\mininventory}{\inventory_{\min}}

\usepackage{color-edits}
\addauthor{yf}{purple}    
\addauthor{fe}{blue}
\addauthor{RN}{red}

\newcommand{\assortment}{S}
\newcommand{\choice}{\mathcal{C}}

\newcommand{\config}{C}
\newcommand{\Config}{\mathcal{C}}

\newcommand{\fullallocthreshold}{\largenumber_0(\buybackcost)}

\begin{document}

 \RUNAUTHOR{Ekbatani, Feng and Niazadeh}

\RUNTITLE{\revcolor{Online Resource Allocation with Cancellations}}

\TITLE{\revcolor{Online Resource Allocation with  Cancellations}}

\ARTICLEAUTHORS{%
\AUTHOR{Farbod Ekbatani}
\AFF{University of Chicago Booth School of Business, Chicago, IL, \EMAIL{fekbatan@chicagobooth.edu}}
\AUTHOR{Yiding Feng}
\AFF{Hong Kong University of Science and Technology (HKUST), Hong Kong, \EMAIL{ydfeng@ust.hk}}
\AUTHOR{Rad Niazadeh}
\AFF{University of Chicago Booth School of Business, Chicago, IL, \EMAIL{rad.niazadeh@chicagobooth.edu}}
} 

\ABSTRACT{%
\oldcolor{We initiate the study of two-sided online resource allocation with costly cancellations within the buyback model. Our focus is on edge-weighted online matching and several of its extensions. In our base matching model, nodes arrive online on one side and request offline resources on the other side. In contrast to the classic literature,} in our model the decision maker can reclaim any fraction of an offline resource that was preallocated to an earlier online node. However, reclaiming a resource not only results in the loss of the previously allocated edge-weight but also incurs an additional penalty equal to a non-negative constant factor $f$ times the edge-weight. Parameterizing the problem by the buyback factor $f$, our main result is the development of optimal competitive algorithms for \emph{all possible values} of $f$ through a novel primal-dual family of algorithms in the fractional (or equivalently, large capacity) setting. We establish the optimality of our results by deriving separate lower bounds for both the small and large buyback factor regimes and showing that our primal-dual algorithm exactly matches these lower bounds by appropriately tuning a parameter as a function of $f$. Interestingly, our results reveal a phase transition: for the small buyback regime ($f < \frac{e-2}{2}$), the optimal competitive ratio is $\frac{e}{e-(1+f)}$, and for the large buyback regime ($f \geq \frac{e-2}{2}$), the competitive ratio is $-W_{-1}\left(\frac{-1}{e(1+f)}\right)$, where $W_{-1}$ is the non-principal branch of the Lambert $W$ function.

We also study the lower and upper bounds on the competitive ratio in variants of this model, such as matching with deterministic integral allocations or single-resource environments with varying demand sizes. For deterministic integral matching, our results again show a phase transition: for the small buyback regime ($f < \frac{1}{3}$), the optimal competitive ratio is $\frac{2}{1-f}$, while for the large buyback regime ($f \geq \frac{1}{3}$), the competitive ratio is $1 + 2f + 2\sqrt{f(1+f)}$. \oldcolor{We further extend our model to settings with combinatorial configuration allocations, buyer-arriving submodular welfare maximization, and assortment planning. We also consider an extension with negative values of $f$, which models scenarios with secondary supply channels or overflow capacities available at discounted rates.} Our unifying family of primal-dual algorithms achieves the exact optimal competitive ratio across all these variants, demonstrating the power of our algorithmic framework for online resource allocations with cancellations. \oldcolor{We further complement our theoretical analysis by numerically evaluating the performance of our algorithms and validating our theoretical results in more realistic scenarios.}

}%


\maketitle
\setcounter{page}{1}


%


\newpage
\section{Introduction}
\label{sec:intro}
Real-time matching of online demand nodes to offline supply nodes is a central problem in revenue management, underpinning diverse applications ranging from digital advertising to resource allocation in platforms. A prominent approach to such problems is captured by the classic online bipartite matching model, introduced by the seminal work of \cite{KVV-90}, which is particularly relevant in robust decision-making scenarios where future demands are highly uncertain or non-stationary~\citep{BQ-09}. Numerous applications and extensions of this model have been studied in the computer science and operations research literature, spanning the assignment of keywords to advertisers in search advertising~\citep{MSVV-07, BN-09}, assigning impressions to advertisers in display advertising~\citep{FKMMP-09,DHKMY-16}, dynamic assortment planning in retail~\citep{GNR-14,MS-20}, real-time matching in ride sharing platforms~\citep{ABDJSS-19,HKTWZZ-20}, and several others. 

\oldcolor{A fundamental assumption in traditional online matching literature is that decisions made online are \emph{irrevocable}---once a demand node is matched to a supply node, this assignment cannot be altered later. While irrevocability is inherent in certain applications, in several others where it can be relaxed, it limits the flexibility of the decision-maker, especially in environments where future matches may offer significantly higher rewards. Relaxing irrevocability allows algorithms to revisit earlier, suboptimal assignments and potentially cancel them, enhancing the overall quality and efficiency of the final matching. 

Such ``value-enhancing'' flexibility is particularly relevant in certain modern applications of matching, and in general, online resource allocation. For instance, in digital banner advertising~\citep{MSN,BHK-09}, publishers often overbook different digital billboard slots (banners) on their websites during an ad campaign, initially allocating each banner to multiple advertisers and subsequently revoking lower-value assignments. Similarly, in display advertising, advertisers typically commit to purchasing up to a certain number of impressions. Allocation algorithms, however, initially assign impressions beyond these commitments, ultimately billing advertisers only for the highest-valued impressions within their capacity constraints, essentially canceling the assignment of excess impressions~\citep{FKMMP-09,DHKMY-16}. \revcolorCAMtwo{In another example, in broadcast advertising, television and radio stations sell ``preemptible'' airtime that may---by the very definition of the product---be displaced later by an advertisement purchased at a higher rate, in exchange for a make-good or a credit~\citep{FCC-73.1942,CJAdvertising-2023}.} \revcolorCAMtwo{Lastly, in classical revenue management contexts such as hotel, airline, and cargo management~\citep{rot-71,LY-78,kasilingam1997air}, providers may oversell (overbook) capacity. Overbooking is primarily to hedge against no-shows; yet once inventory is oversold, they may continue to accept premium demand and displace lower-value commitments (typically in exchange for a value-proportional compensation).}\footnote{\oldcolor{\label{footnote:applications-intro} See \Cref{app:practical} for further discussion of these and additional applications}\revcolorCAMtwo{, including a mapping between the primitives of our model and the primitives of each application.}}

Despite the potential benefits of allowing cancellations, unrestricted cancellation can lead to substantial negative externalities, causing dissatisfaction among users whose assignments are revoked, which leads to undermining market trust. \revcolorCMTwo{Moreover, cancellations may impose a direct monetary cost on the platform, or introduce an opportunity cost in the future.} As a result, the platform faces a tradeoff: on one hand, allowing more cancellations can enhance the value of the final match \revcolorCAMtwo{by avoiding early commitments to low-valued requests}, and on the other hand the total canceled reward has a\revcolorCAMtwo{---possibly explicit or implicit---}cost for the platform. \revcolorCAMtwo{Therefore, in scenarios admitting cancellations, it is paramount for the platform to \emph{control} the extent of cancellations made by its allocation policy:} \revcolorCMTwo{the platform benefits from an algorithmic policy that opens capacity for higher-paying future matches via cancellations, while systematically limiting excessive cancellations.}

To balance the flexibility benefits with the drawbacks associated with cancellations discussed above, we adopt a principled framework \revcolorCAMtwo{that captures such {controlled cancellations}}, known in the literature as the \emph{buyback model}~\citep{BHK-09,CFMP-09}. In this model, the decision-maker aims to maximize a mixed-sign objective function that explicitly accounts for costly cancellations. Specifically, the total reward obtained by serving (uncanceled) requests is penalized by the total reward of canceled requests, scaled by a parameter $\buybackcost \geq 0$. This parameter---often referred to as the \emph{buyback factor}---serves as the ``shadow price,'' or implicit cost, of canceling each unit of reward and quantifies how significantly cancellations impact the decision-maker's goal. Thus, this mixed-sign objective provides fine-grained control over the algorithm's cancellation behavior and interpolates smoothly between two extremes: completely free cancellations ($\buybackcost=0$) and entirely irrevocable decisions ($\buybackcost=+\infty$).

Before presenting our formal problem statement and model, we provide a few remarks about the specific cancellation penalty term introduced in our objective. From a technical viewpoint, this penalty is natural because it ensures invariance under scaling of rewards. From a practical perspective, it provides a simple and interpretable mechanism to control cancellations in revenue management settings. Additionally, this linear cancellation cost can also be interpreted as a \emph{buyback cost}: in certain applications, it represents a physical payment in the form of a ``compensation fee'' paid to users, typically a fixed percentage of the original transaction, when reclaiming previously allocated resources. \revcolorCMTwo{This is a practice frequently employed in real-world industries, such as the issuance of makegoods in media and advertising when an ad slot is preempted by a higher-paying buyer \citep{CJAdvertising-2023}, or standard compensation models in hotel and airline overbooking \citep{USDOT2011}}\revcolorCAMtwo{---other examples include 
pre-specified compensation payments for ``rolled'' cargo bookings \citep{MaerskRollable}, and ad valorem host-cancellation fees on short-term rental platforms \citep{AirbnbHostFee}}.\footnote{\label{footnote:fixed-f}\oldcolor{While we adopt a fixed function in this paper, which is simpler and more interpretable, personalized cancellation costs (when $\buybackcost$ is a function of time) may be interesting to investigate in some applications---which fall outside the scope of this paper.}}}

\smallskip
\noindent\textbf{Problem statement \& model:}
\oldcolor{Motivated by the above discussion, we initiate the study of \emph{two-sided online resource allocation with costly cancellations} under the buyback model.} Unlike previous studies that focused mainly on \emph{single-resource environments}~\citep{BHK-09,AK-09}---similar to the single-leg revenue management problem---our work addresses complex \emph{matching environments} with multiple resources, reflecting realistic settings with multiple advertisers or service providers. In our base model, the nodes on the online side of a bipartite graph arrive one by one (which are demands or requests), revealing their matching weights (or rewards) to the set of offline nodes (which are supply or resources) upon arrival~\citep{KVV-90, FKMMP-09}. The decision maker must determine how to match each arriving online node to offline nodes with capacity constraints either fractionally or integrally. She is allowed to allocate more than the capacity to an offline node, but must revoke the excess allocations by paying the buyback cost given the buyback factor $\buybackcost$. \oldcolor{Importantly, this way of cancellations is equivalent to an alternative online procedure in which past allocations can be revoked greedily online by discarding the smallest weight.}

Given this model, the decision maker aims to find an online matching (fractional or integral) that maximizes the total weight of the final assignments minus $\buybackcost$ times the total canceled weight. Our goal is to develop robust online algorithms capable of handling non-stationary arrivals. Building upon the seminal work of \cite{KVV-90} from the computer science literature and \cite{BQ-09} from the revenue management literature, we consider a framework with adversarial online arrivals.\footnote{\label{footnote:frac-int}\revcolorCMTwo{We primarily focus on deterministic (fractional or integral) algorithms, and therefore we do not distinguish between adaptive and oblivious adversaries. For our randomized integral results for the single-resource and multiple-resource under large capacity, which we elaborate later, we assume the adversary is oblivious.}} By parameterizing the problem with the buyback factor $\buybackcost$, we measure the performance of our online algorithms using the competitive ratio, defined as the worst-case ratio of the optimal offline objective (an omniscient benchmark) to that achieved by the online algorithm. Specifically, we ask the following research questions:

\begin{displayquote}
\mybox{\emph{(i)~Can we design simple (fractional) online algorithms for the above problem that achieve the optimal competitive ratio for \underline{all possible values} of the buyback factor $\buybackcost$? \\
\smallskip
(ii)~What if we restrict our attention to the simpler class of deterministic integral online algorithms?}}
\end{displayquote}

\noindent\textbf{Our main contributions.} 
We affirmatively answer both questions above by introducing a unified approach: \emph{We design and analyze a parametric family of primal-dual online algorithms that, after appropriate parameter tuning, achieve optimal competitive ratios for both fractional and deterministic integral allocations in the matching environment.} \oldcolor{We further show that in the single-resource environment---a special case of the matching setting---our primal-dual algorithms achieve improved and optimal competitive ratios under general parameter choices.} Additionally, using randomized rounding techniques, we show that any fractional algorithm can be converted into a randomized integral algorithm with negligible loss under large capacities in the matching environment \oldcolor{and exactly no loss even under small capacities in the single-resource special case.} Thus, our results provide a complete picture for designing online algorithms that achieve optimal competitive ratios in the edge-weighted online matching problem with buyback under large capacities, \oldcolor{as well as in the single-resource buyback problem with arbitrary capacities.} Our contributions are two-fold:

\smallskip
\textbf{(i)~Lower-bounds in different parameter regimes.}~We first establish separate lower-bounds of $\optCRgen$ and $\optCRdet$ on the optimal competitive ratios for the general setting and the deterministic integral setting, respectively. These lower-bounds (below) are drawn in \Cref{fig:CRs}; here, $\Lambert:[\sfrac{-1}{e},0)\rightarrow(-\infty,-1]$ is the non-principal branch of the Lambert W function:\footnote{In mathematics, the Lambert W function, also called the omega function or product logarithm, is the converse relation of the function $y(x)=x\cdot e^x$. The non-principal branch $x=W_{-1}(y)$ is the inverse relation when $x\leq -1$.} 
\begin{align*}
\optCRgen\displaystyle\triangleq
\left\{
	\begin{array}{lll}
		\frac{e}{e-(1+\buybackcost)}  &~&\mbox{if } \buybackcost \leq \frac{e-2}{2} \\
		&\\
		-\Lambert\left(\frac{-1}{e(1+\buybackcost)}\right) &~&\mbox{if } \buybackcost \geq \frac{e-2}{2}
	\end{array}
\right. ~~~,~~~
\optCRdet\displaystyle\triangleq
\left\{
	\begin{array}{lll}
		\frac{2}{1-\buybackcost}  &~&\mbox{if } \buybackcost \leq \frac{1}{3} \\
		&\\
	1 + 2\buybackcost + 
2\sqrt{\buybackcost(1+\buybackcost)} &~&\mbox{if } \buybackcost \geq \frac{1}{3}
	\end{array}
\right.
\end{align*}
To establish the lower bounds, we identify two sources of uncertainty that prevent an online algorithm from performing as well as the optimal offline solution: \emph{edge-wise uncertainty} and \emph{weight-wise uncertainty}. Intuitively, edge-wise uncertainty is related to the multi-resource aspect of our model and reflects the information-theoretic challenge of identifying the correct offline node to match with an arriving online node, given the uncertainty about future arrivals (hence, it persists even when $\buybackcost=0$). In contrast, weight-wise uncertainty is tied to the buyback aspect of our model and captures the difficulty of determining the edge with the largest weight among all edges adjacent to an offline node, given future uncertainty (hence, it persists even with a single resource as long as $\buybackcost>0$). Notably, as $\buybackcost$ increases, weight-wise uncertainty becomes more significant, while edge-wise uncertainty plays a more important role when $\buybackcost$ is smaller.

We use novel constructions in our worst-case instances for the competitive ratio to exploit the trade-off between the two sources of uncertainty mentioned above. Interestingly, we observe a sharp ``phase transition'' based on the parameter $\buybackcost$ regarding the behavior of the worst-case instance, both in fractional and deterministic integral settings. This phase transition divides the problem into two regimes. First, there is the \emph{small buyback cost} regime where $\buybackcost<\hat{\buybackcost}$, and the worst-case instance primarily relies on the multi-resource aspect of the model. Notably, the special case of edge-weighted online matching with \emph{free-disposal}, that is, $\buybackcost=0$~\citep{FKMMP-09,DHKMY-16,FHTZ-20} falls within this regime. Then, there is the \emph{large buyback cost} regime where $\buybackcost>\hat{\buybackcost}$, and the worst-case instance---which is, in fact, \emph{exactly} the same as the worst-case instance in the single-resource environment---only relies on the buyback aspect. The phase transition thresholds differ between the two settings: $\hat{\buybackcost}=\frac{e-2}{2}$ for the general setting and $\hat{\buybackcost}=\frac{1}{3}$ for the deterministic integral setting.

As alluded to, in the large buyback cost regime, our lower bounds in the matching environment for general (fractional or randomized) algorithms and deterministic integral algorithms match the bounds in \cite{AK-09} and \cite{BHK-09}, established only for the single-resource environment.  Furthermore, as $\buybackcost \to 0$ in the extreme case of the small buyback cost regime, the lower bound $\optCRgen$ converges to $\frac{e}{e-1}$, and $\optCRdet$ converges to $2$. These values correspond to the worst-case competitive ratios for edge-weighted online bipartite matching with free disposal for general fractional \oldcolor{(or randomized integral under large capacities)} and deterministic integral algorithms, respectively.

\begin{figure}[ht]
    \centering
    \input{figs/fig-competitive-ratio}
    \caption{Competitive ratio as a function of the buyback factor $\boldsymbol{\buybackcost}$: The {\color{blue} blue} solid (black dashed) curve represents the optimal competitive ratio $\boldsymbol{\optCRgen}$ for the matching environment (single-resource environment). The {\color{red}red} solid ({\color{gray}gray dashed}) curve represents the optimal competitive ratio $\boldsymbol{\optCRdet}$ for deterministic integral algorithms in the matching environment (single-resource environment). The solid blue curve and the dashed black curve coincide for buyback factors $\boldsymbol{\buybackcost \geq \frac{e-2}{2} \approx 0.359}$, while the solid red curve and the dashed gray curve coincide for buyback factors $\boldsymbol{\buybackcost \geq \frac{1}{3}}$.
    } \label{fig:CRs}
\end{figure}

\smallskip
\textbf{(ii)~Tight upper-bounds via primal-dual in all parameter regimes.}
After establishing lower bounds, we provide \emph{exact} matching upper bounds by introducing a novel primal-dual family of online algorithms for this problem. The primal-dual framework has proven effective in designing and analyzing online bipartite allocation algorithms with irrevocable allocations or free cancellations across various contexts. \footnote{Notable examples include the well-known online bipartite matching problem with vertex weights~\citep{AGKM-11}, the online budget allocation problem (Adwords)~\citep{MSVV-07,BN-09}, online assortment optimization with or without reusable resources~\citep{GNR-14,FNS-19,GIU-20,GGISUW-22,DFNSU-24}, and the special case of online edge-weighted bipartite matching with free disposal, i.e., $\buybackcost=0$~\citep{FKMMP-09,DHKMY-16}.} However, until now, this framework has not been applied to settings where cancellation is costly, and it has remained unclear whether it could be beneficial in such settings.

The main challenge stems from the fact that the optimal offline benchmark does not involve buybacks, meaning it incurs no cancellation costs. As a result, the linear program for the optimal offline solution \emph{does not} take into account any aspects of the buyback feature in our problem. Prior work on single-resource buyback~\citep{BHK-09,AK-09} have not proposed systematic approaches to address this issue; instead, they rely on simple variations of greedy algorithms or ad hoc randomized schemes (to round the weights) to develop competitive online algorithms. It is not clear whether any of these ad hoc methods can be extended to the matching environment. This gap in understanding the buyback problem from the prior literature makes the development of an optimal competitive online algorithm for edge-weighted online matching within the buyback model highly challenging.

In a nutshell, our main technical and conceptual contribution is bridging \revcolorCAMtwo{this} existing gap in the literature by establishing a new connection between the primal-dual framework for online bipartite resource allocation problems and the buyback problem. This newfound understanding enables us to derive optimal competitive algorithms for all regimes of the buyback parameter $\buybackcost$.
We begin by revisiting the single-resource buyback problem through the systematic lens of the primal-dual framework. This approach allows us to develop new optimal competitive  primal-dual online algorithms for both fractional and integral single-resource buyback problems. More interestingly, we demonstrate how to naturally extend our primal-dual fractional algorithm to the matching environment, leading to our main result. Additionally, we show that natural adaptations of this approach in the integral setting yield optimal competitive ratios for deterministic integral algorithms.

\smallskip

\noindent\textbf{Our techniques.} We summarize our technical findings below. The complete technical story in this paper is elaborate and we refer the reader for the details to Sections~\ref{sec:single-resource},~\ref{sec:matching}, \ref{sec:deterministic}, and \ref{sec:lower-bounds-apx}.

\smallskip

\revcolorCMTwo{
\paragraph{Single-resource environment: allocation distributions \& a direct primal-dual algorithm.} As a warm-up, and to build the required machinery for designing and analyzing algorithms in the matching environment, we start by revisiting the single-resource environment, where only one offline node is present. In \Cref{sec:single resource primal dual}, we propose a family of primal-dual online algorithms that runs directly on general instances of this environment, based on a simple (almost trivial) LP relaxation of the optimal offline problem. To describe these algorithms, we rely on the notion of \emph{allocation distributions}, a key concept central to the design and analysis of our algorithms in this paper: upon each arrival, we consider all past (fractional) allocations of the offline node that have not been revoked, effectively capturing the complete state of the allocation by recording the distribution of these non-revoked weights.\footnote{\label{footnote:overbooking} When over-allocation occurs, it is optimal to perform buybacks from the smallest allocated weight until the capacity is no longer exceeded. Although this buyback process typically occurs at the end, it can equivalently be viewed as an online procedure that continuously performs buybacks from the smallest weight whenever the total allocation exceeds capacity. As a result, the allocation distribution can be tracked in an online fashion.}
Equipped with this notion, the algorithm maintains a single dual variable $\beta$---which can be interpreted as the current \emph{opportunity cost} of consuming the resource---and (continuously and fractionally) allocates to an arriving online node as long as its weight exceeds $\beta$, greedily buying back from the smallest allocated weights to free up capacity. Inspired by \cite{DHKMY-16}, we specify a particular way to update the dual variable $\beta$: we integrate a penalized version of the allocation distribution's quantile function, where the penalization is done using a \emph{penalty function} $\pen:[0,1]\rightarrow\mathbb{R}_{\geq 0}$ (see \Cref{alg:primal-dual-direct-single}, line~(7)). With this formulation at hand, we show through a direct dual fitting analysis that a proper choice of the penalty function makes the resulting algorithm optimally competitive.

But why is this way of updating the dual variable $\beta$ the ``right'' design choice? To answer this question, we develop an alternative approach for the single-resource environment (\Cref{sec: apxsec3}), which pinpoints exactly what an online algorithm requires in order to be optimal competitive in this environment. Inspired by our lower-bound result, we first focus on instances in which infinitely many online nodes arrive, with weights forming a \emph{continuum}. Through a simple general-purpose blackbox reduction (\Cref{sec:reduction-to-continuum}), we establish that restricting to this class is without loss of generality: any fractional algorithm for such instances can be transformed into a fractional algorithm for general instances without loss in performance. Given such instances, we identify natural invariant properties---namely \emph{``scale invariance''} (\Cref{def:SI}) and \emph{``greedy buyback''} (\Cref{def:GB})---which should be preserved by the allocation distribution of the optimal competitive online algorithm for the class of weight continuum instances. This investigation exactly motivates our particular way of updating the dual variable $\beta$: under this dual update rule, the two natural properties automatically hold---essentially for any choice of the penalty function---paving the way for our primal-dual analysis.\footnote{As it turns out, not only can we analyze the competitive ratio, but we can also \emph{exactly} characterize the allocation distribution of online algorithms that satisfy our invariant properties when executed on a weight continuum instance. Moreover, the resulting allocation distributions are parametric with respect to a single parameter, leading to an optimal competitive algorithm when appropriately tuned. See \Cref{remark:characterize-family} and \Cref{app:single-resource-characterization} for further details.}

To complement our results, in \Cref{sec:single-lossless-rounding}, we develop a correlated rounding technique that allows us to round this fractional online algorithm---and, more generally, \emph{any} fractional online algorithm---into an online randomized integral algorithm without any loss in expected performance. This second reduction results in a new primal-dual-based randomized integral online algorithm for the single-resource buyback problem.}

\smallskip
\paragraph{Primal-dual for matching environment: the \revcolorCMTwo{generalized} exponential penalty functions.} Turning to the matching environment and building on the insights from our previous investigation, we consider a natural extension of our primal-dual algorithm to the matching environment. This algorithm relies on the natural LP for maximum edge-weighted matching, \emph{ignoring} the buyback cost. Using this LP, we develop a direct primal-dual analysis for the competitive ratio of our algorithm. At a high level, similar to the single-resource setting, our algorithm maintains a dual assignment $\beta_j$ for each offline node $j$. It then (continuously and fractionally) allocates the arriving online node to the offline node with the maximum edge weight minus the dual assignment (following the complementary slackness conditions of the LP), and greedily discards the smallest edge weight assigned to this offline node by paying the buyback cost. The key remaining aspect of the algorithm is determining how to update the dual assignments and what information is required for this update.

\oldcolor{There are two key technical components in the way we update the dual variables. First, we use the allocation distributions of each offline node---particularly their quantile functions---in exactly the same form as used in the single-resource environment. Intuitively, we employ a penalty function to summarize the information in the allocation distribution of each offline node by proper integrations (see the exact formula for the dual assignment $\beta_j$ in our algorithms, \Cref{alg:primal dual matching} and \Cref{alg:opt deterministic matching}). This summarized information allows us to capture \revcolorCMTwo{the ``opportunity cost'' (the dual assignment) of} each offline node based on its complete allocation state, rather than relying solely on partial states such as the total allocated weight. Second, we crucially use a parametric family of \revcolorCMTwo{\emph{generalized} exponential penalty functions} of the form:
$$\pen(y)=\tau\cdot(\lambda^y-1)~,$$ 
where $\lambda,\tau\geq 0$ are parameters of the penalty function. Although our primal-dual analysis for the single-resource environment uses the same specific form of penalty function, it is primarily for simplicity of exposition. As discussed in \Cref{remark:characterize-family} and \Cref{app:single-resource-characterization}, such a specific form is, in fact, \emph{not} necessary for the single resource; rather, almost any penalty function would yield an optimal competitive algorithm. In contrast, for the general matching environment, our primal-dual algorithm \emph{critically} relies on this specific penalty function, and our analysis strongly indicates that it is essentially the unique choice for obtaining the optimal competitive ratio.

To gain intuitions, by carefully selecting the parameters $\tau$ and $\lambda$ based on the buyback factor $\buybackcost$, this family of functions allows us to \emph{indirectly} encode the historical buyback costs into the derived \revcolorCMTwo{opportunity costs of} allocating each offline node. While buyback decisions and costs do not explicitly appear in the standard linear programming formulation for edge-weighted bipartite matching, the mathematical properties of the generalized exponential penalty function enable their implicit role into the primal-dual analysis. We view this property as the critical technical link connecting the buyback problem to other online bipartite allocation frameworks. Our resulting family of primal-dual algorithms is simple, interpretable, and optimally competitive---achieving a competitive ratio of $\optCRgen$ with appropriately chosen parameters (details in \Cref{sec:matching}).

Motivated by practical considerations, we also explore several extensions to our matching model in \revcolorCMTwo{Remark~\ref{rmk:extensions-main} and Section~\ref{apx:extension}}, including general configuration allocations with combinatorial assignments, submodular welfare maximization capturing diminishing returns and substitution effects, and online assortment planning with dynamic product selections.}

\smallskip
\paragraph{Primal-dual for deterministic integral setting: simpler algorithms.} In case of integral allocations, our earlier primal-dual algorithm takes a much simpler form, as only the choice of $\pen(1)$ matters: this integral algorithm discounts each edge-weight by subtracting a fraction $\hat{\tau}(\buybackcost)\equiv \pen(1) $ of the buyback weight and then runs a greedy algorithm---unlike a naive greedy approach that matches each online node to the offline node with the maximum edge weight, which is unboundedly competitive.  By properly tuning $\hat{\tau}$ as a function of $\buybackcost$, the competitive ratio of this algorithm exactly matches the earlier lower bound $\optCRdet$ for deterministic integral algorithms (details in \Cref{sec:deterministic}).

\oldcolor{\paragraph{Negative buyback parameter $\buybackcost$.}  We also generalize our analysis to accommodate negative buyback factors $\buybackcost\in [-1,0]$, modeling scenarios with secondary supply channels or overflow capacities available at discounted rates. \revcolorCAMtwo{Equivalently, one can also think of scenarios where each offline node has two tiers of service: the first tier is capacitated and collects the full reward (edge weight) of each assignment, while the second tier has infinite capacity and collects only a $-\buybackcost$ fraction of the reward of each assignment.} See more details in \Cref{apx:negative buyback cost}. Interestingly, our \emph{exact} same competitive ratio results for the small buyback regime, for both fractional algorithms and deterministic integral algorithms, extend to this regime, which broadens the practical applicability of our theoretical framework.  See \Cref{fig:negative-f} for the comparison of the competitive ratios. Notably, these new curves start at $\frac{e}{e-1}$ (fractional) and $2$ (deterministic integral) when $f=0$, and both become $1$ as $\buybackcost$ goes to $-1$. Moreover, as we show in \Cref{apx:negative buyback cost}, these extended competitive ratio results are again tight.}
\begin{figure}[htb]
    \centering
   \input{figs/fig-cr-negative-f}
    \caption{Competitive ratio as a function of the negative buyback factor $\boldsymbol{\buybackcost}\in [-1,0]$: The {\color{blue} blue} solid curve represents the optimal competitive ratio $\boldsymbol{\frac{e}{e-(1+f)}}$ for the matching environment under fractional allocations. The {\color{red}red} solid curve represents the optimal competitive ratio $\boldsymbol{\frac{2}{1-f}}$ for deterministic integral algorithms in the matching environment. Interestingly, these curves are extensions of the similar curves in \Cref{fig:CRs}.
    } \label{fig:negative-f}
\end{figure}

\xhdr{Numerical simulations.} In addition to theoretical results, we also run simulations on synthetic data to numerically evaluate the performance of our online algorithms. See \Cref{apx:numerics} for more details. 

Our work is connected to various lines of literature in operations research and computer science. See \Cref{sec:further-related-work} for a comprehensive discussion of further related work.

\section{Problem Formulation}
\label{sec:prelim}
We study edge-weighted online matching (the \emph{matching environment}) and its special case with a single offline node (the \emph{single-resource environment}) within the buyback framework. In this section, we describe the model and some notation, and we formalize the type of performance guarantee considered in this paper.

\xhdr{General model.} An instance of our problem consists of a bipartite graph $\bipartitegraph = (\onlinenodes, \offlinenodes, \Edge)$, where $\offlinenodes$ is the set of resources (supply), indexed by $[\totalresource]\equiv\{1,\dots,\totalresource\}$, and $\onlinenodes$ is the set of requests (demand), indexed by $[\totaltime]\equiv\{1,\dots,\totaltime\}$. The edge set $\Edge \subseteq \onlinenodes \times \offlinenodes$ denotes allocation compatibility: an edge $(i,j) \in \Edge$ indicates that request~$i$ can be matched to resource~$j$. Each edge $(i,j)$ has a non-negative weight $\weightij$; for convenience, we set $\weightij = 0$ whenever $(i,j)\notin \Edge$, and thus assume without loss of generality that $\Edge=\onlinenodes\times\offlinenodes$. We impose no additional assumptions on the weights, as if they are chosen adversarially. Resources in $\offlinenodes$ are known upfront and hence referred to as the \emph{offline nodes}. In contrast, requests arrive sequentially over discrete times $1,\dots,\totaltime$, hence called the \emph{online nodes}. Upon arrival at time~$i$, online node~$i$ reveals weights $\weightij$ for all offline nodes~$j\in\offlinenodes$. Given these weights, the decision maker immediately chooses a (possibly fractional) allocation of node~$i$---with total demand equal to one unit---across offline nodes. Unlike typical online matching problems, however, these allocations are \emph{not} irrevocable, as we explain next.

Each offline node~$j$ has one unit of total capacity,\footnote{As shown later in \Cref{app:non-uniform demand}, this assumption is without loss of generality for fractional allocations, both in the matching environment and in the special case of the single-resource environment.} meaning that at any point during times $1,\dots,\totaltime$, no more than one unit of resource $j$ can be allocated in total across all online nodes $\onlinenodes$. Allocating an amount~$dx$ of offline node~$j \in \offlinenodes$ to online node~$i$ upon its arrival yields an immediate reward of~$\weightij dx$. To maintain capacity constraints, the decision maker may revoke (buy back) any fraction of previous allocations at a linear buyback cost. Specifically, buying back an allocation of amount~$dx$ from an earlier online node~$i'<i$ and offline node~$j$ not only forfeits the initial reward~$\weight_{i'j} dx$, but also incurs an additional cost of~$\buybackcost \cdot \weight_{i'j} dx$, where~$\buybackcost \geq 0$ is the given buyback factor. The objective is to design an online algorithm that dynamically determines fractional allocations and buybacks, maximizing the \emph{profit}: the net reward from retained allocations minus buyback costs.

\oldcolor{

\revcolorCMTwo{\xhdr{Offline vs.\ online cancellations.} In our base model above, buyback decisions occur simultaneously with allocation decisions in a fully online manner. This model is exactly equivalent to an alternative model with ``overbooking,'' in which allocations occur without immediate capacity constraints and buybacks are deferred until the end of the horizon (after time $T$). Although the alternative model appears more flexible, the linear structure of the buyback costs makes the two models equivalent. Under our buyback model, the total penalty for over-allocation is proportional to the canceled weights, and therefore, the globally optimal assignment for each offline node in the deferred model reduces to simply retaining the highest-weight allocations up to the capacity limit. Since it is also optimal for any online algorithm to greedily revoke the smallest edge weights immediately upon any capacity violation, it perfectly replicates this final retained subset in the deferred model, resulting in the exact same final assignments and net profit.}

\xhdr{Fractional vs.\ integral allocations.} We primarily focus on fractional allocations and buybacks throughout the paper. \revcolorCMTwo{Fractional matchings serve as a natural and operationally meaningful model for large-capacity regimes---such as display advertising (see \Cref{app:practical}), where large impression volumes make allocations effectively divisible.} For the single-resource environment (\Cref{sec:single-resource}), our results can be extended without loss to randomized integral allocations, as we show later. In the matching environment (\Cref{sec:matching}), our fractional algorithms can be converted into randomized integral algorithms in the asymptotic regime with large capacities, incurring minimal performance loss (see the randomized rounding details in \Cref{apx:matching rounding}). Additionally, we consider the setting with integral and deterministic allocations and buybacks (\Cref{sec:deterministic}), providing guarantees valid even for unit capacities, which are relevant in certain practical applications.}

\xhdr{Performance \& benchmark.} To evaluate the performance of our online algorithms, we benchmark them against an optimal omniscient algorithm, termed the \emph{\revcolorCMTwo{offline optimum}}, which knows the entire problem instance in advance. Since the \revcolorCMTwo{offline optimum} has complete foresight, it incurs no buyback costs—it directly selects the maximum edge-weighted matching (or simply the highest-weight online node in the single-resource scenario). For a fixed buyback factor $\buybackcost \geq 0$, we measure our algorithms' performance via its \emph{competitive ratio}, defined as the worst-case ratio of the offline optimum's profit to that of the online algorithm.

\begin{definition}[Competitive Ratio]
\label{def:comp-ratio}
Given a buyback factor $\buybackcost$, an online algorithm $\ALG$ is $\approxratio(\buybackcost)$-competitive with respect to the \revcolorCMTwo{offline optimum} benchmark within a class of problem instances $\instances$ if
\[
    \sup_{\instance \in \instances}~ \frac{\OPT(\instance)}{\ALG(\instance)} \leq \approxratio(\buybackcost)~,
\]
where $\ALG(\instance)$ and $\OPT(\instance)$ denote the expected total profit of $\ALG$ and the \revcolorCMTwo{offline optimum}, respectively.
\end{definition}

\xhdr{Allocation distribution: density \& quantile functions.} A crucial construct in our algorithm design and analysis is the tracking of an \emph{allocation distribution} for each offline node $j \in \offlinenodes$. This distribution fully captures the ``current state'' of the node, reflecting the history of its fractional allocations to previously arrived online nodes. Specifically, consider fractional allocations in the matching environment. At any point during the execution of the algorithm, each offline node $j$ has past fractional allocations made at various weights. Since each offline node has unit capacity, these allocations define a histogram over weights---like a probability distribution---representing the fraction of the node's capacity currently allocated at each weight.\footnote{To maintain a total mass of $1$ at all times, each node's allocation distribution is initially set as a point mass at weight zero.} Formally, for each offline node $j \in \offlinenodes$, we define an \emph{allocation density function} $\allocj:\reals_+ \rightarrow \reals_+$, where $\allocj(\weight)$ is the density (fraction) of node $j$ currently allocated at weight $\weight$ (excluding revoked allocations). Additionally, we define the \emph{allocation quantile function} $\callocj:\reals_+ \rightarrow[0, 1]$, where
\[
\callocj(\weight) = \int_\weight^\infty \allocj(t)\,\dd t
\]
is the fraction of node $j$ allocated at weights at least $\weight$ (again, excluding revoked allocations).

\section{Warm-Up: Primal-Dual Algorithms for Single-Resource Environment}
\label{sec:single-resource}
\oldcolor{We begin by considering the single-resource environment studied in \cite{AK-09, BHK-09}, under both fractional and integral allocations. We introduce a new systematic approach based on the \revcolorCMTwo{primal-dual} framework to solve this special case, resulting in a simple (and interpretable) optimal competitive online algorithm. \revcolorCMTwo{This section serves as a warm-up for the more general matching environment: the analysis of the main algorithm of this section---which critically relies on a generalized exponential penalty function---closely parallels, and lays the groundwork for, the design and analysis of our main algorithm in \Cref{sec:matching}. This perspective, and the resulting algorithms, will be the key to our main technical results in Sections~\ref{sec:matching} and \ref{sec:deterministic}.} Throughout this section, we drop subscript~$j$ for the single offline node in our notation; for example, we replace weight $\weightij$ with $\weighti$, and functions $\allocj(\cdot)$ and $\callocj(\cdot)$ with $\alloc(\cdot)$ and $\calloc(\cdot)$.

Although the single-resource problem appears to be a relatively simple special case, obtaining an optimal (or even constant) competitive online algorithm in this setting is highly nontrivial. Consider, for example, a greedy algorithm that fractionally and myopically allocates the offline resource to an arriving online node of weight~$\weight$ by canceling the previous allocation of weight $\weight'$ (if necessary), provided the net reward $\weight - (1+\buybackcost)\weight'$ is nonnegative. A simple instance shows that this algorithm has an unbounded competitive ratio:  suppose online nodes have weights $\weighti = (1+\buybackcost+\varepsilon)^i,~i=1,\ldots,\totaltime$ for an arbitrary small $\varepsilon>0$. The optimal offline strategy allocates only to the last node, collecting a reward of $(1+\buybackcost+\varepsilon)^\totaltime$. By contrast, the greedy algorithm allocates to every single online node (and buys back the previous one), obtaining a net reward of $1+\buybackcost+\mathcal{O}(\varepsilon)$. As $\totaltime\to \infty$, the ratio between the two rewards becomes unbounded.

This example underscores that an online algorithm must carefully trade off immediate reward against future (potential) buyback costs, given the uncertainty in future arriving weights. Indeed, as the following lower bound (adapted from \citealp{AK-09}, formally proved in \Cref{sec:apx-missing lower bound single resource}) indicates, this trade-off imposes an inherent limit on the competitive ratio, even under fractional allocations.
\begin{restatable}{proposition}{lowerboundsingleresource}
\label{prop:lower bound single resource}
In the single-resource environment,
for any buyback factor $\buybackcost\geq 0$,
no fractional online algorithm can obtain a competitive ratio better than
$-\Lambertterm$.
\end{restatable}
\subsection{Overview of our approach}
\label{sec:single-resource-overview}
As alluded to earlier, to design a competitive online algorithm that optimally balances this trade-off and matches the above lower bound (under both fractional and integral allocations), we adopt a systematic approach based on the primal-dual framework.} \revcolorCMTwo{Our approach consists of the following three main steps:}

\begin{enumerate}[leftmargin=14pt,label=(\roman*)]
\smallskip
    \item \revcolorCMTwo{\textit{Offline LP (\Cref{sec:single-resource-LP}):}
    Focusing first on fractional allocations, we introduce a simple (almost trivial) linear program, together with its dual, that characterizes the offline optimum of the single-resource problem.}

\smallskip
    \item \revcolorCMTwo{\textit{A parametric family of primal-dual fractional algorithms (\Cref{sec:single resource primal dual}):}
    Using this primal-dual LP pair, along with a carefully constructed dual assignment, we propose a parametric family of primal-dual algorithms---parameterized by a penalty function---that runs directly on general instances. Then, through a novel dual fitting analysis, we show that an appropriate choice of the penalty function yields an optimal competitive fractional online algorithm.}

\smallskip
    \item \revcolorCMTwo{\textit{Integral-to-fractional reduction (\Cref{sec:single-lossless-rounding}):}
    Lastly, by designing a randomized rounding procedure correlated across time, we show that any fractional online algorithm can be converted into an integral algorithm without loss in expectation. This yields an optimal competitive integral primal-dual algorithm.}
\end{enumerate}

\subsection{Designing \& analyzing an optimal competitive primal-dual algorithm}
\label{sec:design and analyze optimal algorithm single-resource}
We now elaborate \revcolorCMTwo{on} our primal-dual approach for designing and analyzing algorithms in the single-resource environment, initially focusing on fractional allocations and then extending our algorithm to integral allocations. Our investigation will lead to the following main result of this section.

\begin{informal}[Theorem (informal).]
In the single-resource environment, for any buyback factor $f\geq 0$, there exist primal-dual-based online algorithms---under both fractional and randomized integral allocations---that achieve the optimal competitive ratio of $-\Lambertterm$ within the class of all instances $\instances$.
\end{informal}

\subsubsection{\revcolorCMTwo{A simple linear program for the offline optimum}}
\label{sec:single-resource-LP}
\revcolorCMTwo{Consider a general instance of the single-resource environment with arriving weights $\{\weighti\}_{i\in \onlinenodes}$, where $\onlinenodes$ is the set of online nodes. Since the \revcolorCMTwo{offline optimum} benchmark incurs no buyback costs and simply allocates the entire (unit) capacity of the resource to the online node with the highest weight, it can be characterized by the following (almost trivial) linear program. We refer to this formulation as the primal linear program and also consider its dual linear program:
\begin{align*}
\tag{$\mathcal{P}_{\texttt{OPT-single-gen}}$}
\label{eq:LP-max-weight-WARMUP}
\arraycolsep=1.4pt\def\arraystretch{1}
\begin{array}{lllllll}
\max~~~~&\displaystyle\sum_{i\in\onlinenodes}
\edgealloci \weighti &~~\text{s.t.}&
& \quad\quad\quad\quad \text{min}~~~~&\offlinedual &~~~~~\text{s.t.} \\[1.4em]
 &\displaystyle\sum_{i\in \onlinenodes}{\edgealloci}\leq 1~, &\qquad\;\;\, &
& &\offlinedual\geq \weighti &~~~~~~i\in \onlinenodes, \\
 &\edgealloci \geq 0~~~&i\in\onlinenodes~.
 \qquad \qquad &
& &\offlinedual\geq 0~.  &
\end{array}
\end{align*}
Observe that the dual program has a single variable $\offlinedual$, and any feasible dual solution---that is, any $\offlinedual\geq \max_{i\in\onlinenodes}\weighti$---provides an upper bound on the \revcolorCMTwo{offline optimum} value by LP weak duality. Notably, neither buyback decisions nor their costs explicitly appear in this linear program; as we show next, the buyback cost will instead be incorporated \emph{indirectly} in our analysis, through the way our online algorithm maintains a feasible dual assignment during its execution.}

\subsubsection{\revcolorCMTwo{A family of optimal competitive fractional primal-dual algorithms}}
\label{sec:single resource primal dual}
\revcolorCMTwo{Inspired by primal-dual algorithms developed for various variants of online bipartite matching~\citep[cf.][]{meh-13}, we propose a parametric family of primal-dual algorithms based on the above linear program. Our primal-dual algorithm maintains a feasible dual assignment $\offlinedual$ for \ref{eq:LP-max-weight-WARMUP} throughout its execution. The algorithm is parameterized by a monotone increasing penalty function $\pen:[0,1]\rightarrow\reals_+$, satisfying $\pen(0)=0$ and $\pen(1)\geq 1$. Given this penalty function, the dual assignment is directly related to the current allocation quantile function $\calloc(\cdot)$ of the offline resource (defined in \Cref{sec:prelim}) as follows:\footnote{As we elaborate later in Remarks~\ref{remark:alternative-approach} and \ref{remark:characterize-family}, the allocation distribution of any candidate optimal competitive online algorithm in this environment should arguably satisfy two natural structural properties. It turns out that under the dual assignment \eqref{eq:dual assignment}, these properties automatically hold---which is the main technical reason behind this particular design choice.}
\begin{align}
\label{eq:dual assignment}
\tag{\textsc{Dual-single}}
    \offlinedual = \displaystyle\int_0^\infty
    \pen\left(\calloc(\weight)\right)\,\dd\weight~.
\end{align}
To maintain dual feasibility, when a new online node $i\in\onlinenodes$ with weight $\weighti$ arrives, the algorithm continuously allocates to this node---greedily buying back from the smallest allocated weights to free up capacity---until the dual matches the arriving weight, i.e., $\offlinedual = \weighti$. This procedure always terminates because $\pen(1) \geq 1$ implies $\int_{0}^{\weighti}\pen(1)\,\dd\weight \geq \weighti$,
ensuring that the stopping condition $\offlinedual = \weighti$ is eventually met. Since the dual variable $\offlinedual$ only increases thereafter, the feasibility remains guaranteed. This process is formalized in \Cref{alg:primal-dual-direct-single}.}

\revcolorCMTwo{\begin{algorithm}[htb]
\caption{\revcolorCMTwo{Primal-dual fractional online algorithm for general instances (single-resource)}}
\label{alg:primal-dual-direct-single}
    \SetKwInOut{Input}{input}
    \SetKwInOut{Output}{output}
 \Input{Penalty function $\pen$
 }
 
 \vspace{2mm}
 
 Initialize $\offlinedual \gets 0$, and for all $\weight \in \reals_+$, set $
    \alloc(\weight) \gets \diracdeltafunction_0(\weight)$ and $\calloc(\weight) \gets \indicator{\weight =0}$.
 
 \vspace{1mm}
 
 {\color{royalazure}\tcc{$\diracdeltafunction_{\weight'}(\cdot)$
 is the Dirac delta function centered at $\weight'$.}}
 
 \vspace{1mm}
 
 \For{each online node $i\in U$}
 {
 
 \vspace{1mm}
 
 \While{$\offlinedual < \weight_i $}
 {
 \vspace{1mm}
 Buy back infinitesimal fraction $dx$ from the smallest allocated weight $\weight_{i'}=\buybackweight$ for some $i'<i$, i.e., $\alloc(\weight) \gets
 \alloc(\weight) - dx \cdot\delta_{\weight_{i'}}(\weight)$~~{\color{royalazure}\tcc{Formally, $\buybackweight = \inf\{\weight\in\reals_+:
 \calloc(\weight) < 1\}$}}
 
 \vspace{1mm}
 Allocate fraction $dx$ to the current online node $i$, i.e., $\alloc(\weight) \gets \alloc(\weight) + dx \cdot\delta_{\weight_i}(\weight)$
 
 \vspace{1mm}
 Update the allocation quantile function $\forall \weight\in[\weight_{i'},\weight_i]:\displaystyle\calloc(\weight) \gets \int_{\weight}^{\infty}
\alloc(\weight')\,d\weight'$

 \vspace{1mm}
Update the dual assignment $\displaystyle\offlinedual\gets 
\int_{0}^{\infty}
\pen(\calloc(\weight))\,d\weight$}
 }
\end{algorithm}}

\revcolorCMTwo{Two aspects of the design of \Cref{alg:primal-dual-direct-single} are worth highlighting. First, whenever the algorithm buys back to free up capacity, it always revokes allocations made at the \emph{smallest} allocated weight. This buyback rule follows from a simple observation: whenever an online algorithm must buy back to free capacity, it is always optimal---given any future sequence of online arrivals---to greedily buy back allocation from previously allocated online nodes with the smallest weight. Second, through the dual assignment \eqref{eq:dual assignment}, the algorithm employs the penalty function to summarize the information in the \emph{entire} allocation distribution of the offline resource by proper integration of its penalized quantile function. Intuitively, the dual assignment $\offlinedual$ captures the ``opportunity cost'' of the resource based on its complete allocation state---rather than relying solely on partial states such as the total allocated weight---and the algorithm keeps allocating to an arriving online node only as long as the node's weight exceeds this opportunity cost.}

\revcolorCMTwo{To obtain the main result of this section, we consider a specific choice of the penalty function that simplifies the primal-dual analysis. This particular choice, termed as the generalized exponential function, also plays a critical role in our main result for the matching environment in \Cref{sec:matching} (see \Cref{penalty-definition}).}

\begin{proposition}
\label{prop:optimal-direct-single-resource-fractional}
\revcolorCMTwo{In the single-resource environment, for any buyback factor $f\geq 0$, consider \Cref{alg:primal-dual-direct-single} with the penalty function $\pen(\calloc)=
\frac{1}{\log(\thresholdweight)}(\thresholdweight^\calloc - 1)$, where $\thresholdweight = -(1 + \buybackcost) \Lambertterm$. This algorithm achieves the competitive ratio of $-\Lambertterm$ within the class of all instances $\instances$.}
\end{proposition}
\begin{proof}{\emph{Proof.}}
\revcolorCMTwo{To analyze the competitive ratio, we employ `dual fitting' via the primal and dual LP pair for the \revcolorCMTwo{offline optimum}, formally described as \ref{eq:LP-max-weight-WARMUP}. Specifically, we construct a feasible dual solution whose objective value is no more than $-\Lambertterm$ times the primal objective (i.e., the total profit) achieved by the algorithm. This, combined with LP weak duality, completes the proof.

As suggested in \Cref{alg:primal-dual-direct-single}, we use the following dual assignment with $\pen(\calloc)=
\frac{1}{\log(\thresholdweight)}(\thresholdweight^\calloc - 1)$:
$$\hat{\beta} = \int_{0}^{\infty} \pen(\calloc(w))dw~,$$
where $\calloc(w)$ is the final allocation quantile function of \Cref{alg:primal-dual-direct-single}, after all online nodes arrive.

First, we show  feasibility of this dual assignment. For our choice of the penalty function $\pen$, we have
$$
 \pen(1) = \frac{\thresholdweight -1}{\log( \thresholdweight)} \overset{(a)}{\geq} \frac{\thresholdweight}{\log(\thresholdweight) +1} \overset{(b)}{=} 1+f \geq 1,
$$
where inequality~(a) holds as for positive values of $\thresholdweight$ we have $\thresholdweight-1\geq \log(\thresholdweight)$, and equality~(b) holds as $\thresholdweight = -(1 + \buybackcost) \Lambertterm$.
Hence, at the end of allocating online node $i$, $\beta = \weighti$ for the dual  $\beta$ maintained by the algorithm (as otherwise $\beta<\weighti$ after fully allocating online node $i$, but we know $\int_{0}^{\weighti}\pen(1)\,d\weight \geq \weighti$, a contradiction). After this point, $\beta$ can only increase; thus, for any $i\in U$, we have $\hat{\beta} \geq \weight_i$, as desired.

Next, we compare the change in the primal objective, denoted by $\Delta(\text{Primal})$, and the dual objective, denoted by $\Delta(\text{Dual})$, as the algorithm allocates an infinitesimal fraction $dx$ of online node $i$ and buys back the same fraction from online node $i'<i$. We have:
\begin{align*}
    \Delta(\text{Primal}) &=
    \left(\weighti - (1+\buybackcost)\weightip\right)dx~.
\end{align*}
Furthermore, let $\calloc(\weight)$ be the quantile allocation function of \Cref{alg:primal-dual-direct-single} before processing this infinitesimal allocation of online node $i$. Then the change in the dual objective is upper-bounded as follows:
\begin{align*}
    \Delta(\text{Dual}) &= \int_{\weightip}^{\weighti}
    \penderivative(\calloc(\weight))\,d\weight\overset{(c)}{\leq}
    \left(
    \log(\thresholdweight)
    \left(
    \weighti -
    \displaystyle\int_0^\infty \pen(\calloc(\weight))\,d\weight
    \right)
    +
    \int_{\weightip}^{\weighti}
    \penderivative(\calloc(\weight))\,d\weight
    \right)dx
    \\
    &\overset{(d)}{\leq}
    \left(
    \log(\thresholdweight)
    \left(
    \weighti -
    \displaystyle\int_0^{\weightip} \pen(1)\,d\weight
    -
    \displaystyle\int_{\weightip}^{\weighti} \pen(\calloc(\weight))\,d\weight
    \right)
    +
    \int_{\weightip}^{\weighti}
    \penderivative(\calloc(\weight))\,d\weight
    \right)dx
    \\
    &=
    \left(
    \log(\thresholdweight)
    \weighti -
    \left(\thresholdweight - 1\right)\weightip
    -\displaystyle\int_{\weightip}^{\weighti}
    \left(
    \thresholdweight^{\calloc(\weight)} - 1
    \right)\,d\weight
    +
    \int_{\weightip}^{\weighti}
    \thresholdweight^{\calloc(\weight)}
    \,d\weight
    \right)dx
    \\
    &=
    \left((\log(\thresholdweight) + 1)\weighti
    -\thresholdweight
    \cdot\weightip \right)dx \overset{(e)}{=} \frac{\thresholdweight}{1+f}\left(\weighti -(1+f)\weight_{i'}\right)dx =-\Lambertterm \Delta(\text{Primal})~,
\end{align*}
where inequality~(c) holds as $\int_{0}^{\infty} \pen(\calloc(w))dw=\beta\leq \weighti$ during the allocation of online node $i$ ($\beta$ is the dual maintained by the algorithm, and not $\hat{\beta}$), and inequality~(d) holds by dropping $\int_{\weighti}^{\infty}\pen(\calloc(\weight))\,\dd\weight$ and the fact that $\calloc(\weight) = 1$ for every $\weight\leq\weightip$ (since \Cref{alg:primal-dual-direct-single} greedily buys back from the smallest allocated weight, the entire capacity is allocated at weights at least $\weightip$ at the moment of this buyback). Lastly, equality~(e) holds as $\thresholdweight = -(1 + \buybackcost) \Lambertterm$. Summing over all allocations, we conclude that the dual objective is at most $-\Lambertterm$ times the primal objective, which finishes the proof.\hfill\halmos}
\end{proof}

\revcolorCMTwo{Combined with the lower bound in \Cref{prop:lower bound single resource}, \Cref{prop:optimal-direct-single-resource-fractional} shows that \Cref{alg:primal-dual-direct-single} (with the above generalized exponential penalty function) is an optimal competitive fractional online algorithm for the single-resource environment within general instances $\instances$.\footnote{We also consider a straightforward generalization of this result to settings with non-unit demand sizes and arbitrary offline resource capacities. We show through a simple reduction that the optimal competitive ratio remains achievable. See \Cref{app:non-uniform demand} for details.} Moreover, similar to our main algorithm for the matching environment, \Cref{alg:primal-dual-direct-single} can be implemented in polynomial time; see \Cref{apx:implementation} for details. Before proceeding, we highlight two important remarks:}

\begin{remark}[Alternative primal-dual approach via weight continuum instances]
\label{remark:alternative-approach}
\revcolorCMTwo{There is an alternative approach for identifying a family of optimal competitive primal-dual fractional algorithms in the single-resource environment. In this approach, we first restrict attention to a special class of instances, termed \emph{truncated weight continuum} instances, in which a spectrum of weights (infinitely many online nodes) arrives continuously until an unknown stopping point; we show that any online algorithm designed for this class can be transformed into a fractional online algorithm for general instances with the same competitive ratio. We then identify natural structural properties---scale invariance and greedy buyback---that the allocation distribution of any ``good'' online algorithm should satisfy within this class, and introduce a family of primal-dual algorithms for weight continuum instances (the continuum analogue of \Cref{alg:primal-dual-direct-single}) that provably satisfies these properties and, with appropriate choices of the penalty function, is optimal competitive. We present this alternative approach in detail in \Cref{sec: apxsec3}.}
\end{remark}

\begin{remark}[Characterization of allocation distributions]
\label{remark:characterize-family}
\revcolorCMTwo{Our primal-dual analysis (in particular, the proof of \Cref{prop:optimal-direct-single-resource-fractional}) relies on a specific choice of the penalty function $\pen$, which we also employ in the more general matching environment in \Cref{sec:matching}. In \Cref{app:single-resource-characterization}, following the alternative approach of \Cref{remark:alternative-approach}, we present a refined analysis specifically tailored to the single-resource environment, showing that our primal-dual algorithm (executed on weight continuum instances) remains optimally competitive for \emph{any} penalty function $\pen$ satisfying:
\begin{align*}
 \thresholdweight =
    \pen(1) +
    \displaystyle\int_1^{\thresholdweight}
    \pen\left(
    1 - \frac{\log(\weight)}{\log(\thresholdweight)}
    \right)d\weight~,
\end{align*}
with $\thresholdweight = -(1 + \buybackcost) \Lambertterm$. 
Moreover, we provide an \emph{exact characterization} of allocation distributions induced by algorithms satisfying the scale invariance and greedy buyback properties when executed on truncated weight continuum instances, and precisely determine their competitive ratios. This characterization identifies a single-parameter class of allocation distributions (parametric in $\thresholdweight$, see \Cref{fig:canonical} in \Cref{app:single-resource-characterization}), which includes the optimal competitive algorithm by setting $\thresholdweight = -(1 + \buybackcost) \Lambertterm$. For details and an alternative optimality proof based on this characterization, see \Cref{app:single-resource-characterization}.}
\end{remark}

\revcolorCMTwo{In the next subsection, we introduce a reduction to transform our fractional algorithm (in a black-box manner) into a randomized integral algorithm, without loss in expected performance. Thus, our result extends to an optimally competitive (randomized) integral online algorithm.}

\oldcolor{
\subsubsection{Reduction from integral to fractional allocations}
\label{sec:single-lossless-rounding} We now briefly describe how to convert any fractional online algorithm for general instances into a randomized integral online algorithm, preserving its expected performance (on every instance) and thus maintaining its competitive ratio.

This integral-to-fractional reduction relies on a simple, correlated, lossless randomized rounding algorithm. This algorithm first samples a random seed $\eta\sim\mathcal{U}[0,1]$ at the outset. Then, given a sequence of fractional allocations $\tilde{\alloc}_1,\dots,\tilde{x}_T$ produced (in an online fashion) by the fractional online algorithm, it allocates randomly and integrally to each online node $i$ upon its arrival in a correlated manner such that $\prob{\text{allocation at time } i} = \tilde{\alloc}_i$, and the probability of buying back the allocation made at time $i$ matches exactly the fraction of node $i$ bought back by the fractional online algorithm. See \Cref{alg:algorithm-randomized-reduction-single-resource} for formal details. \revcolorCMTwo{This leads directly to the following proposition, whose proof we also present below.}

\begin{algorithm}[hbt]
\caption{\oldcolor{Reduction from integral to fractional algorithms for general instances (single-resource)}}
\label{alg:algorithm-randomized-reduction-single-resource}
\SetKwInOut{Input}{Input}
\SetKwInOut{Output}{Output}
\Input{Fractional online algorithm $\ALG$ for general instances}

\vspace{1mm}

Sample $\randseed \sim \mathcal{U}[0,1]$.

\vspace{1mm}

\For{each online node $i \in \onlinenodes$}{

    \vspace{1mm}

    Let $\allocTilde_i$ be the fractional amount allocated to online node $i$ by $\ALG$.

    \vspace{1mm}

    Allocate integrally to online node $i$ (buying back if necessary) if there exists an integer $\ell \in \naturals$ satisfying:
    $
    \displaystyle\sum_{t=1}^{i-1}\allocTilde(t) \leq \ell + \randseed - 1 < \sum_{t=1}^{i}\allocTilde(t).
    $
    }
\end{algorithm}

\vspace{-3mm}
\begin{restatable}{proposition}{thmintegralreduction}
\label{prop:integral-reduction-single-resource}
In the single-resource environment,
for any fractional online algorithm $\ALG$ with the competitive ratio $\approxratio$ within all instances $\instances$,
\Cref{alg:algorithm-randomized-reduction-single-resource} (with $\ALG$ as input)
is a feasible randomized integral online algorithm and has the competitive ratio $\approxratio$ within all instances $\instances$.
\end{restatable}

\revcolorCMTwo{\begin{proof}{\emph{Proof.}}
Fix an arbitrary fractional online algorithm $\ALG$ and an instance $\instance$, and let $\allocTilde(i)$ denote the fraction allocated to each online node $i\in\onlinenodes$ upon its arrival by $\ALG$. Based on these fractions, we construct a randomized integral online algorithm $\ALG\primed$ (formalized in \Cref{alg:algorithm-randomized-reduction-single-resource}) whose allocation and buyback decisions exactly mimic those of $\ALG$ in expectation. The construction is as follows. The algorithm $\ALG\primed$ first samples a random seed $\randseed$ uniformly at random from $[0,1]$. Given this seed, define the (random) function $t_\randseed: \naturals\rightarrow \naturals\cup\{+\infty\}$ where, for every $\ell\in\naturals$, $t_\randseed(\ell)$ is the unique index $t$ satisfying
\begin{align*}
    \sum_{i=1}^{t-1}\allocTilde(i) \leq \ell + \randseed-1 < \sum_{i=1}^{t}\allocTilde(i)~;
\end{align*}
if no such index exists, we set $t_\randseed(\ell)=+\infty$. Note that such an index is unique whenever it exists, since the partial sums $\sum_{i=1}^{t}\allocTilde(i)$ are non-decreasing in $t$. The algorithm $\ALG\primed$ then allocates the offline node, integrally, exactly to the online nodes in $\{t_\randseed(\ell):\ell\in\naturals\}$; moreover, when allocating to online node $t_\randseed(\ell)$, it cancels its previous allocation made to online node $t_\randseed(\ell-1)$ (if such an allocation exists). By construction, $\ALG\primed$ is a feasible randomized integral online algorithm: it holds at most one unit of allocation at any point in time, and its decision upon each arrival only depends on the fractions $\{\allocTilde(i)\}$ revealed by $\ALG$ so far (and the seed $\randseed$).

We first verify that the (marginal) allocation probabilities of $\ALG\primed$ match the fractional allocations of $\ALG$. By construction, $\ALG\primed$ allocates to an online node $t$ if and only if there exists an integer $\ell \in \naturals$ such that $\sum_{i=1}^{t-1}\allocTilde(i) \leq \ell + \randseed-1 < \sum_{i=1}^{t}\allocTilde(i)$. Since $\randseed\sim \text{Unif}[0,1]$, this event happens with probability exactly
\begin{align*}
    \sum_{i=1}^{t}\allocTilde(i) - \sum_{i=1}^{t-1}\allocTilde(i) =  \allocTilde(t)~.
\end{align*}

It remains to show that the buyback probabilities also match; that is, for each online node $t$, the probability that $\ALG\primed$ buys back its allocation to online node $t$ is equal to the fraction of online node $t$ that is eventually bought back by $\ALG$. We distinguish three cases:
\begin{enumerate}[label=(\roman*)]
    \item If $\sum_{i=t}^{T}\allocTilde(i) < 1$, then neither $\ALG$ nor $\ALG\primed$ ever buys back online node $t$.
    \item If $\sum_{i=t+1}^{T}\allocTilde(i) > 1$, then online node $t$ is eventually fully bought back under both $\ALG$ and $\ALG\primed$ (whenever it is allocated).
    \item If $\sum_{i=t}^{T}\allocTilde(i) \geq 1 \geq \sum_{i=t+1}^{T}\allocTilde(i)$, then $\ALG$ buys back a $\sum_{i=t}^{T}\allocTilde(i) - 1$ fraction of online node $t$. As for $\ALG\primed$, conditioned on allocating to online node $t$, it buys back this allocation with probability $\frac{\sum_{i=t}^{T}\allocTilde(i) - 1}{\allocTilde(t)}$; hence, the unconditional probability that $\ALG\primed$ buys back online node $t$ is exactly $\sum_{i=t}^{T}\allocTilde(i) - 1$, matching the fraction bought back by $\ALG$.
\end{enumerate}

In summary, for every online node $t$, the expected allocated amount and the expected bought-back amount under $\ALG\primed$ coincide with the corresponding fractional amounts under $\ALG$. By linearity of expectation, the expected profit of $\ALG\primed$ is equal to the profit of $\ALG$ on every instance $\instance$. Thus, the competitive ratio of the constructed randomized integral online algorithm $\ALG\primed$
is the same as that of algorithm $\ALG$, as desired.
\hfill\halmos
\end{proof}
}

\revcolorCMTwo{By combining \Cref{prop:optimal-direct-single-resource-fractional} and \Cref{prop:integral-reduction-single-resource}, and using the composition of Algorithms~\ref{alg:primal-dual-direct-single} and \ref{alg:algorithm-randomized-reduction-single-resource}, we achieve an optimal competitive primal-dual-based randomized integral online algorithm for the single-resource environment.} We conclude this section with an intriguing related observation.

\begin{remark}[Comparison to \citealp{AK-09}]
\label{remark:primal-dual-BK}
One might wonder how our algorithm relates to that of \cite{AK-09}. In \Cref{app:BK implementation}, we address this question by slightly adjusting our framework and applying a slightly different lossless online rounding scheme directly to our optimal competitive fractional algorithm from \Cref{prop:algorithm reduction single resource} for instances in $\continstances$, and simplifying the resulting algorithm. Remarkably, we show that the derived primal-dual-based algorithm becomes \emph{exactly} identical to the direct randomized algorithm of \cite{AK-09}. Thus, somewhat surprisingly, their optimal competitive algorithm is essentially a primal-dual algorithm \emph{in disguise}.
\end{remark}

}



\section{Primal-Dual Algorithms for Matching Environment \revcolor{(and Beyond)}}
\label{sec:matching}
In this section, we turn to the more general matching environment and present our main result: an optimally competitive online algorithm for edge-weighted online matching with costly cancellations for all factors $\buybackcost\geq 0$. Our algorithm (\Cref{alg:primal dual matching}) is a natural generalization of the primal-dual algorithm introduced earlier for the single-resource environment (\revcolorCMTwo{\Cref{alg:primal-dual-direct-single}}, analyzed in \Cref{sec:single-resource}).

\subsection{Algorithm description}
Inspired by the primal-dual algorithms developed for the single-resource environment, \Cref{alg:primal dual matching} maintains a dual assignment \oldcolor{$\offlinedualj = \int_0^\infty \pen(\calloc_j(\weight))\,\dd\weight$} for each offline node $j$. When a new online node $i$ arrives, the algorithm continuously and fractionally allocates this online node (with a unit total demand) to the offline node $j^*$ maximizing $\weight_{ij^*} - \offlinedual_{j^*}$, while breaking ties arbitrarily. During this process, the algorithm updates only the dual assignment $\offlinedual_{j^*}$. Note that the choice of $j^*$ may shift during this process as the dual variables evolve. Now, this allocation process continues until either node $i$'s demand is fully allocated or the value $\weightij - \offlinedualj$ becomes negative for every offline node $j \in \offlinenodes$. Also, as before, the algorithm greedily buys back capacity from offline node $j^*$ (if required), always revoking fractions allocated at the smallest weights first.\footnote{As a sanity check, when there is only one offline node, \Cref{alg:primal dual matching} reduces to \revcolorCMTwo{\Cref{alg:primal-dual-direct-single}}.}

\begin{algorithm}
\caption{Primal-dual fractional online algorithm (matching)}
\label{alg:primal dual matching}
    \SetKwInOut{Input}{input}
    \SetKwInOut{Output}{output}
 \Input{Penalty function $\pen$
 }
 
 \vspace{2mm}
 
 Initialize $\offlinedual_j \gets 0$, and for all $\weight \in \mathbb{R}_+$ and $j\in\offlinenodes$, set $\allocj(\weight) \gets \diracdeltafunction_0(\weight), \callocj(\weight) \gets \indicator{\weight = 0}$.
 
 
 
 \vspace{1mm}
 
 {\color{royalazure}\tcc{$\diracdeltafunction_{\weight'}(\cdot)$
 is the Dirac delta function centered at $\weight'$.}}
 
 \vspace{1mm}
 

 
 \For{each online node $i\in\onlinenodes$}
 {
 
 \vspace{1mm}
 
 \While{capacity of online node $i$ is not exhausted and
 there exists $j\in\offlinenodes$ s.t.\ $\offlinedualj < \weightij$}
 {
 \vspace{2mm}
 
 Let $j^* \gets \argmax_{j\in\offlinenodes}\ 
 \weightij - \offlinedualj$ 
 
 Buy back fraction $dx$ of offline node $j^*$ from 
 the smallest allocated weight $\buybackweight_{j^*}$ to $j^*$,
 i.e., $\alloc_{j^*}(\weight) \gets
 \alloc_{j^*}(\weight) - dx \cdot\delta_{\buybackweight_{j^*}}(\weight)$~~{\color{royalazure}\tcc{Formally, $\buybackweight_{j} = \inf\{\weight\in\reals_+:
 \calloc_{j}(\weight) < 1\}$}}

 \vspace{1mm}
 Allocate fraction $dx$ of offline node $j^*$
 to the current online node $i$, i.e., $\alloc_{j^*}(\weight) \gets
 \alloc_{j^*}(\weight) + dx \cdot\delta_{\weight_{ij^*}}(\weight)$

 \vspace{1mm}
 Update the allocation quantile function $\forall \weight \in [\buybackweight_{j^*},\weight_{ij^*}]:\calloc_{j^*}(\weight) \gets \int_\weight^{\infty}
\alloc_{j^*}(t)\,\dd t$

 \vspace{1mm}
Update the dual assignment $\offlinedual_{j^*}\gets 
\int_{0}^{\infty}
\pen(\calloc_{j^*}(\weight))\,\dd\weight$}
 }
\end{algorithm}
A few remarks are in order regarding the implementation and analysis of \Cref{alg:primal dual matching}. First, although described as a continuous allocation procedure (due to the while loop in lines 3--8), it is straightforward to see that the allocation for each online node terminates by construction. Second, the algorithm can be implemented efficiently in polynomial time. In \Cref{apx:implementation}, we present two explicit polynomial-time implementations: one based on a direct approach, and another leveraging a convex-programming formulation.

\oldcolor{It is also important to highlight a key difference between the single-resource and matching environments. As discussed in \Cref{remark:characterize-family} in \Cref{sec:single resource primal dual}, the analysis for the single-resource environment simplifies for specific penalty functions. However, the optimal competitive ratio can still be achieved by any monotone increasing penalty function $\pen:[0,1]\rightarrow \reals_+$ with $\pen(0)=0$, when suitably scaled. In contrast, for the matching environment, additional complexities arise, making the choice of penalty function critical to achieving optimal competitiveness, which we discuss next.}
\subsection{Optimal competitive ratio via generalized exponential penalty function}
Although our primal-dual algorithm for the matching environment naturally generalizes our earlier algorithm for the single-resource setting, its analysis poses new challenges due to uncertainty across both future weights of different offline nodes and future weights of the same offline node over time.

\revcolorCMTwo{Recall that our alternative approach to designing and analyzing algorithms in the single-resource environment relies on reducing general instances to truncated weight continuum instances $\continstances$ (see \Cref{remark:alternative-approach} and the overview in \Cref{app:alt-overview}). Given the generality and systematic nature of this approach, it might be tempting to attempt extending it to the matching environment.} However, applying a similar reduction to the matching environment introduces significant technical difficulties. In particular, neither truncated weight continuum instances nor the associated reduction techniques directly extend to this setting. To illustrate, note that in the single-resource environment, it is optimal to ignore (i.e., allocate nothing to) any online node $i\in\onlinenodes$ whose weight $\weight_{i1}$ is not maximal among previously arrived nodes. This property is crucial in our reduction, as it enables our reduction algorithm (\Cref{alg:algorithm reduction single resource}) to effectively "fill gaps" between consecutive weights to form continuum instances (\Cref{sec:reduction-to-continuum}). \revcolorCMTwo{In contrast, in the matching environment, disregarding arriving edges whose weights are dominated by previous arrivals can result in a suboptimal competitive ratio.} Thus, the reduction approach from the single-resource environment does not appear tractable in the matching environment.

To circumvent these difficulties and analyze the competitive ratio \emph{directly}, we restrict attention to the parametric subclass of ``generalized exponential functions,'' introduced earlier in \Cref{sec:single resource primal dual}\revcolorCMTwo{---following the exact same recipe as in our warm-up analysis for the single-resource environment (see \Cref{prop:optimal-direct-single-resource-fractional} and its proof)}. We formally define this subclass below. As we show soon, this restriction has two key advantages: (i) it significantly simplifies the competitive analysis by requiring minimal explicit information about the induced allocation quantile functions $\{\calloc_j\}_{j\in V}$ (see \Cref{remark:exponentialpenalty}); and (ii) this subclass remains sufficiently general to achieve the optimal competitive ratio through suitable parameter choices for the penalty function.

\oldcolor{
\begin{definition}[Generalized exponential penalty function]
\label{penalty-definition}
Given parameters $\expfuncparamone> 1$ and $\expfuncparamtwo \geq 0$, the \emph{generalized exponential penalty function} $\pen:[0,1]\rightarrow \reals_+$ is defined by $\pen(\calloc) = \expfuncparamtwo\left(\expfuncparamone^\calloc - 1\right),~~\forall \calloc \in [0,1].$
\end{definition}}
\oldcolor{Note that a generalized exponential penalty function is always strictly monotone increasing, as $\expfuncparamone> 1$.}

We summarize the competitive ratio of \Cref{alg:primal dual matching} with the generalized exponential penalty function in \Cref{thm:competitive ratio exponential penalty function}. At a high level, the proof involves a novel LP-based primal-dual analysis that, somewhat surprisingly, incorporates the buyback cost \emph{indirectly}. Indeed, the buyback cost emerges through algebraic properties of the generalized exponential function, despite neither buyback decisions nor their costs explicitly appearing in the linear program for the \revcolorCMTwo{offline optimum} solution.

\begin{theorem}
\label{thm:competitive ratio exponential penalty function}
For every $\expfuncparamone \geq e$
and 
$\expfuncparamtwo \geq \frac{1 + \buybackcost}{\expfuncparamone - 1}$,
\Cref{alg:primal dual matching} with generalized exponential penalty function 
$\pen(\calloc) = \expfuncparamtwo(\expfuncparamone^\calloc - 1)$
has competitive ratio at most $\approxratioexp(\buybackcost)$,
where
\begin{align*}
    \approxratioexp(\buybackcost) \triangleq
    \max_{\weight \geq 1+\buybackcost}
    \frac
    {
    (\expfuncparamtwo + 1)\log(\expfuncparamone)\weight 
    -
    \expfuncparamtwo\expfuncparamone\log(\expfuncparamone)
    }
    {
    \weight - (1 + \buybackcost) 
    }
\end{align*}
\end{theorem}
\begin{proof}{\emph{Proof.}}
Consider the linear program formulation of the maximum edge-weighted bipartite matching problem on the complete bipartite graph $\bipartitegraph = (\onlinenodes, \offlinenodes)$ with weights $\{\weightij\}_{i\in \onlinenodes,j\in \offlinenodes}$. We refer to this formulation as the primal linear program and also consider its dual linear program:
\begin{align*}
\tag{$\mathcal{P}_{\textsc{OPT}}$}
\label{eq:LP-max-weight}
\arraycolsep=1.4pt\def\arraystretch{1}
\begin{array}{llllllll}
\max  &\displaystyle\sum_{i\in\onlinenodes}
\displaystyle\sum_{j\in \offlinenodes}
\edgeallocij \weightij &~~\text{s.t.}&
& \quad\quad\quad\quad \text{min} &\displaystyle\sum_{i\in\onlinenodes }{\onlineduali}+\displaystyle\sum_{j\in \offlinenodes}{\offlinedualj}&~~\text{s.t.} \\[1.4em]
 &\displaystyle\sum_{j\in \offlinenodes}{\edgeallocij}\leq1 &  i\in\onlinenodes~,& 
& &\onlineduali+\offlinedualj\geq \weightij& i\in\onlinenodes,j\in \offlinenodes~,\\[1.4em]
 &\displaystyle\sum_{i\in \onlinenodes}{\edgeallocij}\leq 1 &j\in \offlinenodes~, &
& &\onlineduali \geq 0 &i\in \onlinenodes~, \\
 &\edgeallocij \geq 0 &i\in\onlinenodes,j\in\offlinenodes~.
 \qquad \qquad &
& &\offlinedualj\geq 0  &j\in \offlinenodes~. 
\end{array}
\end{align*}
We construct a dual assignment based on the allocation decisions made by \Cref{alg:primal dual matching} (denoted by $\ALG$ throughout the proof) as follows. Initially, set dual variables $\onlineduali \gets 0$ and $\offlinedualj \gets 0$ for all $i\in\onlinenodes$ and $j\in\offlinenodes$. Then, consider every allocation (and buyback) decision in $\ALG$. Whenever the algorithm buys back an infinitesimal fraction $dx$ of offline node $j$ previously allocated to online node $i'$ and reallocates this fraction to the current online node $i$, we update the dual variables as follows:
\begin{align*}
    \onlineduali \gets \onlineduali + 
    \log(\expfuncparamone)\left(
    \weightij - \offlinedualj
    \right)dx
    \qquad\textrm{and} \qquad
    \offlinedualj \gets 
    \offlinedualj + 
    \left(\displaystyle\int_{\weightipj}^{\weightij}
    \penderivative(\callocj(\weight))\,d\weight
    \right)dx~~,
\end{align*}
where $\penderivative(\cdot)$ denotes the first-order derivative of the penalty function $\pen(\cdot)$, i.e., $\penderivative(y)=\frac{\partial}{\partial y}\pen(y)$, and $\callocj(\cdot)$ is the allocation quantile function of offline node $j$ at the time of the update. By construction, the invariant $\offlinedualj = \int_0^{\infty}\pen(\callocj(\weight))\,d\weight$
is maintained throughout the execution of $\ALG$.

The remainder of the proof proceeds in two steps:

\noindent
[\emph{Step i}] \emph{Checking dual feasibility.}  
We first show that the constructed dual assignment is feasible. By construction, whenever $\ALG$ allocates an infinitesimal fraction $dx$ of offline node $j$ to online node $i$, we must have $\weightij - \offlinedualj \geq 0$. Therefore, it follows that $\onlineduali \geq 0$ and $\offlinedualj \geq 0$ for all $i\in\onlinenodes, j\in\offlinenodes$. To show feasibility of the dual constraint $\onlineduali + \offlinedualj \geq \weightij$, we consider two distinct cases:
\begin{itemize}[leftmargin=20pt]
    \item \underline{Case I --- $\ALG$ does not exhaust the unit demand of online node $i$:}  
    By construction, upon completion of the allocation for online node $i$, we must have $\offlinedual_{j'} \geq \weight_{ij'}$ for all offline nodes $j' \in \offlinenodes$. Thus, the dual constraint $\onlineduali + \offlinedual_{j'} \geq \weight_{ij'}$ holds trivially.

    \item \underline{Case II --- $\ALG$ exhausts the unit demand of online node $i$:}  
    By construction, each infinitesimal fraction $dx$ allocated to offline node $j'$ must satisfy $\weight_{ij'} - \offlinedual_{j'} \geq \weightij - \offlinedualj$ \revcolorCMTwo{for all nodes $j \in \offlinenodes$}. Thus, the dual variable $\onlineduali$ can be lower bounded as follows:
   \begin{align*}
        \onlineduali \geq \displaystyle\int_{0}^{1}
        \log(\expfuncparamone) (\weightij - \offlinedualj^{(i, x)})
        dx
        \geq
        \log(\expfuncparamone) (\weightij - \offlinedualj^{(i, 1)})
        \geq 
        \weightij - \offlinedualj~~,
    \end{align*}
    where $\offlinedualj^{(i,x)}$ is the value of dual variable $\offlinedualj$ after a fraction $x$ of online node $i$ is matched with offline nodes, and the last inequality follows since $\expfuncparamone \geq e$ and $\offlinedualj$ is non-decreasing throughout $\ALG$'s execution.
\end{itemize}
   
\smallskip
\noindent[\emph{Step ii}] \emph{Comparing objective values in primal and dual.}  
Next, we show that the total profit of $\ALG$ provides a $\approxratioexp(f)$-approximation to the objective value of the dual assignment constructed above. To establish this, we examine each allocation (and buyback) decision made by $\ALG$ and analyze its incremental impact on both the primal profit and the dual objective.

Suppose $\ALG$ buys back an infinitesimal fraction $dx$ of offline node $j$ from an earlier online node $i'$ and reallocates this fraction to the current online node $i$. The resulting net change in the primal objective, after incorporating the buyback cost, is equal to:
\begin{align*}
    \Delta(\text{Primal}) &= \left(\weightij - (1+\buybackcost)\weightipj\right)dx~.
\end{align*}
Furthermore, the change in the dual objective is upper-bounded as follows:
\begin{align*}
    \Delta(\text{Dual}) &=
    \left(
    \log(\expfuncparamone)
    \left(
    \weightij -
    \displaystyle\int_0^\infty \pen(\callocj(\weight))\,\dd\weight
    \right)
    +
    \int_{\weightipj}^{\weightij} 
    \penderivative(\callocj(\weight))\,\dd\weight
    \right)dx
    \\
    &\overset{(a)}{\leq} 
    \left(
    \log(\expfuncparamone)
    \left(
    \weightij -
    \displaystyle\int_0^{\weightipj} \pen(1)\,\dd\weight
    -
    \displaystyle\int_{\weightipj}^{\weightij} \pen(\callocj(\weight))\,\dd\weight
    \right)
    +
    \int_{\weightipj}^{\weightij} 
    \penderivative(\callocj(\weight))\,\dd\weight
    \right)dx
    \\
    &=
    \left(
    \log(\expfuncparamone)
    \left(
    \weightij -
    \expfuncparamtwo\left(\expfuncparamone - 1\right)\weightipj
    -
    \displaystyle\int_{\weightipj}^{\weightij} 
    \expfuncparamtwo\left(
    \expfuncparamone^{\callocj(\weight)} - 1
    \right)\,\dd\weight
    \right)
    +
    \int_{\weightipj}^{\weightij} 
    \log(\expfuncparamone)
    \expfuncparamtwo
    \expfuncparamone^{\callocj(\weight)}
    \,\dd\weight
    \right)dx
    \\
    &=
    \left( (\expfuncparamtwo + 1)\log(\expfuncparamone)\weightij
    -
    \expfuncparamtwo\expfuncparamone\log(\expfuncparamone)
    \weightipj \right)dx~~,
\end{align*}
where inequality~(a) holds 
by dropping $\int_{\weightij}^\infty \pen(\callocj(\weight))\,d\weight$ 
and the fact that $\callocj(\weight) = 1$ 
for every~$\weight\leq\weightipj$.
Since $\Delta(\text{Dual})\revcolorCMTwo{\leq} 
    ({
    (\expfuncparamtwo + 1)\log(\expfuncparamone)\weightij
    -
    \expfuncparamtwo\expfuncparamone\log(\expfuncparamone)
    \weightipj
    })
    \cdot
    ({
    \weightij
    -
    (1+\buybackcost)\weightipj
    })^{-1}
    \cdot
    \Delta(\text{Primal})$, to show why we have $\Delta(\text{Dual}) \leq \approxratioexp(\buybackcost) 
\cdot \Delta(\text{Primal})$
it remains to argue $\frac{\weightij}{\weightipj} \geq 1 + f$.
Note that 
\begin{align*}
    \weightij 
&\geq \int_0^\infty \pen(\callocj(\weight))\,\dd\weight 
\geq 
\int_0^{\weightipj} \pen(\callocj(\weight))\,\dd\weight 
=
\pen(1)\weightipj 
\geq (1 + \buybackcost)\weightipj~~,
\end{align*}
where the last inequality holds since $\expfuncparamtwo \geq \frac{1+f}{\expfuncparamone - 1}$. By summing $\Delta(\text{Dual})$ and $\Delta(\text{Primal})$ over all allocations and buyback decisions throughout the horizon, we obtain:
$$
\textrm{total-profit}(\ALG)\triangleq  \text{Primal}\geq \frac{1}{ \approxratioexp(\buybackcost) }\cdot\text{Dual}
$$
Finally, by weak duality of the linear program, we have $\text{Dual}\geq \text{profit}(\OPT)$, completing the proof.\hfill\halmos
\end{proof}

\revcolorCMTwo{
\begin{remark}[Role of the exponential penalty function]\label{remark:exponentialpenalty}
In the proof of \Cref{thm:competitive ratio exponential penalty function}, setting $\pen(\calloc) = \expfuncparamtwo(\expfuncparamone^\calloc - 1)$ satisfies the differential relation $\penderivative(\calloc) - \log(\expfuncparamone)\pen(\calloc) = \expfuncparamtwo\log(\expfuncparamone)$. This causes the state-dependent terms $\expfuncparamone^{\callocj(\weight)}$ in the integrals of $\Delta(\text{Dual})$ to cancel out exactly, yielding a simple linear bound in $\weightij$ and $\weightipj$. This specific functional form thus plays a critical role in our analysis.
\end{remark}}
Given \Cref{thm:competitive ratio exponential penalty function}, we are now ready to obtain the optimal competitive algorithm by properly selecting the parameters $(\expfuncparamone,\expfuncparamtwo)$ in the generalized exponential penalty function $\pen(\calloc) = \expfuncparamtwo(\expfuncparamone^\calloc - 1)$. As discussed in the introduction, we distinguish between two regimes based on the buyback factor $\buybackcost$: the \revcolorCMTwo{small buyback cost regime} (i.e., $f\leq \frac{e-2}{2}$) and the \revcolorCMTwo{large buyback cost regime} (i.e., $\buybackcost\geq \frac{e-2}{2}$). In each regime, we explicitly select parameters $(\expfuncparamone,\expfuncparamtwo)$ such that the resulting algorithm achieves the optimal competitive ratio $\optCRgen$; see \Cref{fig:CRs} for the competitive ratios and \Cref{fig:tauLambda} for the exact parameter selections as functions of $\buybackcost$. In \Cref{sec:lower-bound}, we provide tight lower-bound constructions matching these upper bounds, thereby proving that \Cref{alg:primal dual matching}, with appropriately chosen parameters, is optimally competitive for all $\buybackcost \geq 0$.

\begin{corollary}[Optimal competitive ratio in the small buyback cost regime]
\label{coro:optimal competitive ratio small f}
For every buyback factor $0\leq \buybackcost\leq \fracthreshold$,
\Cref{alg:primal dual matching}
with the generalized exponential penalty function $\pen(\calloc) = 
\frac{1 + \buybackcost}{e - (1 + \buybackcost)}(e^\calloc - 1)$
achieves a competitive ratio of at most $\frac{e}{e - (1 + \buybackcost)}$.
\end{corollary}
\begin{proof}{\emph{Proof.}}
Let $\expfuncparamone = e$ and 
$\expfuncparamtwo = \frac{1+\buybackcost}{e - (1 + \buybackcost)}$.
By construction, $\expfuncparamtwo\geq\frac{1+\buybackcost}{\expfuncparamone - 1}$.
Invoking \Cref{thm:competitive ratio exponential penalty function},
the competitive ratio is at most the following bound, which finishes the proof:
\begin{align*}
    \max_{\weight \geq 1+\buybackcost}
    \frac
    {
    (\expfuncparamtwo + 1)\log(\expfuncparamone)\weight 
    -
    \expfuncparamtwo\expfuncparamone\log(\expfuncparamone)
    }
    {
    \weight - (1 + \buybackcost) 
    }
    &=
    \max_{\weight \geq 1+\buybackcost}
    \frac{
    \left(
    \frac{1+\buybackcost}{e - (1 + \buybackcost)} + 1
    \right)\weight 
    -
    \frac{1+\buybackcost}{e - (1 + \buybackcost)} e
    }{
    \weight - (1 + \buybackcost)
    }
    \\
    &=
    \max_{\weight \geq 1+\buybackcost}
    \frac{
    \frac{e}{e - (1 + \buybackcost)}
    \weight 
    -
    \frac{e}{e - (1 + \buybackcost)}(1+\buybackcost)
    }{
    \weight - (1 + \buybackcost)
    }
    =
    \frac{e}{e - (1 + \buybackcost)}.
    \qquad\halmos
    \end{align*}
\end{proof}

\begin{corollary}[Optimal competitive ratio in the large buyback cost regime.]
\label{coro:optimal competitive ratio large f}
For every buyback factor $\buybackcost\geq \fracthreshold$,
let $\thresholdweight = -(1+\buybackcost)\Lambertterm$. Then
\Cref{alg:primal dual matching}
with the generalized exponential penalty function $\pen(\calloc) = 
\frac{1}{\log(\thresholdweight)}(\thresholdweight^\calloc - 1)$
achieves a competitive ratio of at most $-\Lambertterm$.  
\end{corollary}
\begin{proof}{\emph{Proof.}}
By definition, we have 
$    1 + \buybackcost = \frac{\thresholdweight}{\log(\thresholdweight) + 1}$.
Let $\expfuncparamone = \thresholdweight$ and
$\expfuncparamtwo = \frac{1}{\log(\thresholdweight)}$.
By construction $\expfuncparamtwo > \frac{1+\buybackcost}{\expfuncparamone-1}$
and $\expfuncparamone \geq e$
for every $\buybackcost \geq \fracthreshold$. 
We finish the proof by invoking \Cref{thm:competitive ratio exponential penalty function}, showing
the competitive ratio is at most
\begin{align*}
    \max_{\weight \geq 1+\buybackcost}
    \frac
    {
    (\expfuncparamtwo + 1)\log(\expfuncparamone)\weight 
    -
    \expfuncparamtwo\expfuncparamone\log(\expfuncparamone)
    }
    {
    \weight - (1 + \buybackcost) 
    }
    &=
    \max_{\weight \geq 1+\buybackcost}
    \frac
    {
    \left(\frac{1}{\log(\thresholdweight)} + 1\right)
    \log(\thresholdweight)\weight 
    -
    \frac{1}{\log(\thresholdweight)}\thresholdweight\log(\thresholdweight)
    }
    {
    \weight - \frac{\thresholdweight}{\log(\thresholdweight) + 1}
    }
    \\
    &=
    \max_{\weight \geq 1+\buybackcost}
    \frac
    {(\log(\thresholdweight) + 1)\weight - \thresholdweight
    }
    {
    \frac{(\log(\thresholdweight) + 1)\weight - \thresholdweight}{\log(\thresholdweight) + 1}
    }=
    \log(\thresholdweight) + 1 
    =
    -\Lambertterm.
    \quad\halmos
\end{align*}
\end{proof}

\oldcolor{
\begin{remark}[\revcolorCMTwo{Origin of the phase transition in matching}]
\label{rem:phase-transition}
The choice of parameters $(\expfuncparamone, \expfuncparamtwo)$ in the generalized exponential penalty function 
$\pen(\calloc) = \expfuncparamtwo(\expfuncparamone^\calloc - 1)$ reveals the phase-transition point $\hat{f} = \fracthreshold$ for the buyback factor $\buybackcost$ (\Cref{fig:CRs,fig:tauLambda}) and explains why the optimal competitive ratio in the matching environment is \emph{strictly worse} than in the single-resource setting. In the latter, it suffices to ensure $\offlinedual \geq \weight_i$ at the departure of each online node $i$, whereas the matching environment imposes $\onlineduali + \offlinedualj \geq \weightij$ even for edges $(i,j)$ receiving no allocation, forcing $\expfuncparamone \geq e$. When $\buybackcost < \fracthreshold$, this constraint becomes binding (the original value $-(1+\buybackcost)\Lambertterm$ from the single-resource case falls below $e$), thereby fundamentally limiting the achievable competitive ratio in the matching environment.
\end{remark}
}

We conclude this section with two important remarks about \Cref{alg:primal dual matching}. First, we discuss how this fractional algorithm can be effectively rounded into a randomized integral algorithm without significant loss in performance. Second, we address the scenario in which the exact value of the buyback factor $\buybackcost$ may not be known to the algorithm in advance, but either an upper bound on $\buybackcost$ is known or $\buybackcost$ is a random variable drawn from a distribution with a known expectation.

\begin{remark}[\revcolorCMTwo{Fractional-to-integral rounding under large capacities}]
\label{remark:fractional to randomized integral}
Although \Cref{alg:primal dual matching} is originally a fractional online algorithm, it can be shown that under large capacities (i.e., when all capacities $s_j$ are at least $s_\textrm{min}$ for large enough $s_\textrm{min}$), any fractional online algorithm with competitive ratio $\Gamma$ can be converted into an integral randomized online algorithm with competitive ratio $\hat{\Gamma}(s_{\textrm{min}})$, where $\underset{s_\textrm{min}\to +\infty}{\lim}\hat{\Gamma}(s_{\textrm{min}})=\Gamma$. For more details, see \Cref{prop:rounding matching} in \Cref{apx:matching rounding}.
\end{remark}

\begin{remark}[\revcolorCMTwo{Robustness to uncertain and stochastic buyback costs}]
    \Cref{alg:primal dual matching} achieves a competitive ratio of at most $\optCRgen$, even when the actual buyback factor is not equal to $\buybackcost$ but it is guaranteed to be no more than $\buybackcost$. In particular, if the exact value of $\buybackcost$ is unknown to the algorithm, using a known upper bound $\Bar{\buybackcost}$ instead ensures a competitive ratio of at most $\optCRgenfBar$. Additionally, consider the scenario of stochastic buyback costs, where $f$ is drawn independently from a (possibly unknown) distribution each time a buyback occurs. Due to the linearity of the objective function (total weight minus buyback cost), an algorithm that sets the buyback parameter to the known mean $\expect{\buybackcost}$ achieves an expected competitive ratio of at most $\optCRgenexf$.
\end{remark}
\vspace{-3mm}
\begin{remark}[\revcolorCMTwo{Extensions: beyond matching and positive buyback cost}]
\label{rmk:extensions-main}
Motivated by practical considerations, we explore several extensions to our main model in \Cref{apx:extension}. First, in \Cref{sec:configuration}, we consider a general configuration allocation framework, where each online node can simultaneously match with multiple offline resources, capturing applications that involve combinatorial or resource-sharing assignments. Next, in \Cref{apx:submodular welfare maximization} we generalize this model to submodular welfare maximization, where the valuation derived from assigning an online node to a subset of offline nodes follows a submodular function, capturing diminishing returns and substitution. \revcolorCMTwo{Additionally, in \Cref{apx:assortment}, we extend our framework to an online assortment planning problem, where arriving consumers make selections from dynamically displayed assortments and the platform can revoke previously accepted transactions subject to buyback costs.} Finally, we extend our primal-dual analysis to allow negative buyback factors ($\buybackcost\in[-1,0)$) in \Cref{apx:negative buyback cost}, modeling scenarios where resources have secondary supply channels or overflow capacities available at a discounted value, thus generalizing the model beyond traditional buyback scenarios. Together, these extensions highlight the strength of the primal-dual approach in addressing a broad class of online allocation problems.

\end{remark}

\section{Optimal Competitive Deterministic Integral Algorithms}
\label{sec:deterministic}
In this section, we consider the model where allocations and buybacks are restricted to be integral. Specifically, (i) the algorithm must maintain an integral matching at every point in time, and (ii) if the algorithm decides to buy back an edge-weight, it must buy back the entire edge. In \Cref{subsec:deterministic-upper-bound}, we present the optimal deterministic integral algorithm (\Cref{alg:opt deterministic matching}). Both the algorithm and its competitive ratio analysis are natural adaptations of the primal-dual algorithm and analysis for fractional allocations presented in \Cref{sec:matching}. We show the optimality of our algorithm by establishing a tight lower bound on the optimal competitive ratio in \Cref{subsec:deterministic-lower-bound}.

\subsection{Upper bound via primal-dual}
\label{subsec:deterministic-upper-bound}
We now present the optimal competitive deterministic integral algorithm. This primal-dual algorithm is parameterized by a \emph{penalty scalar} $\penscalar \geq \oldcolor{1+}\buybackcost$ and maintains the currently allocated weight $\buybackweightj$ for each offline node $j\in\offlinenodes$.\footnote{Due to the integrality of allocations, $\buybackweightj$ also represents the exact weight that must be bought back when reallocating resources.} Upon arrival of a new online node $i\in\onlinenodes$, the algorithm allocates node $i$ to the offline node $j^*$ with the largest value of $\weight_{ij^*} - \penscalar\cdot \buybackweight_{j^*}$ (i.e., $j^*=\argmax_{j\in\offlinenodes} \weightij-\penscalar\cdot \buybackweightj$), provided this value is nonnegative. If this value is negative for all offline nodes $j\in\offlinenodes$, the algorithm leaves the online node $i$ unmatched. The formal description of this algorithm is given in \Cref{alg:opt deterministic matching}.

\begin{algorithm}
\caption{Primal-dual deterministic integral online algorithm (matching)}
\label{alg:opt deterministic matching}
    \SetKwInOut{Input}{input}
    \SetKwInOut{Output}{output}
 \Input{penalty scalar $\penscalar$
 }
 
 \vspace{2mm}
 
 Initialize $\buybackweightj \gets 0$
 for every offline node $j\in\offlinenodes$.
 
 \vspace{1mm}
 
 \For{each online node $i\in\onlinenodes$}
 {
 
 \vspace{1mm}
    
 \If{there exists $j\in\offlinenodes$ s.t.\ 
 $\penscalar\cdot \buybackweightj < \weightij$}
 {
 
 \vspace{2mm}
 
 Let $j^* \gets \argmax_{j\in\offlinenodes}\ 
 \weightij - \penscalar\cdot \buybackweightj$ 
 
 \vspace{1mm}
 
 Buy back offline node $j^*$ from 
 the previously allocated online node $i'$ with $w_{i'j^*}\equiv\buybackweight_{j^*}$
 
 \vspace{1mm}
 Allocate offline node $j^*$ to online node $i$ 
 
  \vspace{1mm}
 Update $\buybackweight_{j^*}\leftarrow \weight_{ij^*}$
}
}
\end{algorithm}

\begin{remark}[\revcolorCMTwo{Primal-dual interpretation of \Cref{alg:opt deterministic matching}}]
\label{remark:det-special-case}
\Cref{alg:opt deterministic matching} 
belongs to our primal-dual family of algorithms in \Cref{sec:matching} (\Cref{alg:primal dual matching}) when restricted to making integral allocations. 
To see this, consider an algorithm that maintains
a dual assignment $\offlinedualj = \int_0^{\infty} \pen(\calloc(\weight))\,\dd\weight$ for each offline node $j\in\offlinenodes$.
Upon the arrival of a new online node 
$i$, the algorithm assigns the arriving online node to the offline node $j^*$ that maximizes
$\weight_{ij^*} - \offlinedual_{j^*}$, provided this difference is positive. If so, the algorithm also buys back $j^*$ from an earlier online node $i'$ with $\weight_{i'j^*}=\buybackweight_{j^*}$. Importantly, by replacing continuous fractional allocations with deterministic integral, the allocation quantile function of every offline node $j$ becomes a step function that  switches from $1$ to $0$ at the previously allocated weight $\buybackweightj$. Therefore, the dual assignment reduces to $\offlinedualj=\int_{0}^{\buybackweightj }\pen(1)\,\dd\weight \equiv \penscalar\cdot \buybackweightj$ if we set penalty scalar $\tau=\pen(1)$, which is exactly the ``dual assignment'' used in \Cref{alg:opt deterministic matching}.
\end{remark}

We now analyze the competitive ratio of \Cref{alg:opt deterministic matching}
in \Cref{thm:competitive ratio deterministic integral}. The proof is a simple adaptation of our previous primal-dual proof of \Cref{thm:competitive ratio exponential penalty function} for the integral setting.

\begin{restatable}{theorem}{optdetcompetitiveratio}
\label{thm:competitive ratio deterministic integral}
For $\penscalar \geq 1 + \buybackcost$,
\Cref{alg:opt deterministic matching} with penalty scalar $\penscalar$
has competitive ratio at most $\approxratiodet(\buybackcost)$,
where
\begin{align*}
    \approxratiodet(\buybackcost) \triangleq
    \max_{\weight \geq \penscalar}
    \frac
    {
    (\penscalar + 1)\weight 
    -
    2\penscalar
    }
    {
    \weight - (1 + \buybackcost) 
    }
\end{align*}
\end{restatable}
\begin{proof}{\emph{Proof.}}
Recall the linear program of maximum edge-weighted bipartite matching in complete bipartite graph $\bipartitegraph = (\onlinenodes, \offlinenodes)$
as the primal linear program, and its dual in \ref{eq:LP-max-weight}.
We construct a dual assignment based on the allocation decisions
of \Cref{alg:opt deterministic matching} (denoted by $\ALG$ throughout the proof).
First, set $\onlineduali \gets 0$ and $\offlinedualj \gets 0$
for all $i\in\onlinenodes,j\in\offlinenodes$.
Now consider every allocation (and buyback) decision in 
$\ALG$.
Whenever $\ALG$
buys back offline node $j$ from online node $i'$
(with weight $\weightipj \equiv \buybackweightj$)
and re-allocates offline node $j$
to online node $i$, update the dual variables as follows:
\begin{align*}
    \onlineduali \gets \weightij - 
    \penscalar\cdot \buybackweightj
    \qquad\textrm{and}\qquad
    \offlinedualj \gets 
    \offlinedualj + 
    \penscalar(\weightij - \buybackweightj)
\end{align*}
By construction, the invariant $\offlinedualj \equiv 
\penscalar\cdot \buybackweightj$
holds throughout $\ALG$. Similar to our earlier primal-dual proof of \Cref{thm:competitive ratio exponential penalty function}, the rest of this proof is also done in two steps:

\noindent
[\emph{Step i}] \emph{Checking the feasibility of dual.}
We first show that the constructed dual assignment 
is feasible.
By construction, 
whenever $\ALG$ allocates the
offline node $j$ to the online node $i$, 
we know that $\weightij - \penscalar \cdot \buybackweightj \geq 0$.
Thus, $\onlineduali \geq 0$ and $\offlinedualj \geq 0$
for all $i\in\onlinenodes,j\in\offlinenodes$. Next, we show 
$\onlineduali + \offlinedualj \geq \weightij$ by considering two cases:
\begin{itemize}
    \item \underline{Case I --- $\ALG$ does not allocate online node $i$ to any offline node}: 
    By construction, we know that this happens only if
    $\offlinedual_j = \penscalar\cdot \buybackweightj \geq \weightij$.
    Thus, the dual constraint is satisfied.
    \item \underline{Case II --- $\ALG$ allocates online node $i$ to offline node $j'$}:
    By construction, $\ALG$ greedily allocates online node $i$ 
to offline node $j^*$ which maximizes
$\weight_{ij^*} - \offlinedual_{j^*}$, meaning that 
    $\weightij - \penscalar\cdot \buybackweightj \leq 
    \weight_{ij'} - \penscalar\cdot \buybackweight_{j'}$.
    Thus, the dual constraint is satisfied.
\end{itemize}

\smallskip
\noindent
[\emph{Step ii}] \emph{Comparing objective values in primal and dual.}
We show that the total profit of $\ALG$ 
is a $\approxratiodet(f)$-approximation 
of the objective value of 
the above dual assignment.
To show this, we consider every allocation (and buyback)
decision in $\ALG$ and its impact on the total profit 
and the objective value of the dual assignment. In particular, suppose $\ALG$ buys back offline node $j$ 
from online node $i'$ (with weight $\weightipj\equiv\buybackweightj$)
and then re-allocates it to online node $i$.
The change in the profit (the net change in the primal objective \emph{after} we incorporate buyback cost) is $\Delta(\text{Primal}) = 
    \weightij - (1+\buybackcost)\buybackweightj,$
and the change in the dual objective is $
    \Delta(\text{Dual}) =
    \weightij - \penscalar\cdot \buybackweightj
    +
    \penscalar(\weightij - \buybackweightj)=
    (\penscalar + 1)\weightij
    -
    2\penscalar\buybackweightj$. \oldcolor{Combining with the fact that
$\weightij \geq \penscalar\cdot \buybackweightj$,
we have 
\begin{align*}
    \frac{\Delta(\text{Dual})}{\Delta(\text{Primal})} = \frac{
    (\penscalar + 1)\weightij
    -
    2\penscalar\buybackweightj}{\weightij - (1+\buybackcost)\buybackweightj}=
    \frac{
    (\penscalar + 1)
\frac{\weightij}{\buybackweightj}
    -
    2\penscalar}{
\frac{\weightij}{ \buybackweightj} - (1+\buybackcost)} \leq \approxratiodet(\buybackcost).
\end{align*}}
Again, by summing $\Delta(\text{Dual})$ and $\Delta(\text{Primal})$ over the entire horizon, we obtain:
$$
\textrm{total-profit}(\ALG)\triangleq  \text{Primal}\geq \frac{1}{ \approxratiodet(\buybackcost) }\cdot\text{Dual}
$$
Finally, by weak duality of the linear program, $\text{Dual}\geq \textrm{profit}(\OPT)$, which  finishes the proof.
\hfill\halmos
\end{proof}

Given \Cref{thm:competitive ratio deterministic integral},
we are ready to obtain the optimal competitive deterministic integral algorithm
by appropriately picking the value of $\penscalar$ in \Cref{alg:opt deterministic matching}.
Similar to the fractional matching setting in \Cref{sec:matching},
we use two different value assignments of $\penscalar$,
depending on the buyback factor~$\buybackcost$.

\oldcolor{\begin{corollary}[Competitive ratio $\boldsymbol{\optCRdet}$ in the small buyback cost regime]
\label{coro:optimal deterministic competitive ratio small f}
For every buyback factor $0\leq\buybackcost < 1$,
\Cref{alg:opt deterministic matching}
with $\penscalar = \frac{1 + \buybackcost}{1 - \buybackcost}$
has competitive ratio at most $\frac{2}{1 - \buybackcost}$.
\end{corollary}}
\begin{proof}{\emph{Proof.}}
By construction, $\penscalar\geq {1+\buybackcost}$.
We invoke \Cref{thm:competitive ratio deterministic integral}. Therefore, the competitive ratio is at most 
\begin{align*}
    \max_{\weight \geq \penscalar}
    \frac
    {
    (\penscalar + 1)\weight 
    -
    2\penscalar
    }
    {
    \weight - (1 + \buybackcost) 
    }
    =
    \max_{\weight \geq \penscalar}
    \frac{
    \left(
    \frac{1+\buybackcost}{1 - \buybackcost} + 1
    \right)\weight 
    -
    2\cdot \frac{1+\buybackcost}{1 - \buybackcost} 
    }{
    \weight - (1 + \buybackcost)
    }
    =
    \max_{\weight \geq \penscalar}
    \frac{
    \frac{2}{1 - \buybackcost}
    \weight 
    -
    \frac{2}{1- \buybackcost}(1+\buybackcost)
    }{
    \weight - (1 + \buybackcost)
    }
    =
    \frac{2}{1- \buybackcost}~~.\quad\qquad\qquad
    \halmos
\end{align*}
\end{proof}

\begin{corollary}[Competitive ratio $\boldsymbol{\optCRdet}$ in the large buyback cost regime]
\label{coro:optimal deterministic competitive ratio large f}
For every buyback factor $\buybackcost\geq \detthreshold$,
\Cref{alg:opt deterministic matching}
with $\penscalar = 1 + \buybackcost +
\sqrt{\buybackcost(1+\buybackcost)}$
has competitive ratio at most $
1 + 2\buybackcost + 
2\sqrt{\buybackcost(1+\buybackcost)}$.
\end{corollary}
\begin{proof}{\emph{Proof.}}
By construction $\penscalar > 1+\buybackcost$.
Invoking \Cref{thm:competitive ratio deterministic integral},
the competitive ratio is at most 
\begin{align*}
    \max_{\weight \geq \penscalar}
    \frac
    {
    (\penscalar + 1)\weight 
    -
    2\penscalar
    }
    {
    \weight - (1 + \buybackcost) 
    }
    &=
    \max_{\weight \geq \penscalar}
    \frac
    {
    \left(
    2  + \buybackcost +
\sqrt{\buybackcost(1+\buybackcost)}
    \right)\weight 
    -
    2\left( 1 + \buybackcost +
\sqrt{\buybackcost(1+\buybackcost)}\right)
    }
    {
    \weight - (1 + \buybackcost)
    }
\end{align*}
Thus, it is sufficient to show that 
for every $\weight \geq  1  + \buybackcost +
\sqrt{\buybackcost(1+\buybackcost)}$,
\begin{align*}
     \left(
    2  + \buybackcost +
\sqrt{\buybackcost(1+\buybackcost)}
    \right)\weight 
    -
    2\left( 1 + \buybackcost +
\sqrt{\buybackcost(1+\buybackcost)}\right)
\leq 
\left(
1 + 2\buybackcost + 
2\sqrt{\buybackcost(1+\buybackcost)}
\right)
\left(\weight - (1 + \buybackcost)\right)~.
\end{align*}
This is simplified to $\left(
    \buybackcost - 1 + 
    \sqrt{\buybackcost(1 + \buybackcost)}\right)
    \left(
    \weight - 
    \left(1  + \buybackcost +
\sqrt{\buybackcost(1+\buybackcost)}
    \right)\right)
    \geq 0$,
which holds as $\buybackcost\geq \detthreshold$. \hfill\halmos
\end{proof}
\oldcolor{Comparing the competitive ratio bound in \Cref{coro:optimal deterministic competitive ratio small f,coro:optimal deterministic competitive ratio large f}, we remark that for $\buybackcost<1$,
$\frac{2}{1-\buybackcost} \geq 1 + 2\buybackcost + 
2\sqrt{\buybackcost(1+\buybackcost)}$ is equivalent to:
\begin{align*}
    1+2\buybackcost^2-\buybackcost\geq 2(1-\buybackcost)\sqrt{\buybackcost(1+\buybackcost)} \iff \left(\buybackcost-\frac{1}{3}\right)^2 \geq 0.
\end{align*}
Therefore, the bound in \Cref{coro:optimal deterministic competitive ratio large f} always provides a stronger result when applicable, i.e., for $\buybackcost\geq \frac{1}{3}$, whereas the bound in \Cref{coro:optimal deterministic competitive ratio small f} is tight for $\buybackcost \leq \frac{1}{3}$.}

\begin{remark}[\revcolorCMTwo{Economic intuition behind the inflated threshold $\tau$}]\label{remark:tau_greater_1+f}
    Note that for positive values of $f$, the optimal $\tau$ in both \Cref{coro:optimal deterministic competitive ratio small f} and \Cref{coro:optimal deterministic competitive ratio large f} is strictly greater than $1+f$. This arises because although selecting a node with a higher weight yields a greater payoff at the end of the time horizon, it also incurs a higher cost if the algorithm later cancels that node. Consequently, for the algorithm to benefit from switching to a higher-weight node, it must observe strictly positive marginal gains; therefore, the algorithm employs an inflated buyback cost to guide its decision-making.
\end{remark}

\vspace{-2mm}
\section{Competitive Ratio Lower Bounds}
\label{sec:lower-bounds-apx}

In this section, we present worst-case instances and establish tight competitive ratio lower bounds for all values of the buyback parameter $f$. We consider the general setting in \Cref{sec:lower-bound}, when fractional allocations and randomization are allowed, and the deterministic integral setting in \Cref{subsec:deterministic-lower-bound}.
\vspace{-3mm}
\subsection{General setting}
\label{sec:lower-bound}
As alluded to earlier in \Cref{sec:intro}, there are two sources of uncertainty at play in our problem: (i)~weight-wise uncertainty, which is the uncertainty in the future arriving weights of the same offline node, and it was already present in single-resource environment (\Cref{sec:single-resource}); and (ii)~edge-wise uncertainty, which is the uncertainty related to the heterogeneity among offline nodes in terms of their connections and weights and it was only present in the matching environment (\Cref{sec:matching}). We use these two sources quite differently in designing worst-case lower bound instances.
\oldcolor{We first present a hard instance for the small buyback cost regime (i.e., $\buybackcost\leq \fracthreshold$) that exploits both edge-wise uncertainty and weight-wise uncertainty in the matching environment. We then present a hard instance for the large buyback cost regime (i.e., $\buybackcost\geq \fracthreshold$) that solely utilizes weight-wise uncertainty in the single-resource environment. By combining them, we obtain the following theorem.}

\oldcolor{
\begin{restatable}{theorem}{lowerboundmatching}
\label{thm:lower bound all f}
In the matching environment, the optimal competitive ratio $\optCRgen$ 
is at least 
$\frac{e}{e - (1 + \buybackcost)}$ for every buyback factor $\buybackcost\leq \fracthreshold$,
and $-\Lambertterm$ for every buyback factor $\buybackcost\geq \fracthreshold$.\footnote{
As a sanity check, $\frac{e}{e - (1 + \buybackcost)}\geq -\Lambertterm$
for every $\buybackcost \leq \fracthreshold$, and the equality holds 
when $\buybackcost = \fracthreshold$.}
\end{restatable}}

\oldcolor{\subsubsection{Hard instance under small buyback factor}
\label{subsec:lower-bound-matching}
We construct a hard instance in \Cref{example:lower bound matching}, which uses both edge-wise and  weight-wise uncertainties to derive the tight lower bound of the competitive ratio in the small buyback cost regime, that is, when $\buybackcost\leq \fracthreshold$. 
The formal analysis of this instance is deferred to \Cref{sec:apx-missing lower bound matching small f}. Below, we outline the key intuition behind our analysis, focusing on the characterizations of the optimal offline benchmark and optimal online algorithm.}

\begin{example}
\label{example:lower bound matching}
Fix an arbitrary $\largenumber\in\naturals$. 
Let $\permu:[\largenumber]\rightarrow[\largenumber]$ be a uniform random permutation over $[\largenumber]$.
Consider a randomized instance with bipartite graph $\bipartitegraph = (\onlinenodes,\offlinenodes,\Edge)$ 
where $\onlinenodes= [\largenumber]$, $\offlinenodes =[\largenumber]$, $\Edge = [\largenumber]\times[\largenumber]$, and edge-weights $\{\weightij\}_{i,j\in[\largenumber]}$ are as follows:
for every edge $(i,j)\in[\largenumber]\times[\largenumber]$,
$\weightij = \frac{1}{\largenumber - i + 1}
    \cdot 
    \indicator{\permu(j) \geq i}$.
\end{example}

In \Cref{example:lower bound matching},
conditioning on a realized permutation $\permu$,
each online node $i\in[\largenumber]$ 
has edge-weights $\frac{1}{\largenumber-i + 1}$
for $(K - i + 1)$ offline nodes,
and zero edge-weight for the remaining offline nodes. 
Furthermore, for every online nodes $i$ and $i + 1$,
the set of offline nodes with non-zero edge-weights are nested,
i.e., $\{j\in[\largenumber]:\weightij > 0\} \equiv \{j\in[\largenumber]:\weight_{i+1,j}>0\}
\cup\{\permu^{-1}(i)\}$.\footnote{Here $\permu^{-1}(\cdot)$ 
is the inverse of permutation 
$\permu$, i.e., $\pi(i)=j \leftrightarrow i=\pi^{-1}(j)$}
By construction, it is straightforward to verify that the \revcolorCMTwo{offline optimum} benchmark matches 
each online node $i$ to offline node $\permu^{-1}(i)$ and
collects total profit $\sum_{i\in[\largenumber]}\sfrac{1}{i}$.


Next, we characterize the optimal online algorithm 
in \Cref{example:lower bound matching}.
We start with a simple and intuitive observation that when  $\buybackcost = 0$
(i.e., no buyback cost),
since the (non-zero) edge-weights are weakly increasing in \Cref{example:lower bound matching}, it is always optimal to allocate the online node $i$ in full to offline nodes $j$ with non-zero edge-weights.
As the main technical ingredient of this lower bound result,
we prove that this observation for $\buybackcost = 0$ approximately holds when
$\buybackcost\leq \fracthreshold$. Specifically, we show that 
for any such buyback factor $\buybackcost$, there exists $\fullallocthreshold\in\naturals$ (independent of $\largenumber$)
such that for every online node $i\leq \largenumber - \fullallocthreshold$, 
the optimal online algorithm allocates the online node $i$ in full
to offline nodes $j$ with non-zero edge-weights.
Additionally, due to the ex-ante symmetry (over the randomness in $\permu$), we also show that the optimal online algorithm
allocates equal fractions of a given online node to each offline node with non-zero edge-weights. \oldcolor{See the formal statement of the characterization in \Cref{lem:lower bound matching optimum online}.}


\subsubsection{Hard instance under large buyback factor}
\label{subsec:lower-bound-single}
\oldcolor{We construct a hard instance in \Cref{example:lower bound single resource}, which uses weight-wise uncertainty only to derive the tight lower bound of the competitive ratio in the large buyback cost regime, that is, when $\buybackcost\geq \fracthreshold$.
\revcolorCMTwo{In fact, this single-resource hard instance is indeed the truncated weight continuum instance adapted from \citet{AK-09}, formally defined and studied in \Cref{sec:reduction-to-continuum} (see \Cref{def:weight-continuum}).}
} 
\begin{example}
\label{example:lower bound single resource}
Fix an arbitrary $\totaltime_0\geq 1$.
Consider a randomized truncated weight continuum instance $\instance_\totaltime
\in
\continstances$ where $\totaltime$
is sampled from density $\sfrac{1}{\totaltime^2}$
in $\totaltime\in[1,\totaltime_0]$,
and with remaining probability $\sfrac{1}{\totaltime_0}$
is equal to $\totaltime_0$.
\end{example}
\oldcolor{
In this randomized single-resource instance, it is clear that the \revcolorCMTwo{offline optimum} benchmark matches the last online node with weight $T$ to the offline node. \citet{AK-09} characterize the expected performance of the optimal randomized integral online algorithms. Invoking our lossless randomized rounding procedure (\Cref{alg:algorithm-randomized-reduction-single-resource}) in \Cref{sec:single-lossless-rounding}, we can extend the same performance guarantee to the optimal algorithms among \emph{all} online algorithms, fractional or randomized integral. This completes the analysis of the large buyback cost regime, which can be found in \Cref{sec:apx-missing lower bound single resource}.

Remarkably, \Cref{example:lower bound single resource}, as a randomized single-resource instance, suggests that in the large buyback cost regime, weight-wise uncertainty becomes the dominant factor, and edge-wise uncertainty is no longer necessary for constructing the lower bound. Another notable observation is that, unlike the optimal online algorithm in \Cref{example:lower bound matching} (which applies to the small buyback cost regime and allocates almost fully to each online node), the optimal online algorithm in \Cref{example:lower bound single resource} strategically hedges against future weights. It deliberately refrains from fully allocating to some online nodes, even when such allocations yield strictly positive net rewards.}

\vspace{-3mm}
\subsection{Deterministic integral setting}
\label{subsec:deterministic-lower-bound}
\oldcolor{We now present tight lower bounds for deterministic integral algorithms. Notably, for such algorithms, it suffices to construct an (oblivious) \emph{adversary} that generates hard instances based on the algorithm's execution. 
Similar to \Cref{sec:lower-bound}, we explore the roles of edge-wise and weight-wise uncertainty in designing the adversary. 
In \Cref{subsec:deterministic-lower-bound-matching}, we present the adversary for the small buyback cost regime (i.e., $\buybackcost\leq \detthreshold$) that exploits both edge-wise uncertainty and weight-wise uncertainty in the matching environment. In \Cref{subsec:deterministic-lower-bound-single}, we present the adversary for the large buyback cost regime (i.e., $\buybackcost\geq \detthreshold$) that solely utilizes weight-wise uncertainty in the single-resource environment. By combining them, we obtain the following theorem.}
\oldcolor{
\begin{restatable}{theorem}{lowerboundmatchingdet}
\label{thm:lower bound deterministic all f}
In the matching environment, 
the optimal competitive ratio of deterministic integral online algorithms $\optCRdet$ is at least $\frac{2}{1-\buybackcost}$ for every  $\buybackcost\leq \detthreshold$, and $1 + 2\buybackcost + 
2\sqrt{\buybackcost(1+\buybackcost)}$ for every $\buybackcost\geq\detthreshold$.\footnote{\oldcolor{As a sanity check, $\frac{2}{1-\buybackcost}\geq 1 + 2\buybackcost + 
2\sqrt{\buybackcost(1+\buybackcost)}$
for every $\buybackcost \leq \detthreshold$, and the equality holds 
when $\buybackcost = \detthreshold$.}}
\end{restatable}
}

\subsubsection{Adversary under small buyback factor}
\label{subsec:deterministic-lower-bound-matching}
\oldcolor{We sketch our constructed adversary in \Cref{example:lower bound Deterministic matching} and outline the main analysis idea, which leverages both edge-wise and weight-wise uncertainty to derive a tight lower bound on the competitive ratio in the small buyback cost regime ($\buybackcost \leq \detthreshold$). The formal analysis of this instance is deferred to \Cref{sec:apx-missing matching lower determinstic}.}

\begin{example}
\label{example:lower bound Deterministic matching}
\oldcolor{Fix an arbitrary online deterministic integral algorithm $\ALG$. The adversary generates an instance with $n = 2$ offline nodes. Without loss of generality, assume that $\ALG$ allocates the first online node to the offline node number $1$ if $\weight_{11} = \weight_{12}\equiv \weight_0$ (if $\ALG$ decides not to allocate, the competitive ratio will be unbounded, as the adversary can just terminate the instance). Let $(x,y)$ denote an arrival where the weight to offline nodes 1 and 2 are $x$ and $y$, respectively. Let the instance be $\{(\weight_0,\weight_0),(\weight_1,0),...,(\weight_{i},0), \dots\}$, where $\{\weight_i\}_{i\in\naturals}$ is an increasing sequence such that $\weight_i$ is the infimum weight $\weight$ $\ALG$ accepts, i.e., the minimum weight at which 
$\ALG$ allocates the offline node 1 to the newly arriving online node and buys back the previously allocated one, given the previous arrivals.}
\end{example}
\oldcolor{
To see the high-level idea behind our analysis,
 fix an arbitrary online deterministic integral algorithm $\ALG$ with competitive ratio $\approxratio$. When facing the adversary in \Cref{example:lower bound Deterministic matching}, by construction, the following holds:
\begin{align*}
    \forall k\in\naturals,\forall \varepsilon\in\reals_{>0}:\qquad \frac{\weight_0+\weight_k - \varepsilon}{\weight_{k - 1} - \buybackcost\cdot(\weight_0 + \dots + \weight_{k - 2})} \leq \approxratio
\end{align*}
The rationale behind this inequality is as follows. Consider the instance with arrival weights $\{(\weight_0,\weight_0),(\weight_1,0),...,(\weight_{i-1},0),(\weight_i-\varepsilon,0)\}$ as specified in \Cref{example:lower bound Deterministic matching}. By the definition of the sequence $\{\weight_i\}_{i \in \naturals}$, the optimal offline benchmark obtains $\weight_0 + \weight_k - \varepsilon$, whereas $\ALG$ obtains $\weight_{k - 1} - \buybackcost \cdot (\weight_0 + \dots + \weight_{k - 2})$. Since $\ALG$ has competitive ratio $\approxratio$, it must achieve at least an $1/\approxratio$ fraction of the offline benchmark on this instance as well, yielding the stated inequality. Now suppose the above inequality holds with equality for all $k \in \naturals$ and $\varepsilon = 0$. (In our formal analysis, we show that the instance can be adjusted to ensure this.) Under this assumption, we obtain the following recursive equation for the sequence $\{\weight_i\}_{i\in\naturals}$: 
\begin{align*}
    \weight_i = (\approxratio+1)\weight_{i - 1} - \approxratio(1 + \buybackcost)\weight_{i - 2}
\end{align*}
By definition, $\weight_i$ must remain weakly positive for all $i \in \naturals$. This condition holds if and only if the competitive ratio $\approxratio$ is at least $\frac{2}{1-\buybackcost}$, which follows from the recursive equation above and a straightforward algebraic analysis.

\subsubsection{Adversary under large buyback factor}
\label{subsec:deterministic-lower-bound-single}
We sketch our constructed adversary in \Cref{example:lower bound Deterministic single-resource}, which uses weight-wise uncertainty only to derive the tight lower bound of the competitive ratio in the large buyback cost regime, that is, when $\buybackcost\geq \detthreshold$.
This single-resource hard instance are adapted from \citet{BHK-09}. The formal analysis of this instance is deferred to \Cref{sec:apx-missing single lower determinstic}.}


\begin{example}
\label{example:lower bound Deterministic single-resource}
\oldcolor{Fix an arbitrary online deterministic integral algorithm $\ALG$. The adversary generates an instance with $n = 1$ offline node. Fix an arbitrary positive number $\weight_0$. Let $\{\weight_i\}_{i\in\naturals}$ be the increasing sequence such that $\weight_i$ is the infimum weight $\weight$ the algorithm sells to given that the weight of previous arrivals are $\{\weight_0,\weight_1,...,\weight_{i-1}\}$. Let the instance be a sequence of online nodes with weights $\weight_0, \weight_1,\dots$, respectively.}
\end{example}
\oldcolor{The analysis parallels \Cref{subsec:deterministic-lower-bound-matching}: we show that the weight sequence ${\weight_i}_{i\in\naturals}$ obeys a similar recursive equation, and that it remains positive iff the competitive ratio is at least $1 + 2\buybackcost + 2\sqrt{\buybackcost(1+\buybackcost)}$.}


\section{Conclusion and Future Directions}
\label{sec:conclusion}

We studied the online allocation of multiple resources allowing for over-allocation and costly cancellations (or equivalently, online discarding), where cancellation costs are linear with a fixed slope $\buybackcost$ in the previously assigned weight. Our primary contribution is a complete characterization of the optimal competitive ratio, including both upper and lower bounds, in the adversarial model across all parameter regimes of the buyback factor~$\buybackcost$. We established a fundamental connection between the celebrated primal-dual framework and the buyback problem, showing how this connection can be leveraged to derive optimally competitive algorithms.

\smallskip
\noindent\textbf{Future directions.} Several promising directions exist for future research. First, further exploration of the primal-dual connection could extend its applicability to other problems involving costly cancellations. Another fruitful avenue is to investigate more sophisticated methods for augmenting current matching approaches by incorporating appropriate costs. Additionally, analyzing alternative forms of cancellation costs may yield deeper insights and richer models. Another intriguing direction is focusing on the small buyback regime to determine whether randomized integral online algorithms can achieve competitive ratios better than 2 for any $\buybackcost > 0$---analogous to the scenario explored by \cite{FHTZ-20} for the special case of $\buybackcost = 0$. As another future direction, one can think about the Bayesian version of the problem where the arrivals are drawn from known distributions. \oldcolor{The work of \cite{ENNV-24} initiated the study of prophet inequalities under the buyback model, and it would be interesting to see whether that study could be extended to matching and other combinatorial problems.} Lastly, studying a broader class of resource allocation problems, where the reward for each offline resource depends on a high-dimensional function of its allocation distribution, poses a compelling and challenging open question. Determining whether such generalized models can admit constant-competitive algorithms remains an important research direction.


\newcommand{\newblock}{}
\setlength{\bibsep}{0.0pt}
\bibliographystyle{plainnat}
\OneAndAHalfSpacedXI
{\footnotesize
\bibliography{refs}}


\renewcommand{\theHchapter}{A\arabic{chapter}}
\renewcommand{\theHsection}{A\arabic{section}}
\newpage
\ECSwitch
\ECDisclaimer
\renewcommand{\theHchapter}{A\arabic{chapter}}
\renewcommand{\theHsection}{A\arabic{section}}

\section{Further Related Work}
\label{sec:further-related-work}

\paragraph{Single-resource online allocation with buyback.}
As mentioned earlier, \citet{BHK-09} studied the single-resource allocation problem with costly cancellations, a special case of our model. They identified the optimal deterministic integral algorithm and extended their results to matroid environments and knapsack environments (with an inflated capacity constraint). In follow-up works, \citet{AK-09} identified the optimal randomized algorithm for the single-resource setting and generalized it to matroid environments. Additionally, \citet{ash-11} identified the optimal deterministic algorithm for the intersection of general matroid environments, while \citet{HKM-14} identified the optimal deterministic algorithm for knapsack environments with identical weights.
In a parallel line of work, \citet{CFMP-09} studied online resource allocation with buyback under transversal matroids. They further assumed that the online demand requests are strategic, focusing on incentive-compatible and individually rational mechanisms. \oldcolor{More recently, the single-resource buyback problem has also been studied in the prophet inequality setting~\citep{ENNV-24}, where the weights arrive independently from known distributions.}

\paragraph{Online matching with free disposal.}
\citet{FKMMP-09} and \citet{DHKMY-16} studied online edge-weighted matching with free disposal—a special case of our model where $\buybackcost = 0$—under the large capacity regime (or equivalently, fractional allocations) and identified the optimal competitive algorithm. For small capacities, the breakthrough work of \citet{FHTZ-20} introduced a randomized integral algorithm whose competitive ratio strictly exceeds the $2$ ratio achievable by the greedy algorithm. Another related work is \citet{DJ-12}, which studied online matching with concave returns. In their model, offline nodes do not have fixed capacities but instead are associated with a concave function that maps the total reward collected by the offline node to its utility. Our model differs from this because the profit of each offline node depends not only on the total reward but also on the entire allocation history, which determines the buyback costs. Lastly, \citet{BFS-14} studied online submodular maximization with free disposal, which, again, does not capture the concept of buyback costs.

\paragraph{Primal-dual framework in online matching.} The Primal-dual framework is studied extensively for designing and analyzing online resource allocation algorithms, e.g., see \cite{MSVV-07}, \cite{BN-09}, \cite{DJK-13}. More recently, this technique has found more exciting applications in both computer science and operations research, e.g., Adwords under small capacities~\citep{HZZ-20}, Adwords with unknown budgets~\citep{Udw-24,Vaz-23}, online allocation of reusable resource~\citep{GGISUW-22,DFNSU-24,GIU-20,FNS-19},  online assignment of jobs~\citep{EFKN-25}, online advance scheduling with overtime~\citep{KSV-21}, online matching with multi-channel traffic~\citep{MRSS-24}, online matching with stochastic rewards~\citep{MP-12,GU-23}, and multi-stage online matching~\citep{FNS-24,FN-25}.

\paragraph{Robust revenue management.}
Our model makes no assumptions about the sequence of arriving customers (e.g., no need for stochastic knowledge, identical rewards, etc.). Similar adversarial arrival models have been considered in various robust revenue management problems, including online booking \citep[e.g.,][]{BQ-09}, assortment \citep[e.g.,][]{GNR-14,GGISUW-22}, allocation \citep[e.g.,][]{BLM-22}, dynamic staffing with predictions\citep[e.g.,][]{feng2025robust}, and pricing\citep[e.g.,][]{LLV-18}.

\paragraph{Other related models in dynamic revenue management.}

Our work is also conceptually related to dynamic revenue management, where offline supplies have soft inventory constraints, or dynamic decisions can be altered later. Examples include dynamic revenue management with cancellation (on the demand side) and overbooking~\citep[e.g.,][]{rot-71,ET-10,ABFN-13,DKX-19,FKH-22,FZ-21,FJT-25}, as well as dynamic inventory control, matching, and staffing problems~\citep[e.g.,][]{LRS-07,mos-10,BM-13,CS-19,AJ-22,QSW-22,bansak2024dynamic,feng2025robust}.

Another related problem is the revenue management with callable products, introduced by \citet{GKP-08} and motivated by applications in the airline industry. They studied a two-period model in which low-fare consumers arrive in the first period, granting the capacity provider an option to ``call'' the capacity back at a specified recall price (which is essentially what we refer to as a buyback in our model).
In the second period, high-fare consumers arrive, and the capacity provider decides how many callable products to exercise. In two follow-up works, \citet{GL-18} extended the model to include multi-fare consumers and developed heuristic policies. Later, \citet{GL-20} studied online assortment with callable products, characterizing the optimal policy for both the general choice models and the multinomial-logit model.


\revcolor{
\section{Discussion on Practical Applications}
\label{app:practical}

\revcolorCAMtwo{
In the following, we present a list of specific applications highlighting the relevance and practical importance of modeling costly cancellations. Our model is, in particular, motivated by settings in which a decision maker (i) commits resources to sequentially arriving requests under substantial uncertainty about future demand, and (ii) may later revoke a commitment in order to accommodate a more valuable request, either at a cost---explicit or implicit---proportional to the revoked reward or with only partial recovery of its value. Although the exact mathematical model presented in this paper may not directly capture all aspects of these applications, our study of matching models involving revocations of previous allocations at a cost is partly inspired by the trade-off between platform's gain and externalities due to cancellations in these real-world scenarios.}

\subsection{\revcolorCAMtwo{Explicit vs. implicit cancellation cost}}
\label{apx:explicit-implicit-cost}
\revcolorCAMtwo{Across the applications we consider, cancellations can impose either implicit or explicit costs:
\begin{itemize}
    \item \textbf{Implicit cost:} In some applications, the earlier commitment is provisional: the decision maker has tentatively reserved capacity for a request but has not yet executed or externally guaranteed the assignment. Replacing this commitment forfeits its reward but requires little or no direct payment. If only direct payments and generated revenue are considered, this is the free-disposal case ($\buybackcost=0$).
    Display advertising provides another example of this phenomenon, but with a different interpretation: advertisers commit to paying for impressions only up to a specified capacity. Assigning impressions beyond this capacity and later discarding them incurs no direct cancellation payment, so a revenue-only view again corresponds to $\buybackcost=0$. Nevertheless, such cancellations can impose an \emph{implicit} opportunity cost. For example, a discarded impression might otherwise have been assigned to an underallocated advertiser. These indirect costs motivate modeling the application with $\buybackcost>0$, even when cancellations require no immediate payment.
    \item \textbf{Explicit cost:} In other applications, the resource has already been assigned to a user, but the assignment is \emph{explicitly} preemptible or flexible, with \emph{actual compensation} due if it is canceled (ideally proportional to the cancelled value). For example, the customer knowingly purchases a product that may be displaced and receives a make-good, a credit, or a compensation payment if displacement occurs. This corresponds to a positive buyback factor $\buybackcost>0$. Alternatively, if the displaced commitment is partially re-monetized through a discounted secondary channel, the setting corresponds to the negative buyback factors $\buybackcost\in[-1,0)$ studied in \Cref{apx:negative buyback cost}.
\end{itemize}

\subsection{\revcolorCAMtwo{Selected practical applications}}

In the remainder of this section, we mainly discuss three applications: (i) selling banner ads, (ii) display advertising, and (iii) preemptible broadcast advertising. These applications fit naturally within our model; For each application, we identify (a) what the online nodes, offline nodes, and edge weights represent; (b) what constitutes a cancellation, whether its cost is explicit or implicit, and the source—and, when known, the approximate magnitude—of the buyback factor; and (c) which allocation regime is most relevant: fractional/large-capacity or deterministic integral. We further discuss  overselling (also known as overbooking) in revenue management---including airlines, hotels and flexible cargo shipping---as a potential application that resembles certain aspects of our setting and therefore fits only partially. 
}
\subsubsection{\revcolorCAMtwo{Directly related practical applications}}
\smallskip
\paragraph{(i) Selling banner ads.} In certain digital advertising scenarios, companies sell banner ad placements to advertisers in advance~\citep{MSN}. Advertisers request digital banner slots (like a billboard on a website) targeting specific demographics, timeframes, and web locations, each with varying bids. When requests arrive sequentially, ad sellers must decide immediately whether to accept or reject them without knowing future requests. Accepted requests can later be canceled (or "bought back") at a cost proportional to their value, representing contract cancellation penalties. This buyback mechanism allows sellers to revise earlier commitments to accommodate more valuable later requests, optimizing revenue. \cite{BHK-09} studied algorithms for a simplified version of this problem where there is only a single slot for banner ads. However, it is often the case that there are multiple webpages that can accommodate banner ads, with \revcolorCAMtwo{heterogeneous} features and compatibilities with the advertisers' targeted ad. In addition, ad exchange platforms often work with multiple publishers that sell their banner ads. Our framework generalizes the setting in \cite{BHK-09}: it supports multiple heterogeneous banners across publishers, each with distinct compatibilities to advertiser bids. The mechanism we propose can serve as a backend module that optimally matches arriving bids to multiple ad slots while accounting for the cost of revoking earlier allocations.\footnote{\revcolorCAMtwo{We note that \citet{feige2008combinatorial} study an offline counterpart to our model in which the platform may leave a demand request partially unsatisfied at a penalty proportional to its value. Although conceptually similar, their model is distinct from ours, even when our problem is restricted to the offline setting.}}

\revcolorCAMtwo{Allocation decisions in this application are naturally integral: each placement in a banner slot is sold to a single advertiser. In most practical settings, however, each banner accommodates many such commitments over a campaign, so capacities are large. In this regime, our fractional model provides a tractable relaxation. As shown in \Cref{apx:matching rounding}, an online fractional algorithm can be converted into a randomized integral policy with a loss that vanishes as capacity grows. Thus, for this application, we can combine our fractional algorithm with randomized rounding to obtain the final algorithm.}

\revcolorCAMtwo{To summarize the mapping to our model, the offline nodes are the banner slots (across webpages and publishers), the online nodes are the advertisers' arriving requests, and the edge weight $\weightij$ is the bid of advertiser $i$ for banner $j$. This edge weight genuinely depends on both sides, since placements differ in traffic, audience demographics, and creative fit. A cancellation revokes a previously accepted request. The buyback cost $\buybackcost\cdot\weightij$ is the associated contractual cancellation penalty, set as a percentage of the value of the revoked contract, which is the common compensation mechanism in this application~\cite{BHK-09}. Finally, when advertisers request \emph{bundles} of several banners simultaneously (e.g., placements on multiple webpages at once), the resulting problem is captured by our configuration-allocation extension (\Cref{sec:configuration}).}

\smallskip
\revcolorCAMtwo{\paragraph{(ii) Display advertising.} Costly cancellations also arise in display advertising. An ad-exchange platform (such as Google AdX) assigns impressions arriving from different publishers' websites to advertisers that bid for them.\footnote{\revcolorCAMtwo{This application differs from the banner advertising setting discussed above in several important ways. Most importantly, here impressions arrive online and are assigned to advertisers, who pay for them. By contrast, in banner advertising, advertisers arrive online and are assigned to banner slots during the selling horizon (i.e., the advertising campaign); the winning advertisers then use the banners for a fixed period.}} Advertisers typically have daily contracts that specify the maximum number of impressions they may purchase, and they pay for at most that number of impressions rather than for every impression they receive. Because these daily limits cover large volumes of impressions, advertiser capacities are large and the allocation process is effectively divisible. Display advertising is therefore a natural application of our fractional (large-capacity) matching regime. In our model, the offline nodes are the advertisers, which persist throughout the horizon and each have a large daily capacity of impressions; the online nodes are the impressions, which arrive and perish; and the edge weight $\weightij$ is advertiser $j$'s bid for impression $i$, reflecting how well the impression matches the advertiser's targeting criteria. Under these agreements, the allocation algorithm may initially assign impressions beyond an advertiser's committed capacity but ultimately receives payment only for the highest-valued impressions within that limit. This setting is often modeled using the \emph{free-disposal} assumption~\citep{FKMMP-09,DHKMY-16}\revcolorCAMtwo{---the special case $f = 0$ of our framework---which allows the online algorithm to discard previously assigned bids at no cost to make room for higher bids.}}

\revcolorCAMtwo{To mitigate the issue with lost/unpaid impressions under free disposal---which results from the economic structure of advertiser contracts---ad-exchange platforms occasionally use \emph{costly disposal}~\citep{FKMMP-09,DHKMY-16}
. This model allows an ad-exchange platform to quantify and control cancellations, thereby balancing immediate revenue maximization against the negative externalities and adverse effects of excess impression allocations. Echoing the discussion at the beginning of this section, the cost of disposal can have either an \emph{implicit} or an \emph{explicit} interpretation:

\begin{itemize}
    \item \textbf{Implicit interpretation.} Under traditional contracts, advertisers do not pay for excess impressions. This lack of payment can impose implicit costs on the platform. These impressions could instead have been allocated through other advertising channels or to an underallocated advertiser; assigning them to an overallocated advertiser, and thereby triggering disposal, therefore creates an opportunity cost. Moreover, advertisers may interpret excess supply as a reason to lower their bids, which can distort long-term bidding incentives and erode trust in delivery commitments. In this case, the buyback model with $\buybackcost>0$ serves as a \emph{control mechanism} for these implicit costs. By optimizing the mixed-sign objective---the total reward of the finally allocated impressions (i.e., total collected revenue) minus $\buybackcost$ times the total reward of the canceled impressions (i.e., total lost revenue)---the platform accounts for the opportunity cost associated with unpaid impressions\footnote{\revcolorCAMtwo{Specifically, consider an outside option (such as a fallback ad network) that is willing to purchase unlimited unassigned inventory at an \emph{expected} rate of $f \cdot w$.}} and limits cancellations to those justified by sufficiently more valuable arrivals.
    \item \textbf{Explicit interpretation.} Consider instead a slightly different contract under which advertisers are willing to accept a ``soft cap'' on their budgets, meaning that they pay the full rate of $1\times$ their bid value for impressions up to a specified capacity and a \emph{discounted rate} of $-f$ for impressions beyond that capacity, rather than paying nothing for them. Revenue maximization under these contracts (and not maximizing the mixed-sign objective) becomes an application of our model with a negative buyback factor $f\in[-1,0)$. The advertisers' commitments are still respected: they pay the full rate for the higher-value impressions within their capacity and the lower rate for lower-value excess impressions. The ad-exchange platform, in turn, benefits from this additional allocation option.
\end{itemize}}
 
 \smallskip
\revcolorCAMtwo{
\paragraph{(iii) Preemptible broadcast advertising.} The sale and scheduling of broadcast commercial airtime is a classical topic in revenue management~\citep{bollapragada2004scheduling,PDJ-17}. In particular, quite similar to banner ads, selling advertising airtimes in broadcast television and radio could provide an instance of our model. In this application, stations own various airtimes for advertising, where each airtime could be assigned to one commercial break.  Stations sell airtime in classes distinguished by preemption protection---the FCC recognizes non-preemptible, preemptible with notice, immediately preemptible, and run-of-schedule as distinct classes of time \citep{FCC-73.1942}. Focusing on selling only preemptible airtimes, stations typically sell them off during a selling campaign, where advertisers arrive over time and bid for these airtimes. Once a bid receives, the station provisionally assigns the advertiser to one of the airtime spots, but this assignment can change during the campaign. In fact, in this application, an immediately preemptible spot is, by its published definition, one that may be preempted in favor of an advertisement purchased at a higher rate. This value-upgrading revocation is exactly the buyback mechanism in our model.

Mapping to our model, the offline nodes are commercial breaks, the online nodes are advertiser orders arriving over the selling window, and edge weights are order-specific rates that differ across breaks because programs and dayparts differ in audience. Because a commercial break holds only a handful of thirty-second units,
capacities are small and large-capacity rounding does not apply; guarantees
valid at unit capacity---our deterministic integral results
(\Cref{sec:deterministic})---are the relevant ones here.

Preemption in this industry is settled through make-goods \citep{CJAdvertising-2023}: replacement airtime, typically at the station's discretion in a later period. A make-good returns a fraction of the displaced commitment's value rather than extinguishing it, which is precisely the two-tier structure we analyze under negative buyback factors with $f\in[-1,0]$ (see \Cref{sec:intro} and \Cref{apx:negative buyback cost}): the primary tier is the booked position, and the discounted secondary tier is the make-good. The buyback factor, however, depends on how much spare inventory the station has: when unsold airtime is plentiful a make-good is nearly costless, but when inventory is scarce, the replacement spot must displace paying demand, so the make-good costs more than the displaced commitment returns and the relevant regime moves to positive $f$.}

\revcolorCAMtwo{
\subsubsection{Overbooking in revenue management}
\smallskip
\paragraph{} Overbooking is a common practice in the hotel and airline industries, used to protect against revenue losses arising from no-shows and from committing capacity to lower-value, early-arriving customers~\citep{LY-78,rot-71}. Although the literature primarily focuses on stochastic attrition (i.e., no-shows), which is the first-order motivation for overbooking, our model isolates a distinct operational force that may arise simultaneously: value-enhancing revocation. Admittedly, this is not a first-order consideration, but it may still generate additional revenue improvements. Even after available capacity is fully booked, decision-makers may continue to accept premium demand by bumping lower-value commitments. Continuing to sell premium inventory into an oversold cabin while displacing discount fares is therefore an example of the value-upgrading cancellations suggested by our buyback model.}\footnote{\revcolorCAMtwo{Importantly, such firm-initiated displacement is neither uncontrolled nor necessarily involuntary: U.S.\ regulations require carriers first to solicit volunteers willing to relinquish their reservations in exchange for compensation before resorting to involuntary displacement. Moreover, the compensation owed does not depend on \emph{why} the flight was oversold~\citep{USDOT2011}.}} \revcolorCAMtwo{We also note that monetary compensation expressed as a percentage of the booking value is common in the airline industry. For example, U.S.\ Department of Transportation regulations mandate compensation when passengers are involuntarily bumped---typically 200--400\% of the one-way fare, capped at \$1,075--2,150 depending on the delay duration and route~\citep{USDOT2011}. This corresponds to $f \in [2,4]$, well within our large-buyback regime, and the penalty is exactly linear in $w$ for fares below the cap. Analogous dynamics arise in hotel revenue management, where upsell products or rooms may be overbooked and subsequently assigned to the highest-value requests. Similar compensation rules are occasionally used when upgrade products are overbooked in hotel revenue management~\citep{GBU2023personal,NorOne,Booking.com}. As in the airline setting, hotel overbooking is primarily used to counteract no-shows; however, we believe that it could also be used to generate further revenue improvements.\footnote{\label{footnote:overbooking-disclaimer}\revcolorCAMtwo{The main purpose of discussing this application is to provide additional context for settings in which the buyback model fits to some extent; we do not claim that this application is a direct fit for our setting.}}}

\smallskip
\revcolorCAMtwo{
Revenue management for flexible and rollable cargo capacity provides another example of overbooking that resembles several aspects of our model. 
In this setting, displacement is an explicit---and explicitly priced---feature of the product being sold and carriers routinely oversell capacity. As a result,  the trade-off between spoilage costs and \emph{offloading} costs---the costs of bumping booked shipments at departure---is a central object of study in the cargo revenue management literature~\citep{kasilingam1997air,ACG-07}. In ocean freight, which is the subject of a growing revenue management literature~\citep{MZW-19}, carriers overbook containership capacity~\citep{FKH-22} and roll excess containers to later sailings when a vessel reaches capacity~\citep{FlexportRolled}.  For example, under Maersk's spot-booking terms, a shipper may select a ``Rollable'' option, in which case the carrier may move the cargo to a later sailing and pay the shipper a prespecified, lane-dependent compensation fee for each rolled container~\citep{MaerskRollable,MaerskRollableFAQ}.

In the language of our model, the offline nodes are sailings (vessel--departure combinations, with capacity measured in container slots), the online nodes are booking requests arriving over the selling horizon, and the edge weight $\weightij$ represents the freight contribution of booking $i$ on a compatible sailing $j$, where compatibility is determined by origin--destination, equipment type, and delivery window. Rolling a flexible booking constitutes a buyback: it releases capacity on the current sailing for a more valuable booking at a cost equal to the contractual compensation owed to the rolled shipper. Because rolled cargo is typically re-accommodated on a later sailing---and the displaced commitment therefore retains part of its value through delayed service---this application also connects naturally to our negative-$\buybackcost$ extension. 

As in the airline and hotel applications above, we emphasize that overbooking in this industry is driven largely by the need to protect against widespread shipper no-shows. However, we believe that, once capacity on the current sailing becomes binding, reallocating it toward higher-value bookings can also serve as a mechanism for further improving revenue and hedging against commitments to lower-value, early-arriving customers.\footnote{\revcolorCAMtwo{This practice indeed drew regulatory scrutiny during the recent capacity crunch, when carriers rolled contracted cargo while pursuing higher-rate bookings (see, e.g., the Federal Maritime Commission proceedings in \citealp{FMC-23-02}).}}
}

\smallskip
\revcolorCAMtwo{
\subsubsection{Interpreting the allocation regimes and the buyback factor.} 
\smallskip
\paragraph{}We close this section with two remarks that apply across the applications above. First, on allocation regimes: allocations are effectively divisible whenever each offline resource pools a large number of units---impressions in display advertising, slots over the course of a campaign in banner advertising, container slots on a vessel in cargo shipping. This is the large-capacity regime in which our fractional algorithms are the operationally meaningful model, and in which they can be converted into randomized integral policies with vanishing loss (\Cref{apx:matching rounding}). In contrast, when capacities are genuinely small---such as a commercial break holding a handful of thirty-second units in broadcast advertising---our deterministic integral results (\Cref{sec:deterministic}), whose guarantees hold even at unit capacities, are the relevant ones. Second, on the buyback factor: while the linear cancellation cost $\buybackcost\cdot \weight$ can always be read as a reduced-form proxy for the various costs that a revoked commitment imposes (recall the discussion in \Cref{apx:explicit-implicit-cost} on explicit vs. implicit cost of cancellations), in several of the settings above it is a literal description of practice, because penalties are set \emph{ad valorem}, that is, as a percentage of the value of the displaced commitment. Examples include denied-boarding compensation equal to 200--400\% of the one-way fare~\citep{USDOT2011}
and host-cancellation fees on short-term rental platforms, which equal 10--50\% of the reservation value depending on the notice given, i.e., $\buybackcost\in[0.1,0.5]$~\citep{AirbnbHostFee}

We should also note that we use the same factor $f$ for all cancellations. In applications where the cancellation cost is instead specific to the revoked request or to a contract class, the common-factor objective $\buybackcost\cdot\weight$ serves as a parsimonious benchmark that isolates how the overall degree of commitment shapes optimal online policies; enriching the model with request-specific cancellation costs is an interesting direction for future research.}
}

\section{Remaining Technical Details in Section~\ref{sec:single-resource}}
\label{sec: apxsec3}

\revcolorCMTwo{In this section, we present an alternative approach---again based on the primal-dual framework---for designing and analyzing optimal competitive fractional online algorithms in the single-resource environment (cf.\ \Cref{remark:alternative-approach} in \Cref{sec:single-resource}). This section also collects the remaining technical details related to \Cref{sec:single-resource}: the connection between our primal-dual framework and the classic randomized algorithm of \citet{AK-09} (\Cref{app:BK implementation}), and the extension of our results to non-uniform demands and supplies (\Cref{app:non-uniform demand}).}

\subsection{\revcolorCMTwo{Overview of the alternative approach}}
\label{app:alt-overview}
\revcolorCMTwo{In contrast to \Cref{sec:single resource primal dual}, where we proposed a parametric family of primal-dual algorithms that runs directly on general instances, the alternative approach presented in this section identifies a family of optimal competitive primal-dual algorithms by first restricting attention to a special class of instances. This approach consists of the following steps:}

\begin{enumerate}[leftmargin=14pt,label=(\roman*)]
\smallskip
    \item \revcolorCMTwo{\textit{Reduction to a continuum of weights (\Cref{sec:reduction-to-continuum}):}
    Focusing on fractional allocations, we show that the general problem---where online nodes arrive with arbitrary weights---can be reduced to a special case in which a spectrum of weights (infinitely many online nodes) arrive continuously  until an unknown stopping point. We refer to this special instance as a ``truncated weight continuum'' (formally, \Cref{def:weight-continuum}).}

\smallskip
    \item \revcolorCMTwo{\textit{Identifying structural properties of good algorithms (\Cref{sec:single-properties}):}
    For the restricted class of truncated weight continuum instances, we identify key natural properties that the allocation distribution of an effective online algorithm should satisfy. These properties reveal the ``correct'' structural form of a desired algorithm, and help us to characterize its allocation distribution under truncated weight continuum instances.}

\smallskip
    \item \revcolorCMTwo{\textit{Fractional primal-dual algorithms for weight continuum instances (\Cref{app:single-cont-primal-dual}):}
    Using the primal linear program of the optimal offline benchmark and its dual, along with a carefully constructed dual solution, we propose a family of primal-dual algorithms satisfying the identified structural properties. Then, through a novel dual fitting analysis, we show that appropriate parameter selection yields an optimal competitive fractional online algorithm.}

\smallskip
    \item \revcolorCMTwo{\textit{Characterizing allocation distributions (\Cref{app:single-resource-characterization}):}
    Finally, under a mild regularity assumption, we exactly characterize the allocation distributions induced by online algorithms satisfying the identified structural properties when executed on truncated weight continuum instances, and precisely determine their competitive ratios. This characterization identifies a single-parameter class of allocation distributions that includes the optimal competitive algorithm, and leads to an alternative optimality proof that applies to a broader class of primal-dual algorithms (beyond generalized exponential penalty functions).}
\end{enumerate}

\subsection{\revcolorCMTwo{Reduction to truncated weight continuum instances}}
\label{sec:reduction-to-continuum}
\revcolorCMTwo{Central to the proof of the lower bound in \Cref{prop:lower bound single resource} (see \Cref{sec:apx-missing lower bound single resource}) is a special instance in which a spectrum of weights arrives continuously. This special class of instances, denoted by $\continstances \triangleq \{\instance_\totaltime\}_{\totaltime \in \reals_+}$ and formally defined below, plays a crucial role in our algorithmic development.}

\begin{definition}[Truncated Weight Continuum Instance]
\label{def:weight-continuum}
\revcolorCMTwo{Given $\totaltime \in \reals_+$, the \emph{truncated weight continuum instance $\instance_\totaltime \in \continstances$} consists of a continuum of online nodes $\onlinenodes = [0, \totaltime]$, where each node indexed by $\weight \in \onlinenodes$ arrives continuously at time $\weight$ with weight equal to $\weight$. The arrivals conclude at time~$\totaltime$.}
\end{definition}

\revcolorCMTwo{Now consider a general adversarial instance with arriving weights $\weight_1, \weight_2, \ldots$ and let $\weightmax \triangleq \max_{i} \weight_i$. Suppose we have access to an online fractional algorithm $\ALG$ that achieves competitive ratio $\Gamma$ for the restricted class of instances $\continstances$. As the first step in our approach, roughly speaking, we show that by tracking $\weightmax$ online and dynamically feeding \emph{proxy online nodes} to $\ALG$ from the evolving truncated weight continuum instance $I_{\weightmax}$, we can convert the decisions of $\ALG$ into a feasible fractional online algorithm with weakly higher net reward, thus achieving the same competitive ratio $\Gamma$ for general instances. See \Cref{alg:algorithm reduction single resource} for details. The following proposition formally establishes this blackbox reduction from arbitrary instances to $\continstances$.}

\begin{restatable}{proposition}{thmfractionalreduction}
\label{prop:algorithm reduction single resource}
\revcolorCMTwo{In the single-resource environment, for any online algorithm $\ALG$ achieving competitive ratio $\approxratio$ within the truncated weight continuum instances $\continstances$, \Cref{alg:algorithm reduction single resource} (using $\ALG$ as input) is a feasible fractional online algorithm and achieves the same competitive ratio $\approxratio$ within all instances $\instances$.}
\end{restatable}

\begin{algorithm}[hbt]
\caption{\revcolorCMTwo{Reduction from general instances to $\continstances$ for fractional algorithms (single-resource)}}
\revcolorCMTwo{
\label{alg:algorithm reduction single resource}
    \SetKwInOut{Input}{Input}
    \SetKwInOut{Output}{Output}
    \Input{Fractional online algorithm $\ALG$ for truncated weight continuum instances $\continstances$}

    \vspace{1mm}

    Initialize $\weightmax \gets 0$.

    \vspace{1mm}

    \For{each online node $i \in \onlinenodes$}{

        \vspace{1mm}

        \If{$\weighti > \weightmax$}{
            \vspace{1mm}

            \For{each proxy weight $\weight \in (\weightmax, \weighti]$}{

                \vspace{1mm}

                Feed a proxy online node $\weight$ to  $\ALG$ (as arriving input in a weight continuum instance).

                \vspace{1mm}

                \If{$\ALG$ allocates fraction $\alloc$ to the proxy online node $\weight$}{

                    \vspace{1mm}

                    Allocate fraction $\frac{\weight}{\weighti}\alloc$ to actual online node $i$.

                    \vspace{1mm}

                    Preserve capacity by greedily buying back from the smallest weights (if necessary).
                }
            }

            \vspace{1mm}

            Update $\weightmax \gets \weighti$.
        }
    }
    }
\end{algorithm}



\revcolorCMTwo{\begin{proof}{\emph{Proof of \Cref{prop:algorithm reduction single resource}.}}
Our proof proceeds in two steps. First, we introduce a variant of our model in which the decision maker can choose a \emph{discounted price} (at most the arriving weight) for each allocation. Specifically, when a node arrives with weight $\weight_i$, the algorithm can allocate at any discounted price $\weight \leq \weight_i$, rather than being restricted to exactly $\weight_i$. The reward and future buyback costs are then calculated based on $\weight$ instead of $\weight_i$.  We show that algorithms designed for truncated weight continuum instances can easily adapt to this discounted allocation model. In the second step, we establish an equivalence between the base model (the discrete-arrival model defined in \Cref{sec:prelim}) and the discounted allocation model using \Cref{alg:algorithm reduction single resource}.  This algorithm proportionally maps allocations from discounted price $\weight$ to weight $\weight_i$, allocating a fraction $\frac{\weight}{\weight_i} x$ at $\weight_i$ when a fraction $x$ was allocated at discounted price $\weight$. 

\smallskip
\noindent\emph{\underline{Step 1} - Reduction from general instances in the discounted-allocation model.} Consider the following variant of edge-weighted online matching with buyback.
\begin{definition}[Discounted-Allocation Model]
\label{def:lemon juice model}
In the \emph{discounted-allocation model}, once an online node $i\in\onlinenodes$ with edge weights $\{\weightij\}_{j\in\offlinenodes}$ arrives, for each $dx$ amount of offline node~$j$ allocated to online node $i$, the decision maker also immediately and irrevocably specifies a discounted per-unit price $\price$ which is weakly smaller than edge weight $\weightij$, and collects $\price\cdot  dx$ amount of reward. The buyback cost depends linearly on the allocated per-unit prices. In particular, when the decision maker buys back $dx$ amount of allocation with previously specified per-unit discounted price $\price'$, she loses the reward $\price'\cdot dx$ and pays an extra cost of $\buybackcost\cdot \price'\cdot dx$ due to buybacks.
\end{definition}

It is straightforward to see that the \revcolorCMTwo{offline optimum} benchmark is the same for the base model and the discounted-allocation model, since the \revcolorCMTwo{offline optimum} never buys back and always sets discounted per-unit prices equal to edge weights. Similarly, every online algorithm in the base model is well-defined in the discounted-allocation model (by setting discounted per-unit prices equal to the edge weights). 

Now we present the main claim in the first step of our argument.
\begin{lemma}
\label{claim:reduction from discounted allocation model}
In the single-resource environment, for any online algorithm $\ALG$ with the competitive ratio $\approxratio$
within the truncated weight continuum instances $\continstances$ in the base model,
there exists an online algorithm $\ALG\primed$ that has a competitive ratio of  $\approxratio$ within the class of all instances $\instances$ in the discounted-allocation model.
\end{lemma}
\begin{proof}{\emph{Proof.}}
The key idea is a simulation argument. Fix any online algorithm $\ALG$ for the base model on instances from $\continstances$, and consider an arbitrary instance $\instance \in \instances$ with online weights $\weight_1 \leq \weight_2 \leq \dots \leq \weight_\totaltime$.\footnote{If the weights are not sorted, we can define a non-decreasing sequence $\weight_i' = \max_{\ell \leq i} \weight_\ell$ and apply the same construction.} For each online node $i \in [\totaltime]$, we define a weight-continuum segment over the interval $[\weight_{i-1}, \weight_i]$ (with $\weight_0 := 0$). We simulate the behavior of $\ALG$ on this truncated instance, denoted by $\instance_{\weight_i}$, and interpret its decisions as guidance for the algorithm $\ALG\primed$. Let $\alloc(\weight)$ be the allocation density function used by $\ALG$ over the continuum $[\weight_{i-1}, \weight_i]$. That is, for each infinitesimal weight $w$ in this interval, $\ALG$ allocates $\alloc(\weight)\,\dd\weight$ units of the resource at per-unit price $\weight$. In the discounted-allocation model, the algorithm $\ALG\primed$ emulates this behavior by choosing the same allocation density using discounted per-unit prices over the interval $[\weight_{i-1}, \weight_i]$. In particular $\ALG\primed$ sells $\alloc(\weight)\,\dd \weight$ at price $\weight \in [\weight_{i-1}, \weight_i]$. Since this model allows allocating any fraction of the resource at a per-unit price strictly less than the online node's weight, such simulation is feasible. The buyback decisions of $\ALG$ can also be mimicked by discounting previous allocations accordingly.

By construction, the total profit realized by $\ALG\primed$ on the original instance $\instance$ is identical to the profit earned by $\ALG$ on the truncated instance $\instance_{\weight_\totaltime}$. Moreover, the offline benchmarks for both models coincide under this construction, so the competitive ratio is preserved which completes the proof.
\hfill\halmos
\end{proof}

\smallskip
\noindent\emph{\underline{Step 2} - Reduction from general instances in the base model.} Recall that every online algorithm in the base model is also well-defined in the discounted-allocation model (by setting discounted per-unit prices equal to the edge weights). In this step, we argue that the reverse also holds, and thus two models are equivalent.

\begin{lemma}
\label{prop:lemon juice}
For any online algorithm $\ALG$ in the base (resp.\ discounted-allocation) model,
there exists an online algorithm $\ALG\primed$ in the 
discounted-allocation (resp.\ base) model
such that for any instance $\instance\in\instances$,
the expected profit in $\ALG$
is the same as the expected profit in $\ALG\primed$,
i.e., $\ALG(\instance) = \ALG\primed(\instance)$.
\end{lemma}
\begin{proof}{\emph{Proof.}}
Fix any online algorithm $\ALG$ in the base \revcolorCMTwo{model}; 
note that 
it is by definition a valid online algorithm in the discounted-allocation model.
Fix any online algorithm $\ALG$ in the discounted-allocation model. We present the following construction of online algorithm $\ALG\primed$ in the base model. For each online node $i\in\onlinenodes$ and offline node $j\in\offlinenodes$, let $P_{ij}$ be the set of discounted prices used in algorithm $\ALG$ for this pair of online node $i$ and offline node $j$. To simplify the notation, we assume $P_{ij}$ is a finite set.\footnote{The construction and analysis can be extended straightforwardly when $P_{ij}$ is an infinite set.}
Suppose $z_{ij\price}$ fractional units is allocated from offline node $j$ to online node $i$ at discounted price $\price\in P_{ij}$ in algorithm $\ALG$. We then let the constructed algorithm $\ALG\primed$ allocate $\frac{\price}{\weightij}\cdot z_{ij\price}$ fractional units of offline node $j$ to online node $i$. Similarly, suppose algorithm $\ALG$ buys back $\zeta_{ij\price}$ fractional units of offline node $j$ from online node $i$ at discounted price $\price\in P_{ij}$. We then let the constructed algorithm $\ALG\primed$ buy back $\frac{\price}{\weightij}\cdot \zeta_{ij\price}$ fractional units of offline node $j$ from online node $i$.

Since algorithm $\ALG$ is valid, it must satisfy that: after the departure of each online node $i'\in\onlinenodes$, for each offline node $j\in\offlinenodes$, the total allocations  minus buybacks is at most one in algorithm $\ALG$, i.e., 
\begin{align*}
    \sum_{i\in\onlinenodes:i\leq i'}
    \sum_{\price\in P_{ij}}
    z_{ij\price} - \zeta_{ij\price}
    \leq 1
\end{align*}
and for each online node $i\leq i'$, offline node $j\in\offlinenodes$, and discounted price $\price\in P_{ij}$, the amount of buybacks is at most the amount of allocation in algorithm $\ALG$, i.e.,
\begin{align*}
    \zeta_{ij\price} \leq z_{ij\price}
\end{align*}
Note that by construction, after the departure of each online node $i'\in\onlinenodes$, for each offline node $j\in\offlinenodes$, the total allocations minus buybacks in algorithm $\ALG\primed$ is
\begin{align*}
    \sum_{i\in\onlinenodes:i\leq i'}
    \sum_{\price\in P_{ij}}
    \frac{\price}{\weightij}z_{ij\price} - \frac{\price}{\weightij}\zeta_{ij\price}
    \leq 
    \sum_{i\in\onlinenodes:i\leq i'}
    \sum_{\price\in P_{ij}}
    z_{ij\price} - \zeta_{ij\price}
    \leq 1
\end{align*}
where the first inequality holds since $\price \leq \weightij$ for every $\price\in P_{ij}$ and $z_{ij\price} \geq \zeta_{ij\price}$ argued above. Similarly, it can be straightforwardly verified that for each online node $i\leq i'$ and offline node $j\in\offlinenodes$, the amount of buybacks is at most the amount of buybacks in algorithm $\ALG\primed$, since
\begin{align*}
    \max_{i\in\onlinenodes:i\leq i'}\max_{\price \in P_{ij}}
    \frac{\frac{\price}{\weightij}\zeta_{ij\price}}{\frac{\price}{\weightij}z_{ij\price}}
    \leq 1
     \quad \quad
    \Longrightarrow
     \quad \quad
    \sum_{i\in\onlinenodes:i\leq i'}
    \sum_{\price\in P_{ij}}
    \frac{\price}{\weightij}\zeta_{ij\price}
    \leq
    \sum_{i\in\onlinenodes:i\leq i'}
    \sum_{\price\in P_{ij}}
    \frac{\price}{\weightij}z_{ij\price} 
\end{align*}
where the first inequality holds since $z_{ij\price} \geq \zeta_{ij\price}$ argued above. Therefore, the constructed algorithm $\ALG\primed$ is also valid.

Finally, we verify that both algorithms have the same profit. Note that
\begin{align*}
    \ALG(\instance) &= \sum_{i\in\onlinenodes}
    \sum_{j\in\offlinenodes}
    \sum_{\price\in P_{ij}}
    \price \cdot z_{ij\price} - (1 + \buybackcost)\price\cdot \zeta_{ij\price}
    = \sum_{i\in\onlinenodes}
    \sum_{j\in\offlinenodes}
    \sum_{\price\in P_{ij}}
    \price \cdot (z_{ij\price} - (1 + \buybackcost) \zeta_{ij\price})
    \\
    \ALG\primed(\instance) &=
    \sum_{i\in\onlinenodes}
    \sum_{j\in\offlinenodes}
    \sum_{\price\in P_{ij}}
    \weightij\cdot \frac{\price}{\weightij}z_{ij\price} - (1+\buybackcost)\weightij\cdot \frac{\price}{\weightij}\zeta_{ij\price}
    =\sum_{i\in\onlinenodes}
    \sum_{j\in\offlinenodes}
    \sum_{\price\in P_{ij}}
    \price \cdot (z_{ij\price} - (1 + \buybackcost) \zeta_{ij\price})
\end{align*}
which concludes the proof as desired.
\hfill
\halmos
\end{proof}

Combining \Cref{claim:reduction from discounted allocation model,prop:lemon juice}, we finish the proof of \Cref{prop:algorithm reduction single resource}. 
Furthermore, it can be verified that \Cref{alg:algorithm reduction single resource} is in fact the algorithmic construction in the proof of \Cref{claim:reduction from discounted allocation model,prop:lemon juice}.
\hfill
\halmos
\end{proof}

}

\subsection{\revcolorCMTwo{Natural structural properties of allocation distributions}}
\label{sec:single-properties}
\revcolorCMTwo{We now restrict our attention to the class of truncated weight continuum instances $\continstances$ and investigate the allocation distributions induced by fractional online algorithms when we run them on instances in this class. We hypothesize two natural properties that an optimally competitive online algorithm's allocation distribution should satisfy---or, more generally, properties that any candidate algorithm aspiring to be \revcolorCMTwo{optimally} competitive would naturally exhibit.

To formally define these properties, let $\callocwpsup(\cdot)$ denote the allocation quantile function $\calloc(\cdot)$ of a given online fractional algorithm at the moment immediately after processing the online node $\weightprimed$ on a truncated weight continuum instance $\instance_{\totaltime}$ with $T > \weightprimed$. Moreover, define $\buybackweightwpsup \triangleq \inf\{\weight \in \reals_+ : \callocwpsup(\weight) < 1\}$
as the smallest allocated weight at this moment.\footnote{Equivalently, $\callocwpsup(\cdot)$ (resp., $\buybackweightwpsup$) can be viewed as the allocation quantile function $\calloc(\cdot)$ (resp., the smallest allocated weight) at the end of processing the truncated continuum instance $\instance_{\weightprimed}\in\continstances$.}
Without loss of generality, we focus on instances with truncation points $T \geq 1$. We then define the following properties.}

\begin{definition}[Scale Invariance]
\label{def:SI}
\revcolorCMTwo{An online algorithm satisfies the \emph{Scale Invariance (\texttt{SI})} property if its allocation quantile function $\calloc(\cdot)$ satisfies
\begin{align}
\tag{\textsc{SI}}
\label{eq:allocation invariant property}
   \forall \weightprimed,\weight \in \reals_+: \quad \callocwpsup(\weight) =
    \calloconesup\left(\frac{\weight}{\weightprimed}\right).
\end{align}}
\end{definition}
\begin{definition}[Greedy buyback]
\label{def:GB}
\revcolorCMTwo{An online algorithm satisfies the \emph{Greedy Buyback (\texttt{GB})} property if its allocation quantile function $\calloc(\cdot)$ satisfies
\begin{align}
\tag{\textsc{GB}}
\label{eq:greedily allocation property}
   \forall\weightprimed\in[\buybackweightonesup, 1],\weight
\in [\buybackweightonesup, \weightprimed]:~~~ \calloconesup(\weight) -\callocwpsup(\weight) =
    1 - \callocwpsup\left(\buybackweightonesup\right)
\end{align}}
\end{definition}

\revcolorCMTwo{To explain the first property (\Cref{def:SI}), we start with an intuitive observation: from the perspective of an online algorithm, all truncated weight continuum instances $\instance_\totaltime\in\continstances$ appear identical before the last node $\totaltime$ arrives. Moreover, any instance $\instance_\totaltime$ can be viewed as a scaled (stretched or compressed) version of another instance $\instance_{\totaltime'}$ by rescaling arrival times (equivalently, weights) by a factor of $\frac{T}{T'}$. Therefore, an optimally competitive algorithm should treat these instances equivalently after appropriate rescaling. This intuition is precisely captured by property~\eqref{eq:allocation invariant property} in \Cref{def:SI}: the allocation quantile function $\callocwpsup(\cdot)$ after processing node $\weightprimed$ is a scaled version of the quantile function $\calloconesup(\cdot)$ after processing node $1$, with scaling factor~$\weightprimed$.

The second property (\Cref{def:GB}) follows directly from a simple observation: whenever an online algorithm must buy back to free capacity, it is always optimal---given any future sequence of online arrivals---to greedily buy back allocation from previously allocated online nodes with the \emph{smallest weight}. Now consider the interval between the arrivals of online nodes $\weightprimed \in [\buybackweightonesup, 1]$ and node $1$. First, the definition of $\buybackweightonesup$ implies that the algorithm has already revoked all allocations at weights below $\buybackweightonesup$, totaling a mass of $1 - \callocwpsup(\buybackweightonesup)$, and has reallocated this mass to newly arrived online nodes in $[\weightprimed, 1]$. Second, due to the greedy buyback strategy, no allocations within the interval $[\buybackweightonesup, \weightprimed]$ were revoked when nodes in $[\weightprimed, 1]$ arrived. Together, these facts imply that the quantile function $\callocwpsup(\weight)$ increases uniformly for all $\weight$ in the interval $[\buybackweightonesup, \weightprimed]$ by $1 - \callocwpsup(\buybackweightonesup)$ to obtain $\calloconesup(\weight)$, precisely capturing property~\eqref{eq:greedily allocation property} in \Cref{def:GB}.

So far, we have identified two structural properties that allocation distribution of ``good'' online algorithms should satisfy under truncated weight continuum instances: the scale invariance~\eqref{eq:allocation invariant property} and the greedy buyback ~\eqref{eq:greedily allocation property}. Both properties appear to be necessary for any optimally competitive online algorithm. Now, a natural question arises: \emph{Can we identify a simple, interpretable family of algorithms satisfying these properties?}

\smallskip
Ideally, such algorithms should not only satisfy our natural properties, but also provably include the optimal competitive algorithm and easily extend to more complex settings (e.g., the matching environment). Next, we provide an affirmative answer to this question, demonstrating that these desirable features indeed hold.}

\subsection{\revcolorCMTwo{A family of optimal competitive primal-dual algorithms for weight continuum instances}}
\label{app:single-cont-primal-dual}
\revcolorCMTwo{Mirroring the design of the direct primal-dual algorithm in the main body (\Cref{alg:primal-dual-direct-single} in \Cref{sec:single resource primal dual}), we now propose a parametric family of primal-dual algorithms for weight continuum instances that satisfy the two natural structural properties introduced in \Cref{sec:single-properties}. The core idea is simple: for a given instance $I_T\in\continstances$, consider the following (almost trivial) linear program and its dual, which together characterize the offline optimal allocation:
\begin{align*}
\tag{$\mathcal{P}_{\textsc{OPT-single-cont}}$}
\label{eq:primal-dual-single-cont}
\arraycolsep=1.4pt\def\arraystretch{1}
\begin{array}{lllllll}
\max~~~~&\displaystyle \int_{0}^\totaltime\weight\cdot \alloc(\weight)\,\dd \weight&~~\text{s.t.} & & \quad\quad\quad\quad \text{min}~~~~ &\displaystyle\offlinedual&~~\text{s.t.} \\[1.2em]
 &\displaystyle\int_{0}^{\totaltime}\alloc(\weight)\,\dd \weight \leq 1& &
& &\offlinedual\geq \weight~,&~~\weight\in[0,T]\\[1.2em]
&\alloc(\weight)\geq 0~,&~~\weight\in[0,T]\qquad\qquad & & & \offlinedual\geq 0&
\end{array}
\end{align*}
Our primal-dual algorithm maintains a feasible dual assignment $\offlinedual$ for the above linear program throughout its execution. The algorithm is parameterized by a monotone increasing penalty function $\pen:[0,1]\rightarrow\reals_+$, satisfying $\pen(0)=0$ and $\pen(1)\geq 1$. Given this penalty function, the dual assignment is directly related to the current allocation quantile function exactly as in \eqref{eq:dual assignment}, that is, $\offlinedual = \int_0^\infty
    \pen\left(\calloc(\weight)\right)\,\dd\weight$.
To maintain dual feasibility, when a new online node $\weightprimed \in \reals_+$ arrives, the algorithm continuously allocates to this node (greedily buying back from the smallest allocated weights) until the dual assignment matches the arriving weight, i.e., until $\offlinedual = \weightprimed$. This procedure always terminates because $\pen(1) \geq 1$ implies $\int_{0}^{\weightprimed}\pen(1)\,\dd\weight \geq \weightprimed$,
ensuring that the stopping condition $\offlinedual = \weightprimed$ is eventually met. Since the dual variable $\offlinedual$ only increases thereafter, the feasibility remains guaranteed. This process is formalized in \Cref{alg:primal dual single resource}. Interestingly, the class of algorithms provided by \Cref{alg:primal dual single resource} for various choices of the penalty function $\pen$ satisfies the structural properties we discussed earlier. We sketch the proof below, followed by the detailed proof.}

\revcolorCMTwo{\begin{algorithm}
\caption{\revcolorCMTwo{Primal-dual fractional online algorithm for
$\continstances$ (single-resource)}}
\label{alg:primal dual single resource}
    \SetKwInOut{Input}{input}
    \SetKwInOut{Output}{output}
 \Input{Penalty function $\pen$
 }
 
 \vspace{2mm}
 
 Initialize $\offlinedual \gets 0$, and for all $\weight \in \reals_+$, set $
    \alloc(\weight) \gets \diracdeltafunction_0(\weight)$ and $\calloc(\weight) \gets \indicator{\weight =0}$.
 
 \vspace{1mm}
 
 {\color{royalazure}\tcc{$\diracdeltafunction_{\weight'}(\cdot)$
 is the Dirac delta function centered at $\weight'$.}}
 
 \vspace{1mm}
 
 \For{each online node $\weightprimed\in[0,\totaltime]$}
 {
 
 \vspace{1mm}
 
 \While{$\offlinedual < \weightprimed$}
 {
 \vspace{1mm}
 Buy back infinitesimal fraction $dx$ from the smallest allocated weight $\buybackweight$, i.e., $\alloc(\weight) \gets
 \alloc(\weight) - dx \cdot\delta_{\buybackweight}(\weight)$~~{\color{royalazure}\tcc{Formally, $\buybackweight = \inf\{\weight\in\reals_+:
 \calloc(\weight) < 1\}$}}
 
 \vspace{1mm}
 Allocate fraction $dx$ to the current online node $\weightprimed$, i.e., $\alloc(\weight) \gets \alloc(\weight) + dx \cdot\delta_{\weightprimed}(\weight)$
 
 \vspace{1mm}
 Update the allocation quantile function $\forall \weight \in [\buybackweight,\weightprimed]: \displaystyle\calloc(\weight) \gets \int_{\weight}^{\infty}
\alloc(\weight')\,\dd\weight'$

 \vspace{1mm}
Update the dual assignment $\displaystyle\offlinedual\gets 
\int_{0}^{\infty}
\pen(\calloc(\weight))\,\dd\weight$}
 }
\end{algorithm}}

\begin{restatable}{proposition}{propproperties}
\label{prop:primal-dual-properties}
\revcolorCMTwo{In the single-resource environment, given any monotone increasing function $\pen:[0,1]\rightarrow\mathbb{R}_{\geq 0}$ with $\pen(0)=0$ and $\pen(1)\geq 1$, \Cref{alg:primal dual single resource} instantiated with penalty function $\pen$ satisfies properties~\eqref{eq:greedily allocation property} and~\eqref{eq:allocation invariant property} when executed on truncated weight continuum instances $\continstances$.}
\end{restatable}

\begin{proof}{\emph{Proof sketch of \Cref{prop:primal-dual-properties}.}}
\revcolorCMTwo{By construction, \Cref{alg:primal dual single resource} clearly satisfies the greedy buyback property~\eqref{eq:greedily allocation property}. To see why it also satisfies the scale invariance property~\eqref{eq:allocation invariant property}, note that the algorithm continues allocating upon the arrival of weight $\weightprimed$ as long as the following allocation condition holds:
\begin{align*}
    \weightprimed >& \offlinedual =\int_0^{\infty}\pen\left(\calloc(\weight)\right)\,\dd\weight= \weight' \int_0^{\infty}\pen\left(\calloc\left(\frac{\weight}{\weight'}\cdot \weight'\right)\right)\,\dd\left(\frac{\weight}{\weight'}\right),
\end{align*}
for any $\weight' \in (0,\weightprimed)$. Equivalently, the algorithm continues allocating upon the arrival of weight $\weightprimed$ if the following condition holds:
$$
\frac{ \weightprimed}{\weight'}>\int_0^{\infty}\pen\left(\calloc\left(\frac{\weight}{\weight'}\cdot \weight'\right)\right)\,\dd\left(\frac{\weight}{\weight'}\right)~.
$$
The above condition matches precisely the allocation criterion of \Cref{alg:primal dual single resource} after scaling the weight continuum instance $(0,\weightprimed)$ by factor $\weight'$, that is, after transforming $\weight\rightarrow \frac{\weight}{\weight'}$ and correspondingly changing $\calloc(\weight)\rightarrow \calloc(\weight\cdot \weight')$. Now, we can couple the execution of \Cref{alg:primal dual single resource} on an instance with arriving weights in $(0,\weightprimed)$ to its execution on another instance with weights in $(0,\frac{\weightprimed}{\weight'})$ via the mapping $\weight\leftrightarrow \frac{\weight}{\weight'}$. Consequently, the allocation quantile functions in both executions share the same shape, differing only by a rescaling factor of $\weight'$, that is,
\begin{align*}
\forall \weight, \weight':\quad
\calloc^{(\weightprimed)}\left(\weight\cdot\weight'\right) = \calloc^{(\frac{\weightprimed}{\weight'})}\left(\weight\right)\quad\Longrightarrow\quad
\calloc^{(\weightprimed)}\left(\weight\right) = \calloc^{(\frac{\weightprimed}{\weight'})}\left(\frac{\weight}{\weight'}\right).
\end{align*}
Setting $\weight'=\weightprimed$ completes the proof of the scale invariance property~\eqref{eq:allocation invariant property}. \hfill \halmos}
\end{proof}

\revcolorCMTwo{\begin{proof}{\emph{Proof of \Cref{prop:primal-dual-properties}.}}
As stated in the proof sketch, we only need to show that \Cref{alg:primal dual single resource} satisfies the scale invariance property~\eqref{eq:allocation invariant property}. Fix $\weight',\weightprimed\in \reals_{+}$, and consider two separate execution runs of \Cref{alg:primal dual single resource}, where in the first run the input instance is the truncated continuum weight instance $[0,\weightprimed]$ and in the second the input instance is the truncated continuum weight instance $[0,\frac{\weightprimed}{\weight'}]$. Let $\tilde{w}$ be a variable that keeps track of the weight that is currently arriving in the first run. We couple the two runs by letting the currently arriving weight of the second run be $\frac{\tilde{w}}{\weight'}$. We use the notation $y_1^{(\tilde{\weight})}(\weight)$ and $x_1^{(\tilde{\weight})}(\weight)$ to track the allocation quantile function and the allocation density function of the first run, respectively, as the weight of the current arrival $\tilde{\weight}$ changes from $0$ to $\weight$, and. Similarly, we use $y_2^{(\tfrac{\tilde{\weight}}{\weight'})}(\weight)$ and $x_2^{(\tfrac{\tilde{\weight}}{\weight'})}(\weight)$ to track the same quantities for the second run as $\tilde{\weight}$ changes. To show the scale invariance property~\eqref{eq:allocation invariant property}, it is enough to show that
\begin{equation}
\label{eq:x-scale}  
\forall \weight \in \reals_{+}:~x_2^{(\tfrac{\weightprimed}{\weight'})}(\weight)=\weight'x_1^{(\weightprimed)}(\weight'\cdot\weight)~,
\end{equation}
because we can then integrate both sides over $\weight$ and conclude that
$$
\forall \weight \in \reals_{+}:~y_2^{(\tfrac{\weightprimed}{\weight'})}(\weight)=y_1^{(\weightprimed)}(\weight'\cdot\weight).
$$
By a change of variable from $w$ to $\frac{w}{\weight'}$ and replacing $\weight'$ with $\weightprimed$, we obtain the statement in the scale invariance property~\eqref{eq:allocation invariant property}, as desired. 

To show condition~\eqref{eq:x-scale}, we inductively show that it will be preserved during the two coupled runs of \Cref{alg:primal dual single resource}. For the base case,  before we start the execution of the algorithm, we have $\tilde{\weight}=0$. Without loss of generality, we can assume that the algorithm has allocated $1$ unit of weight $0$ at this moment. Therefore, $x_1^{(0)}(\weight)=\delta_0(\weight)$ and $x_2^{(0)}(\weight)=\delta_0(\weight)$, where $\delta_a(b)$ is a Dirac delta function centered at $a$ and evaluated at $b$. Using the scaling property of the Dirac delta function, we have:
$$
x_2^{(0)}(\weight)=\delta_0(\weight)=\weight' \delta_0(\weight'\cdot\weight)=\weight' x_1^{(0)}(\weight'\cdot \weight),
$$
as desired. Now, suppose that we are at some point during the execution of these two coupled runs of the algorithm, and the currently arriving weight that the algorithm is allocating in the first run is $\tilde{\weight}$ (and therefore, the currently allocating weight in the second run is $\frac{\tilde{\weight}}{\weight'}$). Let $\tilde{\weight}_b$ denote the weight that is bought back at this point in the first run. According to the inductive hypothesis, the weight that is bought back in the second run is equal to $\frac{\tilde{\weight}_b}{\weight'}$. Moreover, for the current allocation density functions $x_{1,\textrm{old}}^{(\tilde{\weight})}(\cdot)$ and $x_{2,\textrm{old}}^{(\tfrac{\tilde{\weight}}{\weight'})}(\cdot)$ before the possible allocation at this moment, the following holds due to the inductive hypothesis:
\begin{equation}
\forall \weight \in \reals_{+}:~~~x_{2,\textrm{old}}^{(\tfrac{\weightprimed}{\weight'})}(\weight)=\weight'x_{1,\textrm{old}}^{(\weightprimed)}(\weight'\cdot\weight)~.
\end{equation}
We now consider two cases to finish the induction:

\underline{\emph{Case I}}: the algorithm has terminated its continuous allocation upon arrival of weight $\tilde{\weight}$ in the first run. Let $\beta_1$ and $\beta_2$ be the dual variables constructed during the first and second runs, respectively, at this moment. Because the algorithm has terminated in the first run, we have
$$
\tilde{\weight}=\beta_1=\int_0^{\infty}\pen\left(y_1^{(\tilde{\weight})}(\weight)\right)\,\dd\weight\overset{(a)}{=}\weight' \int_0^{\infty}\pen\left(y_1^{\tilde{\weight}}\left(z\cdot \weight'\right)\right)\,\dd z\overset{(b)}{=}\weight' \int_0^{\infty}\pen\left(y_2^{(\tfrac{\tilde{\weight}}{\weight'})}\left(z\right)\right)\,\dd z=\weight' \beta_2,
$$
where equality~(a) holds by a change of variable in the integral, and equality~(b) holds as a result of the inductive hypothesis. Therefore, $\frac{\tilde{\weight}}{\weight'}=\beta_2$, and therefore the algorithm will also terminate in the second run and allocation density functions remain unchanged in both runs. Hence the statement is true by inductive hypothesis. 

\underline{\emph{Case II}}: the algorithm is still in the middle of its continuous allocation upon arrival of weight $\tilde{\weight}$ in the first run. Again, let $\beta_1$ and $\beta_2$ be the dual variables constructed during the first and second runs, respectively, at this moment. Note that $\tilde{\weight}>\beta_1$. Similar to Case~I, we have:
$$
\tilde{\weight}>\beta_1=\int_0^{\infty}\pen\left(y_1^{(\tilde{\weight})}(\weight)\right)\,\dd\weight=\weight' \int_0^{\infty}\pen\left(y_1^{\tilde{\weight}}\left(z\cdot \weight'\right)\right)\,\dd z=\weight' \int_0^{\infty}\pen\left(y_2^{(\tfrac{\tilde{\weight}}{\weight'})}\left(z\right)\right)\,\dd z=\weight' \beta_2,
$$
and therefore $\frac{\tilde{\weight}}{\weight'}>\beta_2$. As a result, the algorithm allocates $\dd x$ amount at weight $\tilde{\weight}$ (and buys back $\dd x$ amount at weight $\tilde{\weight}_b$) in the first run, and also allocates $dx$ amount at weight $\frac{\tilde{\weight}}{\weight'}$ (and buys back $\dd x$ amount at weight $\frac{\tilde{\weight}_b}{\weight'}$) in the second run. Let $\Delta x_1^{(\tilde{\weight})}(\cdot)$ and $\Delta x_2^{(\tfrac{\tilde{\weight}}{\weight'})}(\cdot)$ denote the change in the allocation density functions in the first and second runs, respectively. Then we have:
\begin{align*}
\Delta x_2^{(\tfrac{\tilde{\weight}}{\weight'})}(\weight)&\overset{(c)}{=}\dd x \cdot \left(\delta_{\frac{\tilde{\weight}}{\weight'}}(\weight)-\delta_{\frac{\tilde{\weight}_b}{\weight'}}(\weight)\right)\\
&=\dd x \cdot \left(\delta_{0}(\weight-\frac{\tilde{\weight}}{\weight'})-\delta_{0}(\weight-\frac{\tilde{\weight}_b}{\weight'})\right)\\
&\overset{(d)}{=}\weight'\cdot \dd x \cdot\left(\delta_{0}\left(\weight'\weight-\tilde{\weight}\right)-\delta_{0}\left(\weight'\weight-\tilde{\weight}_b\right)\right)\\
&=\weight'\cdot \dd x \cdot\left(\delta_{{\tilde{\weight}}}\left(\weight'\cdot \weight\right)-\delta_{\tilde{\weight}_b}\left(\weight'\cdot\weight\right)\right)\\
&\overset{(e)}=\weight' \cdot \Delta x_1^{(\tilde{\weight})}(\weight'\cdot\weight)~,
\end{align*}
where equalities~(c) and (e) hold due to the allocations and buybacks of the algorithm in these two runs as described earlier, and equality~(d) holds due to the scaling property of the Dirac delta function. Putting the above equalities and the inductive hypothesis together, after the algorithm finishes allocating, we have:
\begin{equation*}
    x_2^{(\tfrac{\tilde{\weight}}{\weight'})}(\weight)=x_{2,\textrm{old}}^{(\tfrac{\tilde{\weight}}{\weight'})}(\weight)+\Delta x_2^{(\tfrac{\tilde{\weight}}{\weight'})}(\weight)=\weight'x_{1,\textrm{old}}^{(\weightprimed)}(\weight'\cdot\weight)+\weight' \cdot \Delta x_1^{(\tilde{\weight})}(\weight'\cdot\weight)=\weight' \cdot  x_1^{(\tilde{\weight})}(\weight'\cdot\weight)~,
\end{equation*}
which finishes the proof of the inductive statement, as desired.
\hfill
\halmos
\end{proof}}

\revcolorCMTwo{Given these two properties of the allocation distribution of \Cref{alg:primal dual single resource}---namely the scale invariance~\eqref{eq:allocation invariant property} and the greedy buyback~\eqref{eq:greedily allocation property}---we are ready to analyze the competitive ratio of \Cref{alg:primal dual single resource}. We first consider a specific choice of penalty function that simplifies the primal-dual analysis. This particular choice, termed as the generalized exponential function, is exactly the choice used by the direct algorithm in the main body (\Cref{prop:optimal-direct-single-resource-fractional}) and also plays a critical role in our main result for the matching environment in \Cref{sec:matching}. Later, we extend our single-resource analysis to more general penalty functions (\Cref{app:single-resource-characterization}).}

\begin{proposition}
\label{prop:optimal-single-resource-fractional}
\revcolorCMTwo{In the single-resource environment, for any buyback factor $f\geq 0$, consider \Cref{alg:primal dual single resource} with the penalty function $\pen(\calloc)=
\frac{1}{\log(\thresholdweight)}(\thresholdweight^\calloc - 1)$, where $\thresholdweight = -(1 + \buybackcost) \Lambertterm$. This algorithm achieves the competitive ratio of $-\Lambertterm$ within the class of truncated weight continuum instances $\continstances$.}
\end{proposition}
\begin{proof}{\emph{Proof.}}
\revcolorCMTwo{Consider running \Cref{alg:primal dual single resource} on an instance $I_T\in\continstances$. To analyze the competitive ratio, we employ `dual fitting' via the primal and dual LP pair for the offline optimum, formally described as \ref{eq:primal-dual-single-cont}. Specifically, we construct a feasible dual solution whose objective value is no more than $-\Lambertterm$ times the primal objective. This, combined with LP weak duality, completes the proof.

We use the following dual assignment with $\pen(\calloc)=
\frac{1}{\log(\thresholdweight)}(\thresholdweight^\calloc - 1)$, which is also used by \Cref{alg:primal dual single resource}:
$$\hat{\offlinedual} = \int_{0}^{\infty} \pen(\calloc^{\textrm{final}}(\weight))\,\dd\weight~,$$
where $\calloc^{\textrm{final}}(\weight)$ is the final allocation quantile function of \Cref{alg:primal dual single resource}, after processing the last arrival in $I_T$.

First, we show  feasibility of this dual assignment. For our choice of the penalty function $\pen$, we have
\begin{align*}
 \pen(1) = \frac{\thresholdweight -1}{\log( \thresholdweight)} \overset{(a)}{\geq} \frac{\thresholdweight}{\log(\thresholdweight) +1} \overset{(b)}{=} 1+f\geq 1,
\end{align*}
where the inequality~(a) holds because for positive values of $\thresholdweight$ we have $\thresholdweight-1\geq \log(\thresholdweight)$, and equality~(b) holds as $\thresholdweight = -(1 + \buybackcost) \Lambertterm$.
As a result of inequality $\pen(1)\geq 1$, the algorithm keeps allocating to an arriving online node with weight $\weightprimed\in [0,T]$ until the dual (maintained by the algorithm) matches the weight $\weightprimed$, that is $\beta=\weightprimed$ (as otherwise, $\beta<\weightprimed$, but we know $\int_{0}^{\weightprimed}\pen(1)\,\dd\weight \geq \weightprimed$). After this point, $\beta$ can only increase; thus, for any $\weightprimed \in [0,T]$, we have $\hat{\beta} \geq \weightprimed$.

Next, we compare the change in the primal objective, denoted by $\Delta(\text{Primal})$, and the dual objective, denoted by $\Delta(\text{Dual})$, as the algorithm allocates an infinitesimal fraction $dx$ of online node $\weightprimed\in [0,T]$ and buying back the same fraction from online node with weight $\weight'$. We have:
\begin{align*}
    \Delta(\text{Primal}) &=
    \left(\weightprimed - (1+\buybackcost)\weight'\right)dx~.
\end{align*}
Furthermore, let $\calloc(\weight)$ be the quantile allocation function of \Cref{alg:primal dual single resource} before processing this infinitesimal allocation of online node $\weightprimed$. Then the change in the dual objective is upper-bounded as follows :
\begin{align*}
    \Delta(\text{Dual})&= \int_{\weight'}^{\weightprimed}
    \penderivative(\calloc(\weight))\,\dd\weight\\
    &\overset{(c)}{\leq}
    \left(
    \log(\thresholdweight)
    \left(
    \weightprimed -
    \displaystyle\int_0^\infty \pen(\calloc(\weight))\,\dd\weight
    \right)
    +
    \int_{\weight'}^{\weightprimed}
    \penderivative(\calloc(\weight))\,\dd\weight
    \right)\,dx
    \\
    &\overset{(d)}{\leq}
    \left(
    \log(\thresholdweight)
    \left(
    \weightprimed -
    \displaystyle\int_0^{\weight'} \pen(1)\,\dd\weight
    -
    \displaystyle\int_{\weight'}^{\weightprimed} \pen(\calloc(\weight))\,\dd\weight
    \right)
    +
    \int_{\weight'}^{\weightprimed}
    \penderivative(\calloc(\weight))\,\dd\weight
    \right)\,dx
    \\
    &=
    \left(
    \log(\thresholdweight)
    \weightprimed -
    \left(\thresholdweight - 1\right)\weight'
    -\displaystyle\int_{\weight'}^{\weightprimed}
    \left(
    \thresholdweight^{\calloc(\weight)} - 1
    \right)\,\dd\weight
    +
    \int_{\weight'}^{\weightprimed}
    \thresholdweight^{\calloc(\weight)}
    \,\dd\weight
    \right)\,dx
    \\
    &=
    \left((\log(\thresholdweight) + 1)\weightprimed
    -\thresholdweight
    \cdot\weight' \right)dx \overset{(e)}{=} \frac{\thresholdweight}{1+\buybackcost}\left(\weightprimed -(1+\buybackcost)\weight'\right)dx =-\Lambertterm \Delta(\text{Primal})~,
\end{align*}
where inequality~(c) holds as $\int_{0}^{\infty} \pen(\calloc(\weight))\,\dd\weight=\offlinedual\leq \weightprimed$ after processing  online node with weight $\weightprimed$ ($\offlinedual$ is the dual maintained by the algorithm), and inequality~(d) holds due to the greedy buyback property~\eqref{eq:greedily allocation property}. Lastly, equality~(e) holds as $\thresholdweight = -(1 + \buybackcost) \Lambertterm$. Summing over all allocations, we conclude that the dual objective is at most $-\Lambertterm$ times the primal objective, which finishes the proof. \hfill\halmos}
\end{proof}

\medskip
\noindent\revcolorCMTwo{\textbf{Optimal competitive fractional online algorithm for general instances.}
By combining \Cref{prop:algorithm reduction single resource} and \Cref{prop:optimal-single-resource-fractional}, we conclude that \Cref{alg:algorithm reduction single resource}, using \Cref{alg:primal dual single resource} as input (with an appropriate penalty function, as stated in \Cref{prop:optimal-single-resource-fractional}), yields an optimal competitive fractional online algorithm for the single-resource environment within general instances $\instances$---providing an alternative to the direct algorithm and analysis of \Cref{sec:single resource primal dual} in the main body.}

\subsection{Characterizing the optimal allocation distribution}
\label{app:single-resource-characterization}
\revcolorCMTwo{So far, we have analyzed \Cref{alg:primal dual single resource} under a specific (generalized exponential) choice of the penalty function $\pen$. In this subsection, we present a refined analysis tailored to the single-resource environment, showing that \Cref{alg:primal dual single resource} remains optimally competitive for essentially \emph{any} penalty function (when suitably scaled). To this end, we exactly characterize the allocation distributions induced by online algorithms satisfying the two structural properties of \Cref{sec:single-properties}.}

Consider an online algorithm that satisfies the two properties introduced in \Cref{sec:single-properties}---namely the scale invariance property~\eqref{eq:allocation invariant property} and the greedy buyback property~\eqref{eq:greedily allocation property}. Next, we run this algorithm on a truncated weight continuum instance $I_T\in\continstances$ for some $T\geq 1$. For simplicity, we further assume that the allocation quantile functions $\calloc(\cdot)$ are continuous on the interval $(0, T]$ for all such instances. Under these assumptions, we provide a closed-form characterization of the resulting allocation distributions, parameterized by $\thresholdweight$, as well as a characterization of the competitive ratio as a function of this parameter.

\begin{figure}[htb]
     \centering
     \subfigure[Allocation density function]{
         \centering
         \begin{tikzpicture}[scale=0.66, transform shape]
\begin{axis}[
axis line style=gray,
axis lines=middle,
xtick style={draw=none},
ytick style={draw=none},
xticklabels=\empty,
yticklabels=\empty,
xmin=-0.425,xmax=3,ymin=-0.45,ymax=4,
width=0.8\textwidth,
height=0.45\textwidth,
samples=500]

\addplot[domain=0:3, gray!40!white, dashed, line width=1mm] (x, {1/x});

\addplot[domain=0:1, black, line width=0.5mm] (x, {0});
\addplot[domain=2.71828182845:3, black, line width=0.5mm] (x, {0});
\addplot[domain=1:2.71828182845, black, line width=0.5mm] (x, {1/x});

\addplot[mark=*,only marks, fill=white] coordinates {(1,0)} node[above, pos=1]{};
\addplot[mark=*,only marks, fill=white] coordinates {(2.71828182845,0)} node[above, pos=1]{};

\addplot[gray, dotted] coordinates {(1,0)
(1,1)};
\addplot[gray, dotted] coordinates {(2.71828182845,0)
(2.71828182845,0.3678794411726664)};

\addplot[] coordinates {(1,0.)} node[below, pos=1]{$\buybackweightwpsup=\sfrac{\weightprimed}{\thresholdweight}$};

\addplot[] coordinates {(-0.22,4)} node[below, pos=1]{$x^{(\weightprimed)}(\weight)$};

\addplot[] coordinates {(0.8,2)} node[below, pos=1]{$\hat{x}(w)$};

\addplot[] coordinates {(2.71828182845,0.)} node[below, pos=1]{$\weightprimed$};

\addplot[gray, thick] coordinates {(0.33333333,0.01)
(0.33333333,-0.01)};

\tikzset{
        hatch distance/.store in=\hatchdistance,
        hatch distance=8pt,
        hatch thickness/.store in=\hatchthickness,
        hatch thickness=1.5pt
    }

    \makeatletter
    \pgfdeclarepatternformonly[\hatchdistance,\hatchthickness]{flexible hatch}
    {\pgfqpoint{0pt}{0pt}}
    {\pgfqpoint{\hatchdistance}{\hatchdistance}}
    {\pgfpoint{\hatchdistance-1pt}{\hatchdistance-1pt}}%
    {
        \pgfsetcolor{\tikz@pattern@color}
        \pgfsetlinewidth{\hatchthickness}
        \pgfpathmoveto{\pgfqpoint{0pt}{0pt}}
        \pgfpathlineto{\pgfqpoint{\hatchdistance}{\hatchdistance}}
        \pgfusepath{stroke}
    }
    \makeatother
\addplot+[mark=none,
        domain=1:2.71828182845,
        samples=100,
        pattern=flexible hatch,
        area legend,
        draw=gray!70!white,
        pattern color=gray!70!white]{1/x)} \closedcycle;
\end{axis}
\end{tikzpicture}
   \label{fig:canonical-density}
      }
        \centering
     \subfigure[Allocation quantile function]{
         \centering
         \begin{tikzpicture}[scale=0.65, transform shape]
\begin{axis}[
axis line style=gray,
axis lines=middle,
xtick style={draw=none},
ytick style={draw=none},
xticklabels=\empty,
yticklabels=\empty,
xmin=-0.425,xmax=3,ymin=-0.45,ymax=4,
width=0.8\textwidth,
height=0.45\textwidth,
samples=500]


\addplot[domain=0:1, black, line width=0.5mm] (x, {1});

\addplot[domain=2.71828182845:3, black, line width=0.5mm] (x, {0});
\addplot[domain=1:2.71828182845, black, line width=0.5mm] (x, {1-ln(x)});

\addplot[mark=*,only marks, fill=white] coordinates {(1,0)} node[above, pos=1]{};
\addplot[mark=*,only marks, fill=white] coordinates {(2.71828182845,0)} node[above, pos=1]{};

\addplot[gray, dotted] coordinates {(1,0)
(1,1)};
\addplot[gray, dotted] coordinates {(2.71828182845,0)
(2.71828182845,0.3678794411726664)};

\addplot[] coordinates {(1,0.)} node[below, pos=1]{$\buybackweightwpsup=\sfrac{\weightprimed}{\thresholdweight}$};

\addplot[] coordinates {(2.71828182845,0.)} node[below, pos=1]{$\weightprimed$};

\addplot[] coordinates {(-0.22,4)} node[below, pos=1]{$y^{(\weightprimed)}(w)$};
\addplot[gray, thick] coordinates {(0.33333333,0.01)
(0.33333333,-0.01)};

\tikzset{
        hatch distance/.store in=\hatchdistance,
        hatch distance=8pt,
        hatch thickness/.store in=\hatchthickness,
        hatch thickness=1.5pt
    }

   \makeatletter
    \pgfdeclarepatternformonly[\hatchdistance,\hatchthickness]{thick hatch}
    {\pgfqpoint{0pt}{0pt}}
    {\pgfqpoint{\hatchdistance-1pt}{\hatchdistance-1pt}}
    {\pgfpoint{\hatchdistance-3pt}{\hatchdistance-3pt}}%
    {
        \pgfsetcolor{\tikz@pattern@color}
        \pgfsetlinewidth{\hatchthickness}
        \pgfpathmoveto{\pgfqpoint{0pt}{0pt}}
        \pgfpathlineto{\pgfqpoint{\hatchdistance}{\hatchdistance}}
        \pgfusepath{stroke}
    }
    \makeatother

\addplot+[mark=none,
        domain=0:1,
        samples=100,
        pattern=thick hatch,
        area legend,
        draw=gray!70!white,
        pattern color=gray!70!white]{(1)} \closedcycle;  

\addplot+[mark=none,
        domain=1:2.71828182845,
        samples=100,
        pattern=thick hatch,
        area legend,
        draw=gray!70!white,
        pattern color=gray!70!white]{(1-ln(x))} \closedcycle;

\end{axis}

\end{tikzpicture}
    \label{fig:canonical-quantile}      }
           \caption{Characterization of the optimal allocation distribution: Black solid curves are (a)~the allocation density function $\boldsymbol{x^{(\weightprimed)}(w)}$ after processing online node $\boldsymbol{\weightprimed}$ and (b)~the allocation quantile function $\boldsymbol{y^{(\weightprimed)}(w)}$ after processing online node $\boldsymbol{\weightprimed}$. Gray dashed curve in (a) is the 
     allocation density $\boldsymbol{\canalloc(\cdot)}$ in \Cref{prop:canonical allocation}.}
    \label{fig:canonical}
\end{figure}
\begin{restatable}{proposition}{canallocprop}
\label{prop:canonical allocation}
In the single-resource environment, for any online algorithm $\ALG$ that satisfies both the scale invariance property~\eqref{eq:allocation invariant property} and the greedy buyback property~\eqref{eq:greedily allocation property}, there exists a choice of parameter $\thresholdweight \geq 1$ such that when the online algorithm $\ALG$ is executed on instances in $\continstances$:
\begin{enumerate}[label=(\roman*)]
    \item It always allocates a fraction $\canalloc(\weight)$ to the online node with weight $\weight$ upon its arrival, defined as
    \begin{align*}
       \canalloc(\weight) \triangleq \frac{1}{\weight}\frac{1}{\log(\thresholdweight)}~~~~\left(\textrm{if}~\thresholdweight = 1:~~\canalloc(\weight) \triangleq \delta_0(\weight)\right)~,
    \end{align*}
    where $\delta_0(\cdot)$ is the Dirac delta function centered at $0$.
    Moreover, the allocation density function $x^{(\weightprimed)}(\cdot)$ and the quantile function $y^{(\weightprimed)}(\cdot)$ of $\ALG$ after processing the online node $\weightprimed$ are characterized as follows (\Cref{fig:canonical}):
    \begin{align*}
        x^{(\weightprimed)}(\weight) =\left\{
        \begin{array}{ll}
         \frac{1}{\weight}\frac{1}{\log(\thresholdweight)} & \weight\in[\frac{\weightprimed}{\thresholdweight}, \weightprimed] \\
    0 & \textrm{o.w.}
        \end{array}
        \right.~~~,~~~ y^{(\weightprimed)}(w) =
    \left\{
    \begin{array}{ll}
      1    & \weight \in [0, \frac{\weightprimed}{\thresholdweight}]  \\
      1 - \frac{\log\left(\weight\cdot\frac{\thresholdweight}{\weightprimed}\right)}{\log(\thresholdweight)}
      \qquad\qquad
      & \weight\in[\frac{\weightprimed}{\thresholdweight}, \weightprimed]\\
      0   & \weight\in[\weightprimed, \infty)
    \end{array}
    \right.
    \end{align*}
    \item It achieves a competitive ratio of 
     $   \frac{\thresholdweight\log(\thresholdweight)}{
    \thresholdweight - 1 - \buybackcost}$
    within the instances in the class $\continstances$.
    \end{enumerate}
\end{restatable}
\begin{remark}[\revcolorCMTwo{Monotonicity of the canonical allocation density}]\label{rmk:monotone canonical}
We observe a somewhat counterintuitive behavior: $\canalloc(\weight)$ decreases as $\weight$ increases. This occurs because, despite the decreasing density $\canalloc(\weight)$, the interval of weights over which the algorithm retains allocations, that is $[\weight^\dagger/\hat{\weight},\weight^\dagger]$, expands linearly as $\weight^{\dagger}$ progresses, as illustrated in \Cref{fig:canonical}.
\end{remark}

\begin{proof}{\emph{Proof of part~(i) of \Cref{prop:canonical allocation}.}}
First, we consider the case where $\buybackweightonesup = 1$.
In this case, by definition $\calloconesup(\weight) =
\indicator{\weight \leq 1}$. Invoking the scale invariance property~\eqref{eq:allocation invariant property},
$\callocwpsup(\weight) =
\indicator{\weight \leq \weightprimed}$, 
and thus the algorithm follows $\canalloc(.)$ allocation 
with $\thresholdweight = 1$.

Now, suppose $\buybackweightonesup < 1$. The greedy buyback property~\eqref{eq:greedily allocation property}
ensures that 
for every $\weightprimed\in[\buybackweightonesup, 1]$,
\begin{align*}
    \calloconesup(\weightprimed)
    =
    \callocwpsup(\weightprimed) 
    +
    1 
    -
    \callocwpsup(\buybackweightonesup)~.
\end{align*}
Invoking the scale invariance property~\eqref{eq:allocation invariant property},
the equation 
\begin{align*}
    \calloconesup(\weightprimed)
    =
    \calloconesup\left(
    \frac{\weightprimed}{\weightprimed}
    \right) 
    +
    1 
    -
    \calloconesup
    \left(\frac{\buybackweightonesup}
    {\weightprimed}\right)
\end{align*}
holds for every $\weightprimed\in[\buybackweightonesup, 1]$.
Thus, taking derivatives of both sides with respect to 
$\weightprimed$, we obtain the equation
\begin{align*}
    \alloconesup(\weightprimed) = 
    \frac{\buybackweightonesup}{(\weightprimed)^2}
    \cdot 
    \alloconesup\left(\frac{\buybackweightonesup}
    {\weightprimed}\right)~,
\end{align*}
which holds for every $\weightprimed\in[\buybackweightonesup,1]$.
This equation admits a unique solution (up to a constant $\largeconstant$) of the from
\begin{align*}
    \forall \weight \in [\buybackweightonesup, 1]: \qquad \alloconesup(\weight) = \frac{\largeconstant}{\weight}~.
\end{align*}
Plugging the boundary condition where $\calloconesup(\buybackweightonesup) = 1$ and $\calloconesup(1)= 0$, we have $\largeconstant = \frac{1}{-\log(\buybackweightonesup)} = \frac{1}{\log(\thresholdweight)}$.
Finally, invoking the scale invariance property, we have $\allocwpsup(\weightprimed) = \canalloc(\weightprimed)$ with $\thresholdweight = \frac{1}{\buybackweightonesup}$
for every $\weightprimed\in\reals_+$. Also $\callocwpsup(\weight)$ can be written as 
\begin{align*}
    \callocwpsup(\weight) = 1 - \int_0^w \canalloc(w') dw' = \left\{
    \begin{array}{ll}
      1    & \weight \in [0, \frac{\weightprimed}{\thresholdweight}]  \\
      1 - \frac{\log\left(\weight\cdot\frac{\thresholdweight}{\weightprimed}\right)}{\log(\thresholdweight)}
      \qquad\qquad
      & \weight\in[\frac{\weightprimed}{\thresholdweight}, \weightprimed]\\
      0   & \weight\in[\weightprimed, \infty)
    \end{array}
    \right.
\end{align*}
which finishes the proof.
\hfill\halmos
\end{proof}

\begin{proof}{\emph{Proof of part~(ii) of \Cref{prop:canonical allocation}.}}
Consider an arbitrary truncated weight continuum instance $\instance_\totaltime\in\continstances$.
Due to part~(i) of \Cref{prop:canonical allocation},
online algorithm allocates $\canalloc(\weight)$ density 
of the offline node to each weight $\weight\in[0,\totaltime]$.
Furthermore, every weight $\weight\in [0, \sfrac{\totaltime}{\thresholdweight}]$
has been boughtback. Therefore,
the total profit of $\ALG$ (i.e., total weight minus cancellation costs) is 
\begin{align*}
    \ALG(\instance_\totaltime) &= 
    \displaystyle\int_{0}^\totaltime
    \weight\cdot \canalloc(\weight)\,d\weight 
    -
    (1 + \buybackcost)
    \displaystyle\int_{0}^{\frac{\totaltime}{\thresholdweight}}
    \weight\cdot \canalloc(\weight)\,d\weight
    =
    \frac{\thresholdweight - 1 - \buybackcost}{\thresholdweight\log(\thresholdweight)}\cdot \totaltime
\end{align*}
On the other side, the \revcolorCMTwo{offline optimum} benchmark is $\OPT(\instance_\totaltime) = \totaltime$, finishing the proof.
\hfill
\halmos
\end{proof}
Now, given the parametric characterization in \Cref{prop:canonical allocation}, we can find the parameter $\thresholdweight$ that yields the best competitive ratio, thus identifying the optimal allocation distribution in this class. This provides an online (fractional) algorithm that meets both the greedy buyback property~\eqref{eq:greedily allocation property} and the scale invariance property~\eqref{eq:allocation invariant property}, and has the best competitive ratio within the class $\continstances$. We can then directly compare this competitive ratio to the lower bound established in \Cref{prop:lower bound single resource}.

\oldcolor{
Applying the first-order condition and using straightforward calculations, we can show that the competitive ratio formula  $\frac{\thresholdweight\log(\thresholdweight)}
{\thresholdweight - 1 - \buybackcost}$ attains its minimum $\min_{\thresholdweight \geq 1}\frac{\thresholdweight\log(\thresholdweight)}
{\thresholdweight - 1 - \buybackcost}
= -\Lambertterm$ 
at $\thresholdweight = -(1+\buybackcost)
\Lambertterm$,
which \emph{exactly} matches the lower bound of the competitive ratio 
for the single-resource environment within \emph{all} problem instances (\Cref{prop:lower bound single resource}). As noted earlier, the hard instance for this lower bound belongs to the class of truncated weight continuum instances $\continstances$; therefore, if we have any online algorithm whose induced allocation distribution is one characterized in \Cref{prop:canonical allocation} with the parameter $\thresholdweight = -(1+\buybackcost)
\Lambertterm$, it will be optimally competitive within the class $\continstances$. One approach is to directly use the ``canonical allocation distribution'' in \Cref{prop:canonical allocation} (and \Cref{fig:canonical}) to make fractional allocation decisions (and greedily buy back as necessary). The resulting algorithm, by setting $\thresholdweight = -(1+\buybackcost)\Lambertterm$, will be optimal competitive within the class of truncated weight continuum instances $\continstances$. Another approach, as we show next, is to use \Cref{alg:primal dual single resource} with a general penalty function $\pen$.

\medskip
\noindent\textbf{An alternative indirect proof of \Cref{prop:optimal-single-resource-fractional}.} By combining \Cref{prop:primal-dual-properties} and \Cref{prop:canonical allocation}, we conclude that the allocation distribution of \Cref{alg:primal dual single resource} belongs to the class described in part~(ii) of \Cref{prop:canonical allocation} for some parameter $\thresholdweight \geq 1$, and hence obtains the competitive in part~(i) of \Cref{prop:canonical allocation} within the class of truncated weight continuum instances $\continstances$. If we manage to pick a penalty function $\pen$ (a monotone increasing function from $[0,1]$ to $\reals_{+}$) such that the algorithm's induced canonical allocation function results in the correct value for $\thresholdweight$, that is  $\thresholdweight = -(1+\buybackcost)\Lambertterm$, then \Cref{alg:primal dual single resource} will be optimally competitive within the class of instances in $\continstances$.

Note that, for any function $\pen$, we can establish the relationship between $\pen$ and $\thresholdweight$ using the following equation (which captures the termination condition after processing the online node $\thresholdweight$):
\begin{align*}
    \thresholdweight = \displaystyle\int_{0}^{\infty}\pen\left(\calloc^{(\thresholdweight)}(\weight)\right)\,d\weight~.
\end{align*}
This relationship can be further simplified by substituting the formula for $\calloc^{(\thresholdweight)}(\weight)$:
\begin{align}
\label{eq:parameter}
    \thresholdweight = 
    \pen(1) + 
    \displaystyle\int_1^{\thresholdweight}
    \pen\left(
    1 - \frac{\log(\weight)}{\log(\thresholdweight)}
    \right)\,
    d\weight
~.
\end{align}
The closed-form identity in \cref{eq:parameter} essentially shows how to scale the penalty function $\pen$ in \Cref{alg:primal dual single resource} appropriately given any target $\thresholdweight$ and any specific form for this monotone increasing function $\pen$ (satisfying $\pen(0) = 0$). For instance, if $\thresholdweight = 1$, we obtain $\pen(1) = 1$ and the penalty function can have any arbitrary form. In this case, the resulting algorithm degenerates into the greedy algorithm, which ignores the buyback cost. As another example, for any $\thresholdweight > 1$, straightforward calculations based on the identity in \cref{eq:parameter} show that $\pen(\calloc) =\frac{\thresholdweight \log(\thresholdweight)}{\thresholdweight - 1} \calloc$ for the linear penalty function, and $\pen(\calloc)=\frac{\thresholdweight \log(\thresholdweight) - \thresholdweight}{\thresholdweight - e} \left(e^{y}-1\right)$ for the exponential penalty function. Lastly, for a function of the form $\pen(y)=\tau (\thresholdweight^y-1)$ for a scalar parameter $\tau\geq 0$, by plugging the function in \cref{eq:parameter}, we have $\tau=\frac{1}{\log(\thresholdweight)}$, and therefore $\pen(y)=\frac{\thresholdweight^y-1}{\log(\thresholdweight)}$. This penalty function is \emph{exactly} the same penalty function used in \Cref{prop:optimal-single-resource-fractional}. Now, setting $\thresholdweight$ optimally as described above (and in the statement of \Cref{prop:optimal-single-resource-fractional}) makes the resulting \Cref{alg:primal dual single resource} optimally competitive---establishing also an alternative proof for \Cref{prop:optimal-single-resource-fractional}.\hfill\halmos
}

\subsection{Direct randomized algorithm of \texorpdfstring{\citet{AK-09}}{Badanidiyuru and Kleinberg (2009)} for single-resource environment}
\label{app:BK implementation}

In this section, we explore the relationship between the randomized integral online algorithm proposed by \cite{AK-09} and \Cref{alg:primal dual single resource}. As explained in \Cref{app:single-resource-characterization}, we can characterize the allocation distribution induced by \Cref{alg:primal dual single resource} under truncated weight continuum instances for different choices of the penalty function. This gives rise to the parametric canonical allocation distribution function $\canalloc(\cdot)$ defined in \Cref{prop:canonical allocation} (with parameter $\thresholdweight$). We now show how this (fractional) allocation distribution can be \emph{directly transformed} into a randomized integral algorithm whose performance is at least as good as that of the fractional algorithm inducing this allocation distribution. The resulting randomized algorithm algorithm is \emph{almost} identical to the algorithm proposed in \citealp{AK-09}.\footnote{As it will be clear soon, the only difference is that our algorithm, that is, \Cref{alg: randomized algorithm}, only buys back when there is a demand.}

Without loss of generality, assume the smallest arriving weight is greater than~$1$. Pick a uniform random number $\randseed \sim \mathcal{U}[0,1]$, and for each $\ell \in \naturals$, define the random variable $\weightc_\randseed(\ell)$ as the solution to:
\[
\int_1^{\weightc_\randseed(\ell)} \canalloc(\weight)\,d\weight = \randseed + \ell - 2.
\]
Given that $\canalloc(\weight) = \frac{1}{\weight \log \thresholdweight}$ (by \Cref{prop:canonical allocation}), we explicitly have:
\[
\weightc_\randseed(\ell) = \thresholdweight^{\randseed + \ell - 2}.
\]
Using these definitions, we present our algorithm below.
\begin{algorithm}
\caption{Direct primal-dual-based randomized integral online algorithm for general instances (single-resource)}
\label{alg: randomized algorithm}
    \SetKwInOut{Input}{Input}
    \SetKwInOut{Output}{Output}
    \Input{Random seed $\randseed$}

    \vspace{1mm}

    Initialize $\ell_{\text{max}} \gets 0$.

    \vspace{1mm}

    \For{each online node $i \in \onlinenodes$}{
        \vspace{1mm}

        Observe arriving weight $\weight_i$.

        \vspace{1mm}

        Find the largest integer $\ell \in \naturals$ such that $\weightc_\eta(\ell) \leq \weight_i$. 
        
        {\color{royalazure}\tcc{By definition, $\weightc_\eta(\ell) = \thresholdweight^{\eta + \ell - 2}$.}}

        \vspace{1mm}

        \If{$\ell > \ell_{\text{max}}$}{
            \vspace{1mm}

            Sample $x_i$ randomly as follows:
                $x_i = \begin{cases}
                    1 & \text{with probability } \frac{\weightc_\eta(\ell)}{\weight_i},\\[2pt]
                    0 & \text{otherwise}.
                \end{cases}$

            \vspace{1mm}

            \If{$x_i = 1$}{
                \vspace{1mm}

                Allocate integrally to online node $i$, and buying back the previously allocated node if necessary.
            }

            \vspace{1mm}

            Update $\ell_{\text{max}} \gets \ell$.
        }
    }
\end{algorithm}


\begin{lemma}
\label{lem: randomized arbitrary}
\Cref{alg: randomized algorithm} achieves a competitive ratio within general instances $\instances$ at least as good as that of the fractional online algorithm following the allocation $\canalloc(\cdot)$ (defined in \Cref{prop:canonical allocation}, part~(i)) within truncated weight continuum instances $\continstances$.
\end{lemma}

\begin{proof}{\emph{Proof.}} Consider a general adversarial instance with arriving weights $\weight_1, \weight_2, \dots$ and let $\weightmax \triangleq \max_{i} \weight_i$. Define $\ell^*$ as the largest integer satisfying $\weightc_\eta(\ell^*) \leq \weightmax$, and let $\hat{i}_\eta$ be the index of the first online node whose weight exceeds $\weightc_\eta(\ell^*)$ (which need not be the node with maximum weight). Conditioning on the random seed $\eta$—and thus fixing both $\weightc_\eta(\ell^*)$ and $\hat{i}_\eta$—the expected gain from that final node selected by \Cref{alg: randomized algorithm} is exactly $\frac{\weightc_\eta(\ell^*)}{\weight_{\hat{i}_\eta}} \cdot \weight_{\hat{i}_\eta} = \weightc_\eta(\ell^*)$. 
Note that as $\eta$ is uniformly distributed, $\log\left(\weightc_\eta(\ell)\right)$ is uniformly distributed in $\left[\log\left(\weightmax\right)- \log(\thresholdweight), \log\left(\weightmax\right)\right]$. Hence, for $\weight \in \left[\frac{\weightmax}{\thresholdweight},\weightmax\right]$, 
the unconditional expected gain from the final allocation by the algorithm is at least
\[
    \int_{\frac{\weightmax}{\thresholdweight}}^{\weightmax} \weight \canalloc(\weight)\,d\weight,
\]
as this integral captures the expectation of the random variable $\weightc_\eta(\ell^*)$. This is exactly equal the total weight of the fractional allocations made by any algorithm that follows the allocation density $\canalloc(\cdot)$ (e.g., \Cref{alg:primal dual single resource}).

On the other hand, the expected buyback cost incurred by \Cref{alg: randomized algorithm} is at most
\[
    f \int_{1}^{\frac{\weightmax}{\thresholdweight}} \weight \canalloc(\weight)\,d\weight,
\]
since the algorithm only buys back previously selected nodes whose weights lie strictly below $\frac{\weightmax}{\thresholdweight}$. Again, this integral is equal to the total buyback cost of any algorithm that follows the allocation density $\canalloc(\cdot)$.

Combining these observations, we conclude that
\begin{align*}
    \EX[\text{net reward of \Cref{alg: randomized algorithm}}] 
    &\geq \int_{\frac{\weightmax}{\thresholdweight}}^{\weightmax} \weight \canalloc(\weight)\,d\weight 
        - f\int_{0}^{\frac{\weightmax}{\thresholdweight}} \weight \canalloc(\weight)\,d\weight \\
    &= \text{Net reward of \Cref{alg:primal dual single resource} on instance $I_{\weightmax}\in\continstances$}\\
    &\geq  \weightmax / \Gamma,
\end{align*}
where $\Gamma$ is the competitive ratio achieved by the fractional algorithm following allocation density $\canalloc(\cdot)$ on truncated weight continuum instances.
\hfill\halmos
\end{proof}

\subsection{Non-uniform demand and supply in the single-resource environment}
\label{app:non-uniform demand}
In \Cref{sec:single-resource}, we assumed the offline node had capacity $\capacity = 1$, and each online node had unit demand ($\demandi = 1$ for all $i$). In this subsection, we show how our results extend naturally to the general case where $\capacity$ and $\demandi$ can be arbitrary positive numbers. Without loss of generality, we normalize the problem by scaling so that $\capacity = 1$ and $\demandi \leq 1$ for all online nodes $i \in \onlinenodes$.


Let $\largenumber$ be a large positive integer. Moreover, let $\ALG$ be an algorithm for the (integral) single-resource problem with uniform demand. We now propose \Cref{alg:algorithm non-uniform single resource} that uses $\ALG$ in a blackbox fashion and only improves the competitive ratio.  At a high-level the algorithm breaks the offline node into $\largenumber$ small mini nodes with capacity $\frac{1}{\largenumber}$. Then, upon each arrival of an online node, it runs $\ALG$ for $\lfloor \largenumber \demandi \rfloor$ of these mini offline nodes.

\begin{algorithm}
\caption{
Reduction from non-uniform to 
uniform demand in single-resource environment}
\label{alg:algorithm non-uniform single resource}
    \SetKwInOut{Input}{input}
    \SetKwInOut{Output}{output}
 \Input{Integer $\largenumber$, 
 algorithm $\ALG$ for single-resource environment with uniform demand
 }
 
 \vspace{1mm}
 
 \For{each $j \in [\largenumber]$}{
 \vspace{1mm}
 Initialize $\psi_j \gets 0$
 
 }
 
 \vspace{0mm}
 
 \For{each online node $i\in\onlinenodes$}
 {

 \vspace{1mm}
 \For{each $\lfloor \largenumber \demandi  \rfloor$ mini nodes with smallest   $\psi$ like $j$}
 {
 
 \vspace{1mm}
Let $\allocTilde_i$ denote the fraction of capacity that algorithm $\ALG$ allocates to online node $i$ when node $j$ is the offline node whose capacity and allocations are scaled by $K$.
    
    \vspace{1mm}
    
Allocate $\frac{\allocTilde_i}{K}$ fraction of mini node $j$ to online node $i$.

 \vspace{1mm}
Let $\psi_j \gets \max(\weight_i,\psi_j)$.

}
 }
\end{algorithm}

\begin{proposition}
Suppose there exists a positive constant $\epsilon$ such that $d_i > \epsilon$ for all online nodes $i$. Then \Cref{alg:algorithm non-uniform single resource} (with parameter $K\to \infty$) achieves a competitive ratio at least as good as that of algorithm $\ALG$ on general instances of the single-resource environment.
\end{proposition}

\begin{proof}{\emph{Proof.}}
Assume the optimal offline solution is characterized by an allocation density function \( x(w) \). Then,
\begin{align*}
    \OPT = \sum_{w: x(w) > 0} x(w) w.
\end{align*}

Let \( K > \frac{1}{\epsilon} \). Then, each online node with \( x(w) > 0 \) were assigned to at least \( \lfloor \largenumber x(w) \rfloor \) mini nodes. From each such mini node, the algorithm earns a profit of \( \frac{1}{\largenumber} \cdot \frac{w}{\CR} \), since \(\ALG\) obtains \( \frac{w}{\CR} \) in each scaled instance. 

Due to the priority queue structure of the allocation (see line~4 of Algorithm~\ref{alg:algorithm non-uniform single resource}) and the constraint \( \sum_{w: x(w) > 0} x(w) \leq 1 \), the mini nodes assigned to the weights \( \{ w : x(w) > 0 \} \) are all distinct. Therefore,
\begin{align*}
    \ALG &\geq \frac{1}{\largenumber} \sum_{w: x(w) > 0} \lfloor \largenumber x(w) \rfloor \cdot \frac{w}{\CR} \\
    &= \frac{\OPT}{\CR} - \sum_{w: x(w) > 0} \frac{\{ \largenumber x(w) \}}{\largenumber} \cdot \frac{w}{\CR} \\
    &\geq \frac{\OPT}{\CR} - \sum_{w: x(w) > 0} \frac{1}{\largenumber} \cdot \frac{w}{\CR}.
\end{align*}

The second term vanishes as \( \largenumber \rightarrow \infty \), completing the proof.
\hfill\halmos
\end{proof}

\section{Missing Figure}
\revcolor{\Cref{fig:tauLambda} shows the the assignment of correct values for $(\lambda,\tau)$ for each $\buybackcost$.
\begin{figure}[ht]
    \centering
    \begin{tikzpicture}[scale=0.8, transform shape]
\begin{axis}[
axis line style=gray,
axis lines=middle,
xtick style={draw=none},
ytick style={draw=none},
xticklabels=\empty,
yticklabels=\empty,
xmin=-0.0325,xmax=0.5,ymin=-0.45,ymax=4,
width=0.8\textwidth,
height=0.45\textwidth,
samples=50]

\addplot[domain=0:e/2-1, line width=0.5mm] (x, e);


\addplot[line width=0.5mm] coordinates {(0.36 , 2.721718171997723)
(0.365 , 2.7417183148638458)
(0.37 , 2.761719074133534)
(0.375 , 2.7817209530358844)
(0.38 , 2.8017244334445803)
(0.385 , 2.821729976941276)
(0.39 , 2.841738025814615)
(0.395 , 2.8617490039995177)
(0.4 , 2.8817633179609548)
(0.405 , 2.901781357526122)
(0.41000000000000003 , 2.921803496668576)
(0.41500000000000004 , 2.9418300942476328)
(0.42 , 2.9618614947060444)
(0.425 , 2.9818980287287595)
(0.43 , 3.001940013865324)
(0.435 , 3.0219877551183187)
(0.44 , 3.0420415455000134)
(0.445 , 3.062101666559288)
(0.45 , 3.082168388880685)
(0.455 , 3.1022419725573527)
(0.46 , 3.1223226676394833)
(0.465 , 3.1424107145597664)
(0.47000000000000003 , 3.1625063445372246)
(0.47500000000000003 , 3.1826097799607678)
(0.48 , 3.2027212347536236)
(0.485 , 3.222840914719823)
(0.49 , 3.2429690178737443)
(0.495 , 3.263105734753704)
(0.5 , 3.2832512487205294)
};

\addplot[gray!70!white, dashed, line width=0.5mm] coordinates {
(0.0 , 0.5819767068693265)
(0.005 , 0.5865935091974412)
(0.01 , 0.5912373375247281)
(0.015 , 0.5959084298564189)
(0.02 , 0.6006070270006412)
(0.025 , 0.6053333726097984)
(0.03 , 0.6100877132226896)
(0.035 , 0.6148702983073769)
(0.04 , 0.6196813803048211)
(0.045 , 0.624521214673298)
(0.05 , 0.6293900599336155)
(0.055 , 0.6342881777151437)
(0.06 , 0.6392158328026804)
(0.065 , 0.6441732931841644)
(0.07 , 0.6491608300992602)
(0.075 , 0.6541787180888259)
(0.08 , 0.659227235045291)
(0.085 , 0.6643066622639563)
(0.09 , 0.6694172844952412)
(0.095 , 0.6745593899978943)
(0.1 , 0.6797332705931936)
(0.105 , 0.6849392217201506)
(0.11 , 0.6901775424917489)
(0.115 , 0.6954485357522294)
(0.12 , 0.7007525081354571)
(0.125 , 0.7060897701243807)
(0.13 , 0.7114606361116205)
(0.135 , 0.7168654244612012)
(0.14 , 0.7223044575714584)
(0.145 , 0.7277780619391462)
(0.15 , 0.7332865682247696)
(0.155 , 0.7388303113191715)
(0.16 , 0.744409630411401)
(0.165 , 0.7500248690578931)
(0.17 , 0.755676375252988)
(0.17500000000000002 , 0.7613645015008234)
(0.18 , 0.7670896048886245)
(0.185 , 0.7728520471614343)
(0.19 , 0.7786521947983035)
(0.195 , 0.7844904190899884)
(0.2 , 0.7903670962181771)
(0.20500000000000002 , 0.7962826073362922)
(0.21 , 0.8022373386518961)
(0.215 , 0.8082316815107442)
(0.22 , 0.8142660324825185)
(0.225 , 0.8203407934482859)
(0.23 , 0.8264563716897169)
(0.23500000000000001 , 0.8326131799801118)
(0.24 , 0.8388116366772708)
(0.245 , 0.8450521658182586)
(0.25 , 0.851335197216102)
(0.255 , 0.8576611665584729)
(0.26 , 0.8640305155083994)
(0.265 , 0.8704436918070563)
(0.27 , 0.8769011493786849)
(0.275 , 0.8834033484376953)
(0.28 , 0.8899507555979999)
(0.28500000000000003 , 0.8965438439846363)
(0.29 , 0.9031830933477355)
(0.295 , 0.9098689901788931)
(0.3 , 0.9166020278299994)
(0.305 , 0.9233827066345932)
(0.31 , 0.9302115340318026)
(0.315 , 0.9370890246929313)
(0.32 , 0.9440157006507663)
(0.325 , 0.9509920914316637)
(0.33 , 0.9580187341904949)
(0.335 , 0.9650961738485133)
(0.34 , 0.9722249632342284)
(0.34500000000000003 , 0.9794056632273507)
(0.35000000000000003 , 0.9866388429058991)
(0.355 , 0.9939250796965391)
(0.36 , 0.9987382323060267)
(0.365 , 0.9914882262135047)
(0.37 , 0.984394067353675)
(0.375 , 0.9774504296908163)
(0.38 , 0.9706522357898221)
(0.385 , 0.9639946421585732)
(0.39 , 0.9574730256307973)
(0.395 , 0.9510829707033219)
(0.4 , 0.9448202577497538)
(0.405 , 0.9386808520398608)
(0.41000000000000003 , 0.9326608935004375)
(0.41500000000000004 , 0.9267566871592622)
(0.42 , 0.9209646942189983)
(0.425 , 0.9152815237125985)
(0.43 , 0.9097039246960191)
(0.435 , 0.9042287789378772)
(0.44 , 0.8988530940691439)
(0.445 , 0.8935739971590846)
(0.45 , 0.8883887286864972)
(0.455 , 0.88329463687785)
(0.46 , 0.8782891723862588)
(0.465 , 0.873369883287342)
(0.47000000000000003 , 0.8685344103699274)
(0.47500000000000003 , 0.8637804827013161)
(0.48 , 0.8591059134484188)
(0.485 , 0.8545085959375132)
(0.49 , 0.8499864999367124)
(0.495 , 0.8455376681464428)
(0.5 , 0.8411602128843254)
};

\addplot[mark=*,only marks, fill=white] coordinates {(0.3591409142295226176801437356763312488786235468499797874834838138,e)} node[above, pos=1]{};

\addplot[mark=*,only marks, fill=white] coordinates {(0.3591409142295226176801437356763312488786235468499797874834838138,1)} node[above, pos=1]{};

\addplot[gray, dotted] coordinates {(0.3591409142295226176801437356763312488786235468499797874834838138,e)
(0.3591409142295226176801437356763312488786235468499797874834838138,0.)};

\addplot[gray, thick] coordinates {(0.3591409142295226176801437356763312488786235468499797874834838138,0.01)
(0.3591409142295226176801437356763312488786235468499797874834838138,-0.01)};

\addplot[] coordinates {(0.3591409142295226176801437356763312488786235468499797874834838138,0.)} node[below, pos=1]{$\frac{e-2}{2}$};

\addplot[gray, thick] coordinates {(-0.0025, 1) (0.0025, 1)};
\addplot[] coordinates {(-0.007,1)} node[left, pos=1]{$1$};

\addplot[gray, thick] coordinates {(-0.0025, e) (0.0025, e)};
\addplot[] coordinates {(-0.007,e)}  node[left, pos=1]{$e$};

\addplot[] coordinates {(0.09,3.5)}  node[left, pos=1]{$\lambda(f),\tau(f)$};

\addplot[] coordinates {(0.48,0.)} node[below, pos=1]{$f$};
\addplot[gray, dotted] coordinates {(0,1)
(0.3591409142295226176801437356763312488786235468499797874834838138,1)};

\end{axis}

\end{tikzpicture}
    \caption{\revcolor{The choice of $\lambda$ (black solid line) and $\tau$ (gray dashed line) for different $\buybackcost$. Assignment $\lambda(\buybackcost)$ is fixed at $e$ for $f \leq \frac{e-2}{2}$ and increasing (and strictly convex) after that. Assignment $\tau(\buybackcost)$ on the other hand is increasing before the change in the problem regime and then it decreases until it reaches $0$ at infinity.}
    } 
    \label{fig:tauLambda}
\end{figure}
}

\section{Missing Proofs}
\label{sec:apx-missing}

\subsection{Proof of Theorem~\ref{thm:lower bound all f} for the small buyback regime
\texorpdfstring{{($\boldsymbol{\buybackcost \leq \fracthreshold}$)}}{}}
\label{sec:apx-missing lower bound matching small f}

\lowerboundmatching*

\oldcolor{
Before presenting the proof of \Cref{thm:lower bound all f}, we first develop the following characterization of the optimal online algorithm. The proof of this characterization can be found at the end of this subsection.}

\begin{restatable}{lemma}{lowerboundmatchingoptmumonline}
\label{lem:lower bound matching optimum online}
In \Cref{example:lower bound matching},
for any buyback factor $\buybackcost\leq \fracthreshold$,
the optimal online algorithm satisfies
\begin{enumerate}
    \item \underline{Symmetric-allocation}:
    for each online node $i\in[\largenumber]$,
    the algorithm 
    allocates equal fractions of online node $i$ 
    to all offline nodes
    with non-zero edge-weights.
    \item \underline{Fully-allocation}:
    there exists $\fullallocthreshold\in\naturals$ such that
    for each online node $i\in [\largenumber - \fullallocthreshold]$,
    the algorithm 
    allocates
    the whole unit of online node $i$
    to offline nodes with non-zero edge-weights.
\end{enumerate}
\end{restatable}

\oldcolor{
Equipped with the characterization of the optimal online algorithm, we are ready to prove \Cref{thm:lower bound all f} for the small buyback regime using \Cref{example:lower bound matching}.}

\begin{proof}{\emph{Proof for the small buyback regime ($\buybackcost \leq  \fracthreshold$).}}
Consider the \revcolorCMTwo{offline optimum} benchmark and 
the optimal competitive online algorithm in \Cref{example:lower bound matching}.
By construction, it is straightforward to verify that the \revcolorCMTwo{offline optimum} matches 
each online node $i$ to offline node $\permu^{-1}(i)$ and
collects total profit $\sum_{i\in[\largenumber]}\sfrac{1}{i}$, i.e.,
\begin{align*}
    \OPT(\text{\Cref{example:lower bound matching}})
    = 
    \sum_{\ell\in[\largenumber]}\frac{1}{\ell}
\end{align*}
By \Cref{lem:lower bound matching optimum online},
in the optimal online algorithm,
for every online node $i\in[\largenumber-\fullallocthreshold]$,
the algorithms allocates $\frac{1}{\largenumber - i + 1}$
to each of $\largenumber - i + 1$ offline nodes with non-zero
edge weights.
Therefore, the total profit
induced by every online node $i\in[\largenumber - \nxt(\fullallocthreshold)]$ is at most 
$\frac{1}{\largenumber - i + 1}\left(1 - (1+\buybackcost)\frac{\pre(\largenumber - i + 1)}{\largenumber - i + 1}\right)$.\footnote{
The weight of online node $i$ is $\frac{1}{\largenumber - i + 1}$.
There are $\pre(\largenumber - i + 1)$ offline nodes $j$,
which the algorithm will buyback its allocation from online node $i$ to offline node $j$ in the future.}
The total profit induced by every online node $i\in[\largenumber - \nxt(\fullallocthreshold) + 1:\largenumber]$ is at 
most $\frac{1}{\largenumber - i + 1}$.
Putting two pieces together,
the expected total profit of the 
optimal online algorithm $\ALG^*$ 
is 
at most
\begin{align*}
    \ALG^*(\text{\Cref{example:lower bound matching}})
    &\leq
    \sum_{i\in[\largenumber - \nxt(\fullallocthreshold)]}
    \frac{1}{\largenumber - i + 1}\left(1 - (1+\buybackcost)\frac{\pre(\largenumber - i + 1)}{\largenumber - i + 1}\right)
    +
    \sum_{i\in[\largenumber - \nxt(\fullallocthreshold) + 1:\largenumber]}\frac{1}{\largenumber - i + 1}
    \\
    &=
    \sum_{\ell =1}^{\nxt(\fullallocthreshold)}
    \frac{1}{\ell} 
    +
    \sum_{\ell = \nxt(\fullallocthreshold)}^{\largenumber}
    \frac{1}{\ell}\left(
    1 - (1 + \buybackcost)\frac{\pre(\ell)}{\ell}\right)
    \\
    &\overset{}{\leq} 
    \sum_{\ell =1}^{\nxt(\largenumber_1(\buybackcost,\varepsilon))}
    \frac{1}{\ell} 
    +
    \sum_{\ell = \nxt(\largenumber_1(\buybackcost,\varepsilon))}^{\largenumber}
    \frac{1}{\ell}\left(
    1 - (1 + \buybackcost)\left(\frac{1}{e} - \varepsilon\right)\right)
\end{align*}
where $\varepsilon>0$ is an arbitrary positive constant,
and $\largenumber_1(\buybackcost,\varepsilon) \geq \fullallocthreshold$
is a constant such that $\frac{\pre(\ell)}{\ell} \geq 
\frac{1}{e}-\varepsilon$ for every $\ell \geq \largenumber_1(\buybackcost,\varepsilon)$.
The existence of constant $\largenumber_1(\buybackcost,\varepsilon)$ 
is guaranteed by \Cref{lem:pre suc fact}. 
Notably,
constant $\largenumber_1(\buybackcost,\varepsilon)$
is independent of $\largenumber$.

Finally, to lowerbound the optimal competitive ratio $\optCRgen$,
note that 
\begin{align*}
    \frac{1}{\optCRgen} &\leq 
    \frac{\ALG^*(\text{\Cref{example:lower bound matching}})}
    {\OPT(\text{\Cref{example:lower bound matching}})}
    \\
    &\leq \displaystyle
    \frac{
    \sum_{\ell =1}^{\nxt(\largenumber_1(\buybackcost,\varepsilon))}
    \frac{1}{\ell} 
    +
    \sum_{\ell = \nxt(\largenumber_1(\buybackcost,\varepsilon))}^{\largenumber}
    \frac{1}{\ell}\left(
    1 - (1 + \buybackcost)\left(\frac{1}{e} - \varepsilon\right)\right)
    }{
    \sum_{\ell\in[\largenumber]}
    \frac{1}{\ell}}
    \\
    &\overset{}{\leq} 
    o(1) + 1 - (1 + \buybackcost)\left(\frac{1}{e} - \varepsilon\right)
\end{align*}
where the last inequality holds since 
$\displaystyle\frac{\sum_{\ell = \nxt(\largenumber_1(\buybackcost,\varepsilon))}^{\largenumber}
    \frac{1}{\ell}}{
     \sum_{\ell\in[\largenumber]}\frac{1}{\ell}} = o(1)$ when 
     we let $\largenumber$ go to infinite
     and hold $\largenumber_1(\buybackcost,\varepsilon)$ as constant.
     \hfill\halmos
\end{proof}

We conclude this subsection by proving \Cref{lem:lower bound matching optimum online}.
We first introduce two auxiliary functions 
and prove one related technical lemma 
which will be used 
in the final analysis.
Define function $\pre:\naturals\rightarrow\naturals$ and 
function $\nxt:\naturals\rightarrow\naturals$ where:
\begin{align*}
    \forall i\in\naturals:
    \qquad
    \pre(i) = \max
    \left\{i'\in\naturals:\sum_{\ell=i'}^i \frac{1}{\ell} > 1
    \right\}
    \qquad \textrm{and} \qquad
    \nxt(i) = \min
    \left\{i'\in\naturals:\sum_{\ell=i}^{i'} \frac{1}{\ell} > 1~
    \right\}~~.
\end{align*}
As a sanity check, note that $\pre$ is the inverse function of $\nxt$ by definition. Now we have the following lemma regarding the asymptotic behaviours of these two functions. 
\begin{lemma}
\label{lem:pre suc fact}
$\lim\limits_{i\rightarrow\infty}
\frac{\pre(i)}{i} = \frac{1}{e}$,
$\lim\limits_{i\rightarrow\infty}
\frac{\nxt(i)}{i} = e$,
and 
$\lim\limits_{i\rightarrow \infty}
\sum\nolimits_{i'=\pre(i) + 1}^{i} \frac{1}{\nxt(i')} \leq \frac{1}{e}$~.
\end{lemma}
\begin{proof}{\emph{Proof.}}
We first show $\lim\limits_{i\rightarrow\infty}
\frac{\pre(i)}{i} = \frac{1}{e}$,
which is equivalent to $\lim\limits_{i\rightarrow\infty}
\frac{i}{\pre(i)} = {e}$.
On one hand, 
\begin{align*}
    \log\left(
    \frac{i}{\pre(i) - 1}
    \right)
    =
    \displaystyle\int_{\pre(i) - 1}^i \frac{1}{\ell} \,d\ell
    \geq 
    \displaystyle\sum_{\ell=\pre(i)}^{i} \frac{1}{\ell} 
    \geq 1~~.
\end{align*}
Similarly, 
\begin{align*}
     \log\left(
    \frac{i+1}{\pre(i)}
    \right)
    =
    \displaystyle\int_{\pre(i)}^{i+1} \frac{1}{\ell} \,d\ell
    \leq 
    \displaystyle\sum_{\ell=\pre(i) + 1}^{i} \frac{1}{\ell} 
    < 1~~.
\end{align*}
Combining the above inequalities with the fact that 
\begin{align*}
    \lim\limits_{i\rightarrow \infty}
    \frac{i}{\pre(i)}
    =
    \lim\limits_{i\rightarrow \infty}
    \frac{i+1}{\pre(i)}
    =
    \lim\limits_{i\rightarrow \infty}
    \frac{i}{\pre(i)-1}~,
\end{align*}
we prove $\lim\limits_{i\rightarrow\infty}
\frac{\pre(i)}{i} = \frac{1}{e}$ by the sandwich theorem, as desired.
Since $\suc$ is the inverse function of $\pre$,
we have $\lim\limits_{i\rightarrow\infty}
\frac{\nxt(i)}{i} = e$ as well.

Finally, we show $\lim\limits_{i\rightarrow \infty}
\sum\nolimits_{i'=\pre(i) + 1}^{i} \frac{1}{\nxt(i')} \leq \frac{1}{e}$.
As we shown above, for every $\varepsilon>0$,
there exists $N\in\naturals$ such that 
$\frac{\nxt(i)}{i} \geq \frac{1}{\frac{1}{e} + \varepsilon}$
for every $i\geq N$.
Therefore, 
\begin{align*}
    \lim\limits_{i\rightarrow \infty}
\sum\nolimits_{i'=\pre(i) + 1}^{i} \frac{1}{\nxt(i')}
\leq 
\lim\limits_{i\rightarrow \infty}
\sum\nolimits_{i'=\pre(i) + 1}^{i} 
\left(\frac{1}{e}+\varepsilon\right)\frac{1}{i}
\leq 
\frac{1}{e}+\varepsilon
\end{align*}
Letting $\varepsilon$ goes to zero finishes the proof.
\hfill\halmos
\end{proof}
We are ready to prove \Cref{lem:lower bound matching optimum online}.

\begin{proof}{\emph{Proof of \Cref{lem:lower bound matching optimum online}.}}
We first show the symmetric-allocation property.
In particular, we show how to
convert an arbitrary online algorithm $\ALG$
to an online algorithm $\ALG\primed$ that satisfies the symmetric-allocation
property and achieves the same profit:
Fix an arbitrary online algorithm $\ALG$.
Let $\edgeallocijp$ be the expected allocation
between online node $i$ and offline node $j$ under 
permutation $\permu$ in $\ALG$.
Note that for every $i,j,j'\in[\largenumber]$
\begin{align*}
    \expect[\permu]{\edgeallocijp\condition \weightij > 0}
    =
    \expect[\permu]{\edgealloc_{ij'}^\permu\condition \weight_{ij'}> 0}
\end{align*}
which holds due to the symmetry of
our bipartite graph construction, i.e., conditioned on the event 
$[\weightij > 0\land \weight_{ij'}> 0]$,
it is an automorphism
that exchanges $j$ and $j'$.
Therefore, we construct algorithm $\ALG\primed$
by simulating algorithm $\ALG$.
For each online node $i$, 
let $\ALG\primed$ allocates 
 $\expect[\permu]{\edgeallocijp\condition \weightij > 0}$
 to each offline node $j$ with non-zero edge weight,
 and allocates nothing to each offline node $j$ with zero edge weight.
 By construction, it is straightforward to verify that
 $\ALG\primed$ satisfies the symmetric-allocation property
 and achieves the same profit as $\ALG$.
 
Next, we show the fully-allocation property.
Fix an arbitrary online node $i\in[\largenumber]$.
Let $k\triangleq \largenumber - i + 1$.
Note that there are $k$ offline nodes $j$ with non-zero edge weight $\weightij > 0$.
Suppose the optimal online algorithm allocates $k\cdot dx$ amount of online node $i$ to offline nodes. 
By the symmetric-allocation property,
each offline node $j$ with non-zero edge weight
receives $dx$ amount.
Given the definition of function 
$\pre(\cdot)$ and the fact that the optimal
algorithm always greedily buys back the smallest allocated weight,
there are $\pre (k)$ offline nodes with non-zero edge weight,
each of which has total allocation exceeding one in the future.\footnote{Namely, suppose the algorithm 
fully allocates each future online node $i'> i$ 
with equal fractions to offline nodes with non-zero edge weight,
then
the total amount allocated to offline node $j>\largeconstant-\pre(k)$ exceeds one.}
Thus, these $\pre (k)$ offline nodes 
are the ones 
which the algorithm may buyback their $dx$ amount of online node $i$ in the future. 
As a consequence, by allocating this $k\cdot dx$ amount of online node $i$,
the algorithm incurs an additional buyback cost at most 
$
\pre(k)\cdot \left(\frac{\buybackcost}{k}\right)dx$.\footnote{There are $\pre(k)$ offline nodes which may be buyback in the future.
Each offline node receives $dx$ amount of online node with edge weight $\frac{1}{k}$.}
 On the other hand,
 each of
 the remaining $k - \pre(k)$
 offline nodes with non-zero edge weight
 has total allocation less than one in the future.
 Consider the allocated weights of these offline nodes 
 in the end, the algorithm makes a marginal profit of at least
$\sum_{\ell = \pre(k) + 1}^{k}\left(\frac{1}{k} 
- \frac{1+\buybackcost}{\nxt(\ell)}\right)dx$.
Putting all pieces together, the increase of the final profit
by allocating extra $k\cdot dx$ amount of online node $i$ is 
at least
\begin{align*}
    -\pre(k)\cdot \left(\frac{\buybackcost}{k}\right)dx
    +
    \sum_{\ell = \pre(k) + 1}^{k}\left(\frac{1}{k} 
- \frac{1+\buybackcost}{\nxt(\ell)}\right)dx
\end{align*}
By algebra, this increase of the final profit is non-negative if 
\begin{align*}
    1 - (1 + \buybackcost)\left(
    \frac{\pre(k)}{k}
    +
     \sum_{\ell = \pre(k) + 1}^{k}
   \frac{1}{\nxt(\ell)}
    \right)
\geq 0
\end{align*}
Invoking \Cref{lem:pre suc fact},
for every $\varepsilon > 0$,
there exists $N(\varepsilon)\in\naturals$ such that 
    for every $k\geq N(\varepsilon)$,
$\left(
    \frac{\pre(k)}{k}
    +
     \sum_{\ell = \pre(k) + 1}^{k}
   \frac{1}{\nxt(\ell)}
    \right) \leq \frac{2+\varepsilon}{e}$.
Thus, for every buyback factor $\buybackcost < \fracthreshold$,
let $\varepsilon(\buybackcost)$ be the number such that $\buybackcost \leq \frac{e - 2 - \varepsilon(\buybackcost)}{2 + \varepsilon(\buybackcost)}$.
Setting $\fullallocthreshold \triangleq N(\varepsilon(\buybackcost))$
finishes the proof.
\hfill\halmos
\end{proof}

\subsection{Proof of Proposition~\ref{prop:lower bound single resource}; Theorem~\ref{thm:lower bound all f} for the large buyback regime
\texorpdfstring{($\boldsymbol{\buybackcost \geq \fracthreshold}$)}{}}
\label{sec:apx-missing lower bound single resource}
\lowerboundsingleresource*

To prove \Cref{prop:lower bound single resource}, we first use a lemma in \cite{AK-09} to establish the same lower bound on competitive ratio of integral algorithms using a randomized
truncated weight continuum instance, defined in \Cref{example:lower bound single resource} below. We then show that we can extend it to online fractional algorithms by combining it with the fractional-to-integral rounding for truncated weight continuum instances.

\begin{lemma}[\citealp{AK-09}]
\label{lem:lower bound single resource integral}
In \Cref{example:lower bound single resource},
the expected profit in
the \revcolorCMTwo{offline optimum} benchmark is
$\log(\totaltime_0) + 1$,
and the expected profit in 
the optimal integral algorithm is 
$\max_{\weight\geq 1} 1 + (k(\weight) - 1)\cdot 
\frac{\weight - (1 + \buybackcost)}{\weight}$
where $k(\weight)= \max\{\ell\in\naturals:\weight^{\ell-1} \leq \totaltime_0\}$.
\end{lemma}

\begin{lemma}
\label{lem:lower bound reduction main body1}
In the single-resource environment, for any online algorithm $\ALG$ with the competitive ratio $\approxratio$
within the truncated weight continuum instances $\continstances$,
there exists a randomized integral online algorithm $\ALG\primed$ with the same competitive ratio $\approxratio$
within the truncated weight continuum instances $\continstances$.
\end{lemma}
\begin{proof}{\emph{Proof.}}
Fix an arbitrary online algorithm $\ALG$
for the truncated weight continuum instances $\continstances$.
Note that every truncated weight continuum instance $\instance_T$ looks exactly the same 
before its termination. Thus, 
let $\allocTilde(\weight)$
be the allocated density/fraction
when online node $\weight\in\reals_+$ arrives 
in algorithm $\ALG$.
For any $\randseed \in [0, 1]$,
define function $\weight_\randseed: \naturals\rightarrow \reals_+$
where for every $\ell\in\naturals$,
$\weight_\randseed(\ell)$
is the solution $\weight$ such 
that $\int_0^{\weight}\allocTilde(t)\,dt = \ell + \randseed-1$.\footnote{If there exists no $\weight$
such that $\int_0^{\weight}\allocTilde(t)\,dt = \ell + \randseed-1$,
we let $\weight_\randseed(\ell)=\infty$.}

Now we construct a randomized integral algorithm $\ALG\primed$ as follows.
The algorithm samples $\randseed$ from $[0, 1]$ uniformly at random. 
Then it allocates the offline node 
to each online node $\weight \in \{\weight_\randseed(\ell):\ell\in\naturals\}$.
By construction, $\ALG\primed$
is a randomized integral algorithm.
We just need to prove the probability that $\ALG\primed$
allocates the offline node to an online node in the interval $[\underline{\weight},\overline{\weight}]$
is exactly $\min(1,\int_{\underline{\weight}}^{\overline{\weight}} \allocTilde(\weight) \,d\weight)$
for each truncated weight continuum instance $\instance_\totaltime$ with $\totaltime \geq \overline{\weight}$. First of all if $\int_{\underline{\weight}}^{\overline{\weight}} \allocTilde(\weight) \,d\weight > 1$ it means that there exists a $\weight \in [\underline{\weight},\overline{\weight}]$ and $\ell \in \naturals$ such that $\int_0^{\weight}\allocTilde(t)\,dt = \ell + \randseed-1$. Which means $\ALG\primed$ always allocate to the online node $\weight \in [\underline{\weight},\overline{\weight}]$. Otherwise, $\ALG\primed$ allocates to an online node in the interval $[\underline{\weight},\overline{\weight}]$ if and only if there exists an $\ell \in \naturals$ with $\int_0^{\underline{\weight}} \allocTilde(\weight) \,d\weight < \ell + \eta - 1 < \int_0^{\overline{\weight}} \allocTilde(\weight) \,d\weight$. Since $\eta$ is uniformly sampled, this happens with probability $ \int_0^{\overline{\weight}} \allocTilde(\weight) \,d\weight - \int_0^{\underline{\weight}} \allocTilde(\weight) \,d\weight = \int_{\underline{\weight}}^{\overline{\weight}} \allocTilde(\weight) \,d\weight $. Furthermore, notice that in both cases the arrival $\weight$ is bought back if and only if $\int_{\weight}^{T} \allocTilde(\weight) \,d\weight \geq 1$.
Thus, the competitive ratio (within 
truncated weight continuum instances $\continstances$)
of the constructed randomized integral online algorithm $\ALG\primed$
is the same as algorithm $\ALG$.
\hfill\halmos
\end{proof}

Now we are ready to prove \Cref{prop:lower bound single resource}.

\begin{proof}{\emph{Proof of \Cref{prop:lower bound single resource}.}}
To lower-bound the competitive ratio $\optCRgen$, by \Cref{lem:lower bound reduction main body1}, it is enough to find a lower-bound on the competitive ratio of integral algorithms within the class of truncated weight continuum instances. Now consider  \Cref{example:lower bound single resource} which is in this class. 
By 
\Cref{lem:lower bound single resource integral},
the expected total profit of the \revcolorCMTwo{offline optimum} benchmark as well as the optimal online integral algorithm $\ALG^*$
is 
\begin{align*}
     \OPT(\text{\Cref{example:lower bound single resource}}) &=  \log(\totaltime_0) + 1 \\
     \ALG^*(\text{\Cref{example:lower bound single resource}}) &= \max_{\weight\geq 1} 1 + (k(\weight) - 1)\cdot 
\frac{\weight - (1 + \buybackcost)}{\weight}
\end{align*}
where $k(\weight)= \max\{\ell\in\naturals:\weight^{\ell-1} \leq \totaltime_0\}$. Let $\weight^* \triangleq \argmax_{\weight\geq 1}\max_{\weight\geq 1} 1 + (k(\weight) - 1)\cdot 
\frac{\weight - (1 + \buybackcost)}{\weight}$,
and $k^*\triangleq k(\weight^*)$.
Now note that 
\begin{align*}
    \frac{1}{\optCRgen} 
    &\leq 
    \frac{
    \ALG^*(\text{\Cref{example:lower bound single resource}})
    }{
    \OPT(\text{\Cref{example:lower bound single resource}})}
    \\
    &=
    \frac{
     1 + (k^* - 1)\cdot 
\frac{\weight^* - (1 + \buybackcost)}{\weight^*}
    }{
    \log(\totaltime_0) + 1
    }
    \\
    &\leq 
    \frac{1}{\log(\totaltime_0)}
    +
    \frac{
     (k^* - 1)\cdot 
\frac{\weight^* - (1 + \buybackcost)}{\weight^*}
    }{
    (k^*-1)\log(\weight^*) 
    }
    \\
    &\leq 
    \frac{1}{\log(\totaltime_0)}
    +
    \max_{a\geq 1}
    \frac{
     {a - (1 + \buybackcost)}
    }{
    a\log(a) 
    }
    \\
    &\leq 
    \frac{1}{\log(\totaltime_0)}
    -
    \frac{1}{\Lambertterm}
\end{align*}
Finally, letting $\totaltime_0$ go to infinite finishes the proof.
\hfill\halmos
\end{proof}

\oldcolor{\subsection{Proof of Theorem~\ref{thm:lower bound deterministic all f} for the small buyback regime
\texorpdfstring{{($\boldsymbol{\buybackcost \leq \detthreshold}$)}}{}}}
\label{sec:apx-missing matching lower determinstic}

\oldcolor{
\lowerboundmatchingdet*

We prove the theorem statement for the small buyback regime using \Cref{example:lower bound Deterministic matching}.
}

\begin{proof}{\emph{Proof for the small buyback regime ($\buybackcost \leq \detthreshold$).}}
By way of contradiction assume $\ALG$ has a competitive ratio better than $\approxratio$ where $1 + 2\buybackcost + 
2\sqrt{\buybackcost(1+\buybackcost)} < \approxratio < \frac{2}{1-\buybackcost}$. 
Then for all $i$ and \oldcolor{$\epsilon > 0$} we know
$$\frac{\weight_i \oldcolor{-\epsilon} +\weight_0}{\weight_{i-1}-\buybackcost(\weight_0+\weight_1+...+\weight_{i-2})}\leq \approxratio.$$
\oldcolor{Because of continuity in $\epsilon$, we may take the limit as $\epsilon \rightarrow 0$ and the inequality also holds when $\epsilon = 0$.} Let $\ell$ be the first index such that left hand side is strictly smaller.
Let $\rho = \frac{\weight_\ell + \weight_0}{\approxratio(\weight_{\ell-1} - \buybackcost(\weight_0+\weight_1+...+\weight_{\ell-2}))}\oldcolor{<1}$.
Define a new sequence $\{z_i\}_{i=0}^\infty$ such that $z_i = \rho \weight_i$ for $i < \ell$ and $z_i=\weight_i$ otherwise. This new sequence still satisfies
$$\frac{z_i + z_0} {z_{i-1} - \buybackcost(z_0+z_1+...+z_{i-2})}\leq \approxratio,$$
\oldcolor{for all $i$, but for $i \leq \ell$ specifically, the inequality holds with equality:}
$$\frac{z_i + z_0} {z_{i-1}- \buybackcost(z_0+z_1+...+z_{i-2})} = \approxratio.$$
Continuing this process for all positive number $M$ one can find a sequence $\{z_i\}_{i=0}^M$ such that
$$\frac{z_i+z_0}{z_{i-1}-\buybackcost(z_0+z_1+...+z_{i-2})} = \approxratio,$$
or
$$z_i = (\approxratio +1)z_{i-1} -\approxratio (1+\buybackcost) z_{i-2}.$$
Solving this we get:
$$z_i = C_1 \left(\frac{\approxratio+1 + \sqrt{(\approxratio+1)^2 -4\approxratio(1+\buybackcost)}}{2}\right)^i + C_2 \left(\frac{\approxratio+1 - \sqrt{(\approxratio+1)^2 -4\approxratio(1+\buybackcost)}}{2}\right)^i.$$
Given that $\approxratio > 1 + 2\buybackcost + 
2\sqrt{\buybackcost(1+\buybackcost)}$ both roots are real and positive numbers. Looking at the first terms:
$$z_0 = C_1+C_2$$
$$(\approxratio - 1)z_0 = z_1 = (C_1+C_2)\frac{\approxratio+1}{2} + (C_1-C_2)\frac{\sqrt{(\approxratio+1)^2 -4\approxratio(1+\buybackcost)}}{2}$$
We can normalize the sequence by setting $z_0 = 1$. Then $$C_1 = \frac{\approxratio-3}{2\sqrt{(\approxratio+1)^2-4\approxratio (1+\buybackcost)}} + \frac{1}{2}$$ 
Since $\buybackcost < \frac{1}{3}$ \oldcolor{by our assumption $\approxratio < \frac{2}{1-f} \leq 3$}, so
$$C_1 \geq 0 \iff 1 \geq \frac{3 - \approxratio}{\sqrt{(\approxratio+1)^2-4\approxratio (1+\buybackcost)}} \iff \approxratio \geq \frac{2}{1-\buybackcost}$$
This implies that $C_1 < 0$ in our setting. But we know
$$z_i \geq 0 \iff C_1 + C_2\left(\frac{\approxratio+1 - \sqrt{(\approxratio+1)^2 -4\approxratio(1+\buybackcost)}}{\approxratio+1 + \sqrt{(\approxratio+1)^2 -4\approxratio(1+\buybackcost)}}\right)^i \geq 0$$
\oldcolor{Clearly, this relationship cannot hold as} $i \rightarrow \infty$.
\hfill
\halmos
\end{proof}

\oldcolor{\subsection{Proof of Theorem~\ref{thm:lower bound deterministic all f} for the large buyback regime
\texorpdfstring{{($\boldsymbol{\buybackcost \geq \detthreshold}$)}}{}}}
\label{sec:apx-missing single lower determinstic}



\oldcolor{
\lowerboundmatchingdet*

We prove the theorem statement for the large buyback regime using \Cref{example:lower bound Deterministic single-resource}. Both the example and its analysis are adopted from \citet{BHK-09}, but we include them here for completeness.
}

\begin{proof}{\emph{Proof for the large buyback regime ($\buybackcost \geq \detthreshold$).}}
Suppose an algorithm has a competitive ratio $\approxratio$ which is strictly smaller than $1 + 2\buybackcost + 2 \sqrt{\buybackcost(1+\buybackcost)}$. \oldcolor{For all $\epsilon > 0$}, we have:
$$\frac{\weight_i -\oldcolor{\epsilon}}{\weight_{i-1}-\buybackcost(\weight_0+\weight_1+...+\weight_{i-2})}\leq \approxratio.$$
\oldcolor{Because of continuity in $\epsilon$, we may take the limit as $\epsilon \rightarrow 0$ and the inequality also holds when $\epsilon = 0$.} Let $\ell$ be the first index such that left hand side is strictly smaller.
Let $\rho = \frac{\weight_\ell}{\approxratio(\weight_{\ell-1} - \buybackcost(\weight_0+\weight_1+...+\weight_{\ell-2}))}\oldcolor{<1}$.
Define a new sequence $\{z_i\}_{i=0}^\infty$ such that $z_i = \rho \weight_i$ for $i < \ell$ and $z_i=\weight_i$ otherwise. This new sequence still satisfies
$$\frac{z_i}{z_{i-1}-\buybackcost(z_0+z_1+...+z_{i-2})}\leq \approxratio,$$
\oldcolor{for all $i$, but for $i \leq \ell$ specifically, the inequality holds with equality:}
$$\frac{z_i}{z_{i-1}-\buybackcost(z_0+z_1+...+z_{i-2})} = \approxratio.$$
Therefore, for all positive number $M$ one can find a sequence $\{z_i\}_{i=0}^M$ such that
$$\frac{z_n}{z_{n-1}-\buybackcost(z_0+z_1+...+z_{n-2})} = \approxratio.$$
This means:
$$z_n = (\approxratio +1)z_{n-1} -\approxratio (1+\buybackcost) z_{n-2},$$
and solving this recursion we get:
$$z_n = C_1 \left(\frac{\approxratio+1 + \sqrt{(\approxratio+1)^2 -4\approxratio(1+\buybackcost)}}{2}\right)^n + C_2 \left(\frac{\approxratio+1 - \sqrt{(\approxratio+1)^2 -4\approxratio(1+\buybackcost)}}{2}\right)^n.$$
Since $1 < \approxratio < 1 + 2\buybackcost + 
2\sqrt{\buybackcost(1+\buybackcost)}$ both roots are non-real complex numbers. Using the fact that $z_i \in \reals$ we can get $\Bar{C_1} = C_2$ the equation can be rewritten as:
$$z_n = ar^ne^{in\theta} + \Bar{a}r^ne^{-in\theta} = 2r^n\Re (ae^{in\theta}) = 2|a|r^n \cos (\phi + n\theta)$$
where $0 < \theta < \pi$. Let $M > \frac{2 \pi}{\theta}$ then there will be at least one positive number $n<M$ with $(2m+\frac{1}{2})\pi<\phi + n\theta < (2m+\frac{3}{2})\pi$ 
which means $z_n < 0$. This contradicts our assumption about existing of an increasing sequence $\{z_i\}_{i=0}^\infty$.
\hfill
\halmos
\end{proof}

\section{Implementation of \texorpdfstring{\Cref{alg:primal dual matching}}{Algorithm~4}}
\label{apx:implementation}

In this section, we present two approaches to efficiently implement it with polynomial running time. (Recall that \Cref{alg:primal dual matching} is described as a continuous procedure.)

\paragraph{Water-filling method.} 
The first method utilizes the simple observation that for each online node $i$, \Cref{alg:primal dual matching} is essentially increasing ``\emph{water level}'' $\offlinedualj$ with identical $d\beta$ for every offline node $j\in \argmax_{j} \weightij - \offlinedualj$. Therefore, the continuous allocation procedure for online node $i$ can be divided into at most $m = |\offlinenodes|$ critical discrete time stamps where either (i) the set $\argmax_{j} \weightij - \offlinedualj$ increases, or (ii) the termination condition, i.e., the capacity of online node $i$ exhausts or $\offlinedualj \geq \weightij$ for all $j$, is satisfied. Between each pair of adjacent time stamps, we can binary search the increment $\Delta \offlinedual$ of water level $\offlinedualj$ for every offline node $j\in \argmax_{j} \weightij - \offlinedualj$ and then compute the increment of allocation $\Delta \alloc_j(\weightij)$ based on $\Delta \offlinedual$ accordingly. It is straightforward to verify that both $\Delta\offlinedual$ and $\Delta \alloc_j(\weightij)$ can be computed in polynomial time.

\newcommand{\penIntegrate}{\hat{\pen}}

\paragraph{Convex-programming method.} In this method, we determine the fractional allocation $\{\edgeallocij\}_{j\in\offlinenodes}$ for each online node $i$ by solving the following convex program:
\begin{align}
\label{eq:convex}
\tag{$\mathcal{P}_{\texttt{Convex}}(i)$}
\arraycolsep=1.4pt\def\arraystretch{1}
\begin{array}{llllllll}
\max\limits_{\boldsymbol{\edgealloc,\eta}\geq \mathbf{0}}  &\displaystyle
\displaystyle\sum_{j\in \offlinenodes}
\edgeallocij \weightij 
-
(1 + \buybackcost)\cdot 
\displaystyle\int_0^\infty \left(\alloc_j(\weight) - \eta_j(\weight)\right)\weight\,d\weight
\\
&\qquad\qquad\ 
-
\displaystyle\int_0^\infty \penIntegrate\left(
\int_\weight^\infty \eta_j(t)\,dt + 
\edgeallocij\cdot \indicator{\weightij \geq \weight}
\right)
\,d\weight
~~&\text{s.t.} \\[1.4em]
 &\displaystyle\sum_{j\in \offlinenodes}{\edgeallocij}\leq1 &
 \\[1.4em]
 &\displaystyle\int_{0}^\infty \eta_j(\weight)\,d\weight + \edgeallocij \leq 1 &j\in \offlinenodes~, \\ [1.4em]
 &\eta_j(\weight) \leq \alloc_j(\weight) &  j\in\offlinenodes~,~\weight\in[0,\infty)~.\\
\end{array}
\end{align}
where $\penIntegrate(x) \triangleq \int_0^x \pen(t)\,dt$,
variable $\edgeallocij$ specifies the fractional allocation between online node $i$ and offline node $j$, and $\alloc_j(\weight) - \eta_j(\weight)$ specifies the fractional buyback of offline node $j$ from weight $\weight$. Since allocation probability function $\alloc_j(\cdot)$ has at most $i$ strictly positive entries upon the arrival of online node $i$, it suffices to consider variable $\eta_j(\cdot)$ for those entries. Therefore, program~\ref{eq:convex} is a convex program with a polynomial number of variables and constraints. Consequently, it can be solved in polynomial time using classic approaches such as the Frank-Wolfe algorithm, projected gradient ascent, or mirror descent.

Finally, it suffices to argue that the optimal solution $\{\edgeallocij^*\}_{j\in\offlinenodes}$ coincides with the continuous allocation procedure described in \Cref{alg:primal dual matching}. This follows a similar argument as the one in \citet{FN-25} where the authors introduce the convex-programming-based online algorithm for the online edge-weighted bipartite matching with free disposal under batch arrival. At a high level, by analyzing the KKT condition of program~\ref{eq:convex}, the following two claims can be shown: (i) $\{\eta_j^*(\weight)\}_{j\in\offlinenodes,\weight\in[0, \infty)}$ follows greedy buyback; and (ii) if $\edgeallocij^* > 0$ then $\weightij - \offlinedualj^* \geq  \weight_{ij'} - \offlinedual_{j'}^*$ for every $j'\in \offlinenodes$ and $\weightij - \offlinedualj^* \geq 0$ where $\offlinedualj^*$ and $\offlinedual_{j'}^*$ are updated based on $\{\edgeallocij^*\}_{j\in\offlinenodes}$. The proof of the second claim relies on the first claim, and the second claim itself implies $\{\edgeallocij^*\}_{j\in\offlinenodes}$ coincides with the continuous allocation procedure described in \Cref{alg:primal dual matching} as desired. For the detailed proofs of these two claims, check Lemmas~2 and 3 in \citet{FN-25}.

\section{Randomized Rounding for Large Inventory}
\label{apx:matching rounding}

In this subsection, we focus on a variant model where each offline node $j$ has a large initial capacity (inventory) $\inventory_j\in \reals_+$. We present a near-optimal online rounding, i.e., we show any fractional online algorithm can be converted to a randomized integral online algorithm whose competitive ratio suffers an additional multiplicative factor $(1-(1+\buybackcost)\cdot O(\sqrt{\log (\mininventory)/\mininventory}))^{-1}$ where $\mininventory\triangleq \min_{j\in\offlinenodes}\inventory_j$ is the smallest initial capacity.
\begin{proposition}
\label{prop:rounding matching}
In the matching environment, any online fractional algorithm $\ALG$ with the competitive ratio $\approxratio$ can be converted into a randomized integral online algorithm $\ALG\primed$ with competitive ratio $\frac{\approxratio}{1-(1+\buybackcost)\cdot O(\sqrt{\log (\mininventory)/\mininventory})}$.
\end{proposition}
\begin{proof}{\emph{Proof.}}
    Let $\kappa \triangleq O\left(\sqrt{\log (\mininventory)/\mininventory}\right)$.
    Without loss of generality, we assume $\ALG$ always greedily buys back, and its profit from each offline node is non-negative.
    We construct a randomized integral online algorithm $\ALG\primed$ following the fractional allocation decision in $\ALG$ as follows: 
    For each online node $i\in\onlinenodes$, let $\{\edgealloc_{ij}\}_{j\in\offlinenodes}$ be the fractional allocation between online node $i$ and each offline node~$j$ in algorithm $\ALG$. 
    The randomized integral algorithm $\ALG\primed$ samples an offline node $j^* = j$ with probability $(1-\kappa)\cdot \edgealloc_{ij}$
    for each offline node $j$, and $j^* = \emptyset$ otherwise.
    If $j^*$ is not $\emptyset$, algorithm $\ALG\primed$ matches online node $i$ with offline node $j^*$, and greedily buys back (when it is necessary). Importantly, the sampled offline nodes for each online node in algorithm $\ALG\primed$ are independent. 

\newcommand{\edgeallocTilde}{\tilde\edgealloc}
\newcommand{\edgeallocTildeij}{\edgeallocTilde_{ij}}

    Now we analyze the competitive ratio of algorithm $\ALG\primed$. It suffices to show for each offline node $j$, the profit in algorithm $\ALG\primed$ is an $(1-(1+\buybackcost)\kappa)$-approximation of the profit in algorithm $\ALG$. Fix an arbitrary offline node $j$. Recall that $\edgealloc_{ij}$ is the fractional allocation between online node $i$ and each offline node~$j$ in algorithm $\ALG$. Let $\tilde\edgealloc_{ij}$ be the fractional buyback between online node $i$ and each offline node~$j$ in algorithm $\ALG$. As a sanity check, the profit from offline node $j$ in algorithm $\ALG$ is 
    \begin{align*}
       \sum_{i\in\onlinenodes} \weightij \edgeallocij - (1+\buybackcost)\cdot \sum_{i\in\onlinenodes} \weightij\edgeallocTildeij
    \end{align*}
    Let $\hat\onlinenodes_j\triangleq \{i\in\onlinenodes: \edgeallocij > \edgeallocTildeij\}$ and $\tilde\onlinenodes_j \triangleq \onlinenodes\backslash \hat\onlinenodes_j$.
    By definition, $\sum_{i\in\hat\onlinenodes_j}\edgeallocij - \edgeallocTildeij \leq \inventory_j$. 
    Moreover, since algorithm $\ALG$ greedily buys back, we have $\sum_{i\in\hat\onlinenodes_j}\edgeallocij \leq \inventory_j + 1$
    
    Let $Z_{ij}\primed$ be the event that algorithm $\ALG\primed$ matches online node $i$ with offline node $j$. By definition, $\prob{Z_{ij}} = (1-\kappa)\cdot \edgeallocij$.
    Similarly, let $\tilde Z_{ij}\primed$ be the event that algorithm $\ALG\primed$ buys back the offline node $j$ from the online node $i$. For each online node $i\in \tilde\onlinenodes_j$, it guarantees that 
    \begin{align*}
        \prob{\tilde Z_{ij}} \leq \prob{Z_{ij}} = 
        \left(1-\kappa\right)
        \cdot \edgeallocij 
        = 
        \left(1-\kappa\right)
        \cdot \edgeallocTildeij
        \end{align*}
    where the last equality holds since $i\in\tilde\onlinenodes_j$. For each online node $i\in\hat\onlinenodes_j$, 
    it guarantees that 
    \begin{align*}
        \prob{\tilde Z_{ij}} &\overset{(a)}{\leq}
\prob{Z_{ij}}\cdot\prob{\sum_{i'\in\hat\onlinenodes_j:i'\not=i}\indicator{Z_{ij}} \geq \inventory_j}
    \overset{(b)}{\leq}
    \left(1-\kappa\right)\cdot \edgeallocij
\cdot 
\kappa
=
\kappa\cdot \edgeallocij
    \end{align*}
    where inequality~(a) holds since both algorithm $\ALG$ and $\ALG\primed$ greedily buyback; and inequality~(b) holds due to the multiplicative form of the Chernoff bound \citep{che-52}. Putting all pieces together, the profit from offline node $j$ in algorithm $\ALG\primed$ is 
    \begin{align*}
        &\sum_{i\in\onlinenodes}\weightij\prob{Z_{ij}\primed}
        -
        (1+\buybackcost)\cdot\sum_{i\in\onlinenodes}\weightij\prob{\tilde Z_{ij}\primed}
        \\
        \geq&
         \sum_{i\in\tilde\onlinenodes_j}
         \weightij(1-\kappa)\edgeallocij
         -
         (1+\buybackcost)
         \sum_{i\in\tilde\onlinenodes_j}
         \weightij(1-\kappa)\edgeallocTildeij
         +
         \sum_{i\in\hat\onlinenodes_j}
         \weightij(1-\kappa)\edgeallocij
         -
         (1+\buybackcost)
         \sum_{i\in\hat\onlinenodes_j}
         \weightij\kappa\edgeallocij
         \\
         \geq{} &
         (1-(1+\buybackcost)\kappa)\cdot \left(\sum_{i\in\onlinenodes} \weightij \edgeallocij - (1+\buybackcost)\cdot \sum_{i\in\onlinenodes} \weightij\edgeallocTildeij
         \right)
    \end{align*}
    which is an $(1-(1+\buybackcost)\kappa)$-approximation to the profit in algorithm $\ALG$ as desired.
    \hfill\halmos
\end{proof}

\section{Numerical Experiments}
\label{apx:numerics}
In this section we run numerical simulations on synthetic instances to measure the empirical performance of our proposed algorithms and compare them with alternative methods. We utilize the concept of a performance gap, represented as $\frac{\ALG (\textit{I})}{\OPT (\textit{I})}$ to measure the effectiveness of an algorithm on each problem instance $\textit{I}$. This metric serves as a measure of how well the algorithm performs in comparison to the optimal solution for a given problem instance.

\smallskip
\noindent\textit{Experimental setup}. We consider an online booking setting in which customers who are willing to pay for the service arrive over time (online nodes) and should be matched to available resources (offline nodes). As for the offline nodes, there are $5$ resources with independent qualities randomly drawn from $U[0,1]$. On the online side, there are $50$ customers arriving over time with different willingness to pay. The willingness to pay $v_i$ of the customers is increasing over time and its increment is $2.5\%$ of the previous online node. In particular, $v_i = 1.025^{i-1}$.
We define the value on an edge to be the product of resource's quality and customer's willingness to pay. In order to add edge-wise uncertainty to the problem, we assume a customer $i$ accepts an offline node with quality $u_j$ with probability $\min(\frac{50}{i}\times u_j,1)$. Note that in such an instance, more valuable and less flexible customers arrive later in the sequence --- hence it is critical for the online algorithm to protect the inventory to avoid paying large amounts of cancellation/buyback cost. \\
\textit{Policies}. In the experiments, we consider 4 different algorithms.
\begin{enumerate}
    \item \textbf{Primal-Dual Fractional}:
    this policy is \Cref{alg:primal dual matching} defined in \Cref{sec:matching}. It achieves the optimal competitive ratio among all online policies.
    \item \textbf{Primal-Dual Integral}:
    this policy is \Cref{alg:opt deterministic matching} defined in \Cref{sec:deterministic}. It achieves the optimal competitive ratio among all deterministic integral policies.
    \item \textbf{Marginal Greedy algorithm}: this policy that assigns the online node $i$ to the offline node $j$ with the largest positive marginal profit. $$\argmax_j \weightij - (1+f)\weight_j$$
    \item \textbf{Free-disposal algorithm}: The fraction algorithm that does not take into account the buyback cost. This is, in fact, \Cref{alg:primal dual matching} for $f=0$, which was introduced in the earlier work of \cite{DHKMY-16}.
\end{enumerate}
\textit{Result}.
We generate \revcolorCMTwo{$200$} sample instances for each $\buybackcost \in \{0,0.04, \dots, 0.36\}$.
For each policy, we compute its performance gap as the ratio between its profit and the profit in the optimal offline benchmark. We also report the \revcolorCMTwo{central $90\%$ empirical interval} for these performance gaps for each $\buybackcost$. The result of this simulation can be seen in \Cref{fig:numerics}.\\
 \revcolorCMTwo{Both \Cref{alg:opt deterministic matching} and \Cref{alg:primal dual matching} outperform the benchmark algorithms for nearly all choices of $f$ considered in our experiments. Overall, the fractional algorithm (\Cref{alg:primal dual matching}) tends to achieve a higher mean performance and lower variability than \Cref{alg:opt deterministic matching}, although the latter attains a slightly higher mean in a few instances.}
\begin{figure}
    \centering\includegraphics[width=\textwidth]{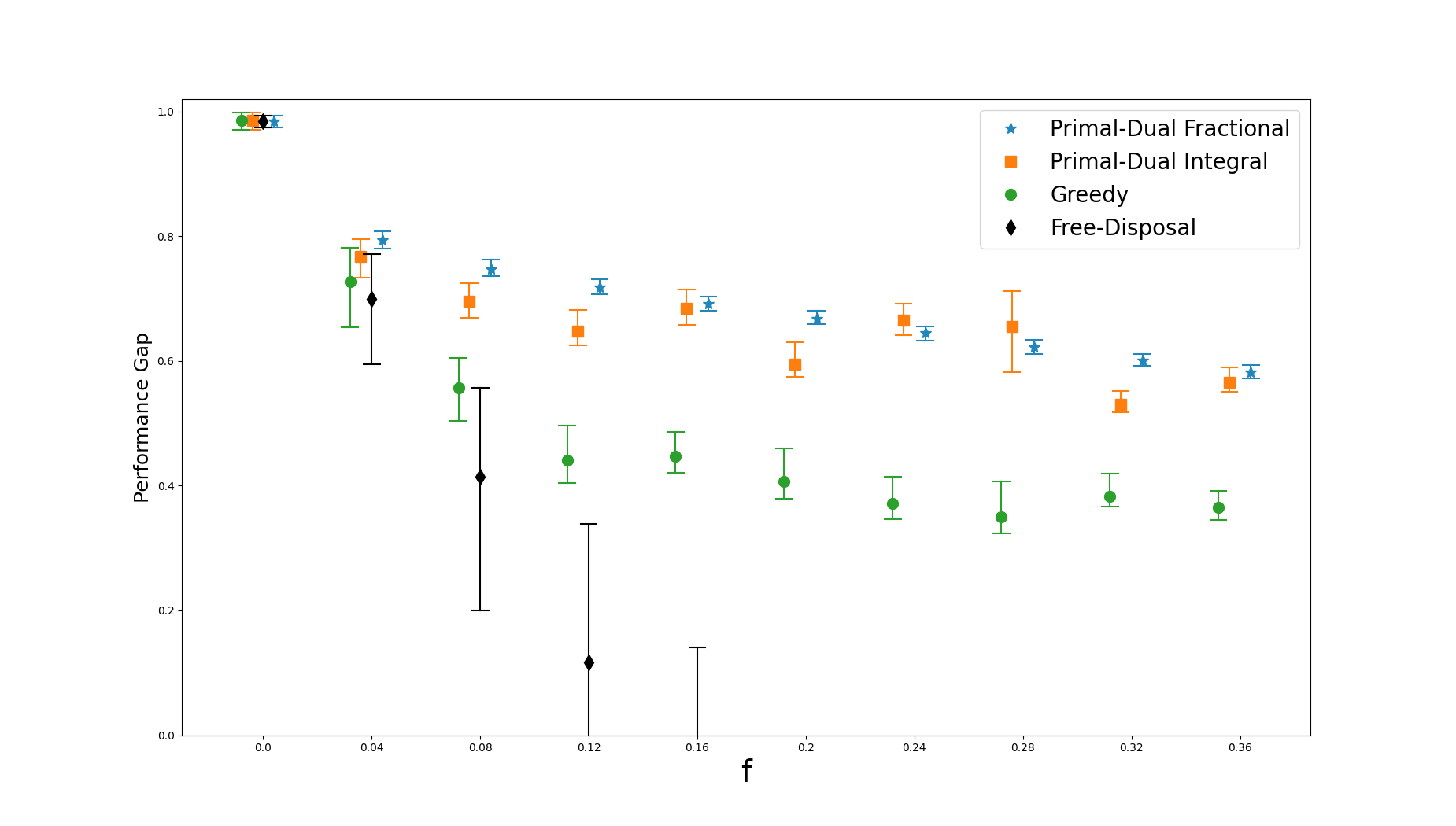}
    \caption{Numerical experiment results:
    The vertical lines represent the \revcolorCMTwo{central $90\%$ empirical interval}, and the shapes indicate the mean for each algorithm and each $\buybackcost$.}
    \label{fig:numerics}
\end{figure}

\revcolor{\section{Extensions}
\label{apx:extension}
In this section, we generalize our primal-dual analysis framework to multiple extension models with cancellation costs. 
While all subsections focus on the generalization of the fractional algorithm (\Cref{alg:primal dual matching}), the deterministic integral algorithm (\Cref{alg:opt deterministic matching}) can be extended using a similar analysis. \revcolorCMTwo{Crucially, these extensions demonstrate that our continuous primal-dual approach is not an ad-hoc machinery limited to simple one-to-one matching, but rather a robust framework capable of accommodating complex constraints.}

\subsection{Configuration allocation problem}
\label{sec:configuration}
In this section, we extend our model to allow richer forms of assignment, where an online node may be simultaneously matched to multiple offline nodes, potentially using different portions of each. This generalization captures a broader class of applications where resource sharing or combinatorial allocations are essential--such as in submodular welfare maximization (\Cref{apx:submodular welfare maximization}) and assortment optimization (\Cref{apx:assortment}). We present a unified framework that accommodates these more flexible configuration structures and demonstrate that the performance guarantee achieved in the original setting continues to hold under this generalized model. \revcolorCMTwo{In practice, this formulation naturally captures scenarios such as broadcast advertising, where an advertiser purchases a specific configuration of commercial slots across multiple nights, and preempting this bundle incurs a make-good buyback cost.}
Notably, the special case of the configuration allocation model with zero buyback factor $\buybackcost = 0$ has been studied in \citet{DHKMY-16,FN-25}.

\revcolorCAMtwo{Before presenting the formal model, we mention a concrete application in which buyback arises naturally in configuration allocation. Consider the banner advertising (cf.\ \Cref{app:practical}), where an advertiser may request a \emph{bundle} of banner slots across several webpages at once; accommodating a later, higher-paying advertiser may then require revoking (parts of) a previously accepted bundle at a contractual cancellation penalty proportional to its value. This is precisely a buyback model within an online configuration allocation model.
}

The primal-dual presentation of the configuration allocation model can be formulated as follows:
\begin{align}
\tag{$\mathcal{P}_{\textsc{Config}}$}
\label{eq:LP-max-weight-config}
\arraycolsep=1.4pt\def\arraystretch{1}
\begin{array}{llllllll}
\max  &\displaystyle\sum_{i\in\onlinenodes}
\displaystyle\sum_{C \subseteq \mathcal{C}_i} \displaystyle\sum_{j\in C}
\edgeallocijConfig \weightijConfig &~~\text{s.t.}&
& \quad\quad \text{min} &\displaystyle\sum_{i\in\onlinenodes }{\onlineduali}+\displaystyle\sum_{j\in \offlinenodes}{\offlinedualj}&~~\text{s.t.} \\[1.4em]
 & \edgeallocijConfig \leq \configallocij &  i\in\onlinenodes,j\in\offlinenodes,C \subseteq \mathcal{C}_i,& 
& &\theta_{ijC}+n_{jC}\offlinedualj\geq \weightijConfig&  i\in\onlinenodes,j\in \offlinenodes,C \subseteq \mathcal{C}_i~,\\[1.4em]
 &\displaystyle\sum_{C \subseteq \mathcal{C}_i}{\configallocij}\leq 1 &i\in \onlinenodes~, &
& &\onlineduali \geq \displaystyle\sum_{j\in \offlinenodes} \theta_{ijC} &i\in \onlinenodes,C \subseteq \mathcal{C}_i~, \\
 &\displaystyle\sum_{i\in\onlinenodes}
\displaystyle\sum_{C \subseteq \mathcal{C}_i} 
n_{j C} \cdot \edgeallocijConfig \leq 1 & j\in\offlinenodes~.
 \qquad \qquad &
& &\offlinedualj,\onlineduali,\theta_{ijC} \geq 0  &i\in\onlinenodes,j\in \offlinenodes,C \subseteq \mathcal{C}_i.
\end{array}
\end{align}
In the primal program above, the decision variable $z_{iC}$ represents the fractional allocation of online node $i$ under configuration $C$. The decision variable $x_{ijC}$ indicates the utilization level of offline node $j$ under configuration $C$ contributed by online node $i$. The parameter $n_{jC}$, a primitive model parameter, specifies the amount of capacity consumed by offline node $j$ per unit of utilization under configuration $C$. Therefore, the term $n_{jC} \cdot x_{ijC}$ captures the capacity consumption of offline node $j$ from online node $i$ under configuration~$C$. (In this setup, we assume that the platform can implement any strict sub-utilization, i.e., $x_{ijC} < z_{iC}$, without incurring additional costs.)

Our proposed dual-based algorithm works as follows. Upon the arrival of each online node $i$,
if there exists, a configuration with positive marginal, the algorithm picks the configuration $C$ maximizing the marginal as defined:
\begin{align}\label{eq:configALG}
    \sum_{j\in C} \left(w_{ijC} - n_{jC}\beta_j\right)^+.
\end{align}
Suppose an online node $i$ arrives and we choose to allocate a configuration $C$ of offline resources to it. To make this allocation, we may need to buy back some of the capacity of each offline node $j \in C$, which was previously assigned to another online node $i'_j$, via configuration $C'_j$. The resulting change in the primal objective can be expressed as:
\begin{align*}
    \Delta(\text{Primal}) = \sum_{\{j| j \in C, w_{ijC} > n_{jC}\beta_j \}} \left(w_{ijC} - (1+f)w_{i'_j j C'_j} \right) dx,
\end{align*}
Notice that for the $w_{ijC} \leq n_{jC}\beta_j$ means not using $n_{jC}$ capacity of online node $i$. In this case the algorithm may choose not to reallocate any portion of offline resource $j$ to the new arrival, in which case no buyback is necessary. The corresponding algorithm is presented in \Cref{alg:primal dual configuration}.
\begin{algorithm}
\oldcolor{
\caption{Primal-dual fractional online algorithm for 
configuration allocation with buyback}
\label{alg:primal dual configuration}
    \SetKwInOut{Input}{input}
    \SetKwInOut{Output}{output}
 \Input{Penalty function $\pen$
 }
 
 \vspace{2mm}
 
 Initialize $\offlinedual_j \gets 0$, and $\forall \weight \in \mathbb{R}_+:\allocj(\weight) \gets \diracdeltafunction_0(\weight), \callocj(\weight) \gets \indicator{\weight = 0}$ for every $j\in\offlinenodes$.

{\color{royalazure}
\tcc{$\diracdeltafunction_0(\cdot)$
 is the Dirac delta function centered at 0.}}
  
 Initialize $\theta_{ijC} \gets 0$, for every $i \in \onlinenodes$ and $j\in C \in \mathcal{C}_i$.
 
 
 \vspace{1mm}
 
 
 \vspace{1mm}
 

 
 \For{each online node $i\in\onlinenodes$}
 {
 
 \vspace{1mm}
 
 \While{capacity of online node $i$ is not exhausted and
 there exists $C\in \mathcal{C}_i$ s.t.\ $\sum_{j\in C} \left(w_{ijC} - n_{jC}\beta_j\right)^+ > 0$}
 {
 \vspace{2mm}
 
 Let $C^* \gets \argmax_{C\in\mathcal{C}_i} \sum_{j\in C} \left(w_{ijC} - n_{jC}\beta_j\right)^+$

  \For{each $j \in \{j \in C^* \mid \weight_{ijC^*} \geq n_{jC^*} \offlinedual_{j}\}$}
 {
 
\vspace{1mm}
 Buy back fraction $dx$ of offline node $j$ from 
 the smallest allocated weight $\buybackweight$,
 i.e., $\alloc_{j}(\weight) \gets
 \alloc_{j}(\weight) - dx\cdot\delta_{\buybackweight}(\weight)$
 
 \vspace{1mm}
 Allocate fraction $dx$ of offline node $j$
 to online node $i$,
 i.e., $\alloc_{j}(\weight) \gets \alloc_{j}(\weight) +  dx\cdot\delta_{\weight_{ijC^*}}(\weight)$
 
 \vspace{1mm}
 Update the allocation quantile function 
$\displaystyle\calloc_{j}(\weight) \gets \int_\weight^{\infty}
\alloc_{j}(\weight')\,d\weight'$

 \vspace{1mm}
Update the dual assignment $\displaystyle\offlinedual_{j}\gets 
\int_{0}^{\infty}
\pen_{j}(\calloc_{j}(\weight))\,d\weight$}
 }
 
 }}
\end{algorithm}

In the remainder of this section, we prove a generalized version of \cref{thm:competitive ratio exponential penalty function}.

\begin{theorem}
\label{thm:CRConfig}
For every $\expfuncparamone \geq e$
and 
$\expfuncparamtwo \geq \frac{1 + \buybackcost}{\expfuncparamone - 1}$,
\Cref{alg:primal dual configuration} with generalized exponential penalty function 
$\pen(\calloc) = \expfuncparamtwo(\expfuncparamone^\calloc - 1)$
has competitive ratio at most $\approxratioexp(\buybackcost)$,
where
\begin{align*}
    \approxratioexp(\buybackcost) \triangleq
    \max_{\weight \geq 1+\buybackcost}
    \frac
    {
    (\expfuncparamtwo + 1)\log(\expfuncparamone)\weight 
    -
    \expfuncparamtwo\expfuncparamone\log(\expfuncparamone)
    }
    {
    \weight - (1 + \buybackcost) 
    }
\end{align*}
\end{theorem}

\begin{proof}{\emph{Proof.}}
We construct a dual assignment based on the allocation decision
made in \Cref{alg:primal dual configuration} (denoted by $\ALG$ throughout the proof) as follows.
First, set $\onlineduali \gets 0$, $\theta_{ijC} \gets 0$, and $\offlinedualj \gets 0$
for all $i\in\onlinenodes,j\in\offlinenodes, C \subseteq \mathcal{C}_i$.
Now consider every allocation (and buyback) decision in 
$\ALG$.
Whenever $\ALG$
buys back $dx$ fraction of offline node $j$ from online node $i'$
and then 
re-allocates this fraction to online node $i$,
update the dual variables as follows:
\begin{align*}
    \offlinedualj \gets 
    \offlinedualj + 
    \left(\displaystyle\int_{w_{i'_jjC'_j}}^{\weightijConfig}
    \penderivative(\callocj(\weight))\,d\weight
    \right)dx
    \qquad \textrm{and} \qquad
    \onlineduali \gets \onlineduali + 
    \log(\expfuncparamone)\sum_{j\in C}\left(
    \weightijConfig - n_{jC}\offlinedualj
    \right)^+dx,
\end{align*}
By construction, the invariant $\offlinedualj \equiv 
\int_0^{\infty}\pen(\callocj(\weight))\,d\weight$
is retained throughout the execution of $\ALG$. Moreover, the value of $\theta_{ijC}$ will set to satisfy constraint $2$. More specifically at the end of allocating online node $i$:
\begin{align*}
    \theta_{ijC} \gets \left(w_{ijC}-n_{jC}\beta_j\right)^+.
\end{align*}

The rest of the proof is done in two steps:

\noindent
[\emph{Step i}] \emph{Checking the feasibility of dual.}
We first show that the constructed dual assignment 
is feasible.
By construction, 
whenever $\ALG$ allocates $dx$ amount 
of offline node $j$ to online node $i$, 
we know that $\weightijConfig - n_{jC}\offlinedualj \geq 0$.
Thus, $\onlineduali \geq 0$ and $\offlinedualj \geq 0$
for all $i\in\onlinenodes,j\in\offlinenodes$. Furthermore, feasibility of $\theta_{ijC}+n_{jC}\beta_j \geq w_{ijC}$ follows from the construction.
Next, we show the feasibility of dual constraint 
$\onlineduali \geq \displaystyle\sum_{j\in \offlinenodes} \theta_{ijC}$
by considering two cases.
\begin{itemize}
    \item\underline{Case I --- $\ALG$ does not exhaust the unit capacity of online node $i$}:
    By construction, after the departure of online node $i$,
    $n_{j'C'}\offlinedual_{j'} \geq \weight_{ij'C'}$ for every offline node and configuration $j'\in C \subseteq \mathcal{C}_i$.
    Thus, $\theta_{ijC} = 0$ and the dual constraint is satisfied.
    \item\underline{Case II --- $\ALG$ exhausts the unit capacity of online node $i$}:
    By construction, each $dx$ fraction is allocated 
    to a configuration $C'$ such that $\sum_{j \in C'} \left(\weight_{ijC'} - n_{jC}\offlinedual_{j}\right) \geq 
    \sum_{j \in C} \left(\weight_{ijC} - n_{jC}\offlinedual_{j}\right)$. 
    Thus, the dual variable $\onlineduali$ can be lower bounded for all $C$ as
    \begin{align*}
        \onlineduali \geq \displaystyle\int_{0}^{1}
        \log(\expfuncparamone) \sum_{j\in C_x} \left(w_{ijC_x} - \offlinedualj^{(i, x)}\right)^+
        dx
        \geq
        \log(\expfuncparamone) \sum_{j\in C} \left(\weightijConfig - \offlinedualj^{(i, 1)}\right)^+
        \geq 
        \sum_{j\in C} (\weightijConfig - \offlinedualj)^+= \displaystyle\sum_{j\in C} \theta_{ijC},
    \end{align*}
    where $\offlinedualj^{(i, x)}$ is the value of dual variable 
    $\offlinedualj$ after $x$ fraction of online node $i$ is matched 
    with offline nodes,
    and the last inequality holds since $\expfuncparamone \geq e$
    and $\offlinedualj$ is increasing throughout the execution of $\ALG$.
\end{itemize}

\smallskip
\noindent
[\emph{Step ii}] \emph{Comparing objective values in primal and dual.}
Here we show that the total profit of $\ALG$ 
is a $\approxratioexp(f)$-approximation 
of the objective value of 
the above dual assignment in the dual program.
To show this, we consider every allocation (and buyback)
decision in $\ALG$ and its impact on total profit,
as well as the objective value of the dual assignment.

Suppose $\ALG$ allocates configuration $C$ to the arriving online node $i$, and for each offline node $j \in \tilde{C} \coloneqq \{j \in C \mid \weight_{ijC} \geq n_{jC} \offlinedual_{j}\}$, it buys back a fraction $dx$ from a previously assigned online node $i'_j$ and its corresponding configuration $C'_j$.
The change in the profit, that is, the net change in the primal objective \emph{after} we incorporate the buyback cost, can be lower bounded in the following way:
\begin{align*}
    \Delta(\text{Primal}) \geq \sum_{j \in \tilde{C}} \left(w_{ijC} - (1+f)w_{i'_jjC'_j}\right)dx.
\end{align*}
Furthermore, the change in the dual objective is upper-bounded as follows:
\begin{align*}
    \Delta(\text{Dual}) &=
    \left(
    \log(\expfuncparamone)\sum_{j \in \tilde{C}}
    \left(
    \weightijConfig -
    \displaystyle\int_0^\infty \pen(\callocj(\weight))\,d\weight
    \right)
    +
    \int_{w_{i'_jjC'_j}}^{\weightijConfig} 
    \penderivative(\callocj(\weight))\,d\weight
    \right)dx
    \\
    &\overset{(a)}{\leq} 
    \left(
    \log(\expfuncparamone)\sum_{j \in \tilde{C}}
    \left(
    \weightijConfig -
    \displaystyle\int_0^{w_{i'_jjC'_j}} \pen(1)\,d\weight
    -
    \displaystyle\int_{w_{i'_jjC'_j}}^{\weightijConfig} \pen(\callocj(\weight))\,d\weight
    \right)
    +
    \int_{w_{i'_jjC'_j}}^{\weightijConfig} 
    \penderivative(\callocj(\weight))\,d\weight
    \right)dx
    \\
    &=
    \left(
    \log(\expfuncparamone)\sum_{j \in \tilde{C}}
    \left(
    \weightijConfig -
    \expfuncparamtwo\left(\expfuncparamone - 1\right)w_{i'_jjC'_j}
    -
    \displaystyle\int_{w_{i'_jjC'_j}}^{\weightijConfig} 
    \expfuncparamtwo\left(
    \expfuncparamone^{\callocj(\weight)} - 1
    \right)\,d\weight
    \right)
    +
    \int_{w_{i'_jjC'_j}}^{\weightijConfig} 
    \log(\expfuncparamone)
    \expfuncparamtwo
    \expfuncparamone^{\callocj(\weight)}
    \,d\weight
    \right)dx
    \\
    &=
    \log(\expfuncparamone)\sum_{j \in \tilde{C}} \left( (\expfuncparamtwo + 1)\weightijConfig
    -
    \expfuncparamtwo\expfuncparamone
    w_{i'_jjC'_j} \right)dx~~,
\end{align*}
where inequality~(a) holds 
by dropping $\int_{\weightij}^\infty \pen(\callocj(\weight))\,d\weight$ 
and the fact that $\callocj(\weight) = 1$ 
for every~$\weight\leq\weightipj$.
Note that 
\begin{align*}
    \weightijConfig
&\geq \int_0^\infty \pen(\callocj(\weight))\,d\weight 
\geq 
\int_0^{w_{i'_jjC'_j}} \pen(\callocj(\weight))\,d\weight 
=
\pen(1)w_{i'_jjC'_j}
\geq (1 + \buybackcost)w_{i'_jjC'_j}~~,
\end{align*}
where the last inequality holds since $\expfuncparamtwo \geq \frac{1+f}{\expfuncparamone - 1}$. Therefore, $\frac
    {
    (\expfuncparamtwo + 1)\weightijConfig
    -
    \expfuncparamtwo\expfuncparamone
    w_{i'_jjC'_j}
    } 
    {
    \weightijConfig
    -
    (1+f)w_{i'_jjC'_j}
    }\log(\expfuncparamone) \leq \approxratioexp(\buybackcost) $ and 

$$\Delta(\text{Dual})= \frac
    {
    \log(\expfuncparamone) \sum_{j \in \tilde{C}}\left(
    (\expfuncparamtwo + 1)\weightijConfig
    -
    \expfuncparamtwo\expfuncparamone
    w_{i'_jjC'_j}\right)
    } 
    {
    \sum_{j \in \tilde{C}}\left(
    \weightijConfig
    -
    (1+f)w_{i'_jjC'_j}\right)
    }
    \Delta(\text{Primal}) \leq \approxratioexp(\buybackcost) 
\cdot \Delta(\text{Primal})$$
By summing $\Delta(\text{Dual})$ and $\Delta(\text{Primal})$ over all allocations and buyback decisions throughout the horizon, we obtain:
$$
\textrm{total-profit}(\ALG)\triangleq  \text{Primal}\geq \frac{1}{ \approxratioexp(\buybackcost) }\cdot\text{Dual}
$$
Finally, by weak duality of the linear program, $\text{Dual}\geq \textrm{profit}(\OPT)$, which  finishes the proof.
\hfill\halmos
\end{proof}

\subsection{Submodular welfare maximization problem}
\label{apx:submodular welfare maximization}

In this section, we explore another generalization of our model in which each online node may be matched to a subset of offline nodes, and the value derived from such a match is derived by a submodular function. This formulation captures settings where the marginal contribution of each additional offline node depends on the context of the chosen subset and has diminishing returns---common in recommendation systems, resource allocation with substitution, and combinatorial auctions. We demonstrate that this model is a non-trivial special case of the configuration-based allocation framework introduced in \Cref{sec:configuration}, and thus inherits its structural and algorithmic insights.

One can think of $w_{ijC}$ for $j\in C$ to be $f_i([j] \cap C) - f_i([j-1] \cap C)$. Then the Primal-Dual will be very similar to program~\ref{eq:LP-max-weight-config}:
\begin{align}
\tag{$\mathcal{P}_{\texttt{Submodular}}$}
\label{eq:LP-max-weight-submodular}
\arraycolsep=1.4pt\def\arraystretch{1}
\begin{array}{llllllll}
\max  &\displaystyle\sum_{i\in\onlinenodes}
\displaystyle\sum_{C \subseteq \offlinenodes} \displaystyle\sum_{j\in C}
\edgeallocijConfig \weightijConfig &~~\text{s.t.}&
& \quad\quad\quad\quad \text{min} &\displaystyle\sum_{i\in\onlinenodes }{\onlineduali}+\displaystyle\sum_{j\in \offlinenodes}{\offlinedualj}&~~\text{s.t.} \\[1.4em]
 & \edgeallocijConfig = \configallocij &  i\in\onlinenodes,j\in\offlinenodes,C\subseteq \offlinenodes,& 
& &\theta_{ijC}+\offlinedualj\geq \weightijConfig&  i\in\onlinenodes,j\in \offlinenodes,C\subseteq \offlinenodes~,\\[1.4em]
 &\displaystyle\sum_{C \subseteq \offlinenodes}{\configallocij}\leq 1 &i\in \onlinenodes~, &
& &\onlineduali \geq \displaystyle\sum_{j\in \offlinenodes} \theta_{ijC} &i\in \onlinenodes,C\subseteq \offlinenodes~, \\
 &\displaystyle\sum_{i\in\onlinenodes}
\displaystyle\sum_{C \subseteq \offlinenodes} 
\edgeallocijConfig \leq 1 & j\in\offlinenodes~.
 \qquad \qquad &
& &\offlinedualj,\onlineduali \geq 0  &i\in\onlinenodes,j\in \offlinenodes.
\end{array}
\end{align}
This program differs from program~\ref{eq:LP-max-weight-config} in that it replaces the inequality constraint $x_{ijC} \leq z_{iC}$ with the equality $x_{ijC} = z_{iC}$. This adjustment is required due to the definition of the weights $w_{ijC}$, as a configuration $C$ can only be exploited if, for every $j \in C$, the condition $x_{ijC} = z_{iC}$ holds. In what follows, we show that relaxing this equality to an inequality does not affect the optimal value, implying that program~\ref{eq:LP-max-weight-submodular} is in fact a special case of program~\ref{eq:LP-max-weight-config} with $n_{jC} = 1$. 

We proceed by contradiction. Let \((x^*, z^*)\) be an optimal solution to the relaxed program, and suppose there exists some \(\hat{i}, \hat{j}, \hat{C}\) such that \(x^*_{\hat{i}\hat{j}\hat{C}} < z^*_{\hat{i}\hat{C}}\). We construct a new feasible solution \((\tilde{x}, \tilde{z})\) by modifying only a few coordinates as follows:
\begin{align*}
    \tilde{z}_{\hat{i}C} =
    \begin{cases}
        x^*_{\hat{i}\hat{j}C} & \text{if } C = \hat{C},\\
        z^*_{\hat{i}C} + (z^*_{\hat{i}\hat{C}} - x^*_{\hat{i}\hat{j}\hat{C}}) & \text{if } C = \hat{C} \setminus \{\hat{j}\}, \\
        z^*_{\hat{i}C} & \text{otherwise},
    \end{cases}
    &&
    \tilde{x}_{\hat{i}jC} =
    \begin{cases}
        x^*_{\hat{i}jC} - (z^*_{\hat{i}\hat{C}} - x^*_{\hat{i}\hat{j}C}) & \text{if } j \in \hat{C} \setminus \{\hat{j}\},\, C = \hat{C}, \\
        x^*_{\hat{i}jC} + (z^*_{\hat{i}\hat{C}} - x^*_{\hat{i}\hat{j}C}) & \text{if } j \in \hat{C} \setminus \{\hat{j}\},\, C = \hat{C} \setminus \{\hat{j}\}, \\
        x^*_{\hat{i}jC} & \text{otherwise}.
    \end{cases}
\end{align*}
It is straightforward to verify that \((\tilde{x}, \tilde{z})\) remains feasible. Furthermore, the objective value strictly increases by
\begin{align*}
    (z^*_{\hat{i}\hat{C}} - x^*_{\hat{i}\hat{j}\hat{C}}) \sum_{j \in \hat{C} \setminus \{\hat{j}\}} \left(w_{\hat{i}j\hat{C} \setminus \{\hat{j}\}} - w_{\hat{i}j\hat{C}}\right),
\end{align*}
which, using the definition of \(w\), simplifies to
\begin{align*}
    (z^*_{\hat{i}\hat{C}} - x^*_{\hat{i}\hat{j}\hat{C}}) \sum_{j \in \hat{C} \setminus \{\hat{j}\}} \bigg( 
        & \left(f_{\hat{i}}([j] \cap (\hat{C} \setminus \{\hat{j}\})) - f_{\hat{i}}([j-1] \cap (\hat{C} \setminus \{\hat{j}\}))  \right)\\
        -\; & \left(f_{\hat{i}}([j] \cap \hat{C}) - f_{\hat{i}}([j-1] \cap \hat{C}) \right) 
    \bigg),
\end{align*}
which is non-negative by the submodularity of \(f_{\hat{i}}\). This contradicts the optimality of \((x^*, z^*)\). Therefore, the relaxed version of the problem achieves the same objective value as the original formulation. This confirms that the submodular welfare maximization problem is indeed a special case of the configuration allocation problem.

\subsection{Assortment planning problem}
\label{apx:assortment}

In this section, we explore another generalization of our model to the online assortment planning problem. The special case with buyback factor $\buybackcost = 0$ has been studied in many prior work \citep[e.g.,][]{GNR-14,FNS-19,GGISUW-22}. \revcolorCMTwo{A natural application of this setting arises in capacity-constrained travel and hospitality platforms: the system dynamically displays assortments of flight or room tiers, and may selectively bump or upgrade earlier customers---incurring a compensatory buyback cost---to accommodate later high-yield arrivals.} We first define the model and then discuss the extension of our proposed algorithms and their competitive ratio guarantees.

In the online assortment planning problem,
each offline node $j$ in $\offlinenodes$ corresponds to a product with capacity~$\inventory_j$.
Each node $i$ in $\onlinenodes$ corresponds a consumer with a choice model $\choice_i:2^{\offlinenodes}\times \offlinenodes \rightarrow [0, 1]$.
By displaying an \emph{assortment} $\assortment\subseteq\offlinenodes$ (i.e., a subset of products) to her,
consumer $i$ selects each product $j\in\assortment$ with probability $\choice_i(\assortment, j)$.
If product $j$ has available units, the platform decides whether to \revcolorCMTwo{approve} the transaction of product $j$ to consumer $i$.
Approving this transaction, the platform gains reward $\weight_{ij}$ and the inventory of product $j$ is decreased by one.
The goal of the platform is to maximize the expected total rewards.\footnote{\oldcolor{For simplicity, we focus on a model where the platform can display assortments that include non-available products and has the ability to revoke consumers' selections. However, our results can be extended to an alternative model, where the platform is restricted to displaying assortments consisting only of available products and lacks the power to revoke selections. This extension can be achieved using the sub-assortment sampling technique developed in \citet{FNS-22}.}}

We introduce the \emph{fluid relaxation} of the assortment problem as follows:
for each consumer $i$, the platform \revcolorCMTwo{decides} a fractional assortment assignment $\{z_{i\assortment}\}_{\assortment}$
such that $\sum_{\assortment} z_{i\assortment} \leq 1$ and fractional product allocation $\{x_{i\assortment j}\}_{\assortment,j}$
such that $x_{i\assortment j} \leq  z_{i\assortment}$.
Given fractional decision $\{z_{i\assortment}, x_{i\assortment j}\}_{\assortment,j}$,
the platform collect rewards $\sum_{\assortment}\sum_{j}\weight_{ij} x_{i\assortment j}$ from consumer $i$,
and the inventory of each product $j$ is decreased by $\sum_{\assortment}\choice_i(\assortment, j)\cdot x_{i\assortment j}$.

First, we argue that the fluid relaxation of the assortment problem is a special case of the configuration allocation problem, and therefore, our proposed algorithms and their competitive ratio guarantees (\Cref{thm:CRConfig}) apply directly. To illustrate this, we describe a reduction where each instance of the assortment problem corresponds to an instance of the configuration allocation problem as follows: The feasible configuration set in the configuration allocation problem encodes all possible assortments in the assortment problem, i.e., $\Config \gets 2^{\offlinenodes}$. For each offline node $j$, we set $n_{jC} \gets \choice_i(\config, j)$ for all configurations (assortment) $\config \in \Config$. All other model parameters can be similarly defined. 


Finally, we explain how to extend the algorithm for the {fluid relaxation} to the original assortment planning problem in the large inventory regime. The overall approach is analogous to the {randomized rounding} technique used in the base model for large inventories (\Cref{apx:matching rounding}) plus the sub-assortment sampling procedure developed in \citet{FNS-22}.

Let $\ALG$ denote the duality-based algorithm for the fluid relaxation. 
We construct a randomized algorithm $\ALG\primed$ as follows: for each online node $i$ and each assortment $\assortment$, suppose that $\ALG$ displays assortment $\assortment$ to node $i$ with probability $z_{i\assortment}$ and selects product allocation $x_{i\assortment j}$ for each node $j$.
In the constructed algorithm $\ALG\primed$, with probability $z_{i\assortment}$, we invoke the {sub-assortment sampling procedure} \citep{FNS-22} with the targeted product allocation $\{x_{i\assortment j}\}$. \citet{FNS-22} shows that by using the sub-assortment sampling procedure, a sub-assortment $\assortment\primed \subseteq \assortment$ is displayed at random, and the probability that node $i$ selects each product $j$ is exactly $x_{i\assortment j}$.

Using the same reasoning as in \Cref{prop:rounding matching}, the competitive ratio of $\ALG\primed$ is that of $\ALG$ with an additional multiplicative factor of $(1-(1+\buybackcost)\cdot O(\sqrt{\log (\mininventory)/\mininventory}))^{-1}$
where $\mininventory = \min_{j \in \offlinenodes} \inventory_j$ is the smallest initial capacity.

\subsection{Resource allocation with secondary supply channels}
\label{apx:negative buyback cost}

In this section, we extend our model to allow for \emph{negative} values of buyback factor $\buybackcost$. Specifically, we consider a setting with $n$ resources, each having a nominal capacity of $1$. However, this capacity is not a hard constraint--exceeding it is permitted at a penalty: the cost of fulfilling an order beyond capacity becomes $(1 + \buybackcost)$ times its original value. This leads to the following extreme cases:
\begin{itemize}
    \item When $\buybackcost = -1$, overflow incurs no additional cost, effectively removing the capacity constraint. The problem then reduces to an uncapacitated online matching problem, where a greedy algorithm achieves a competitive ratio of $1$.
    \item When $\buybackcost = 0$, exceeding capacity yields zero value, mimicking the classical online matching with free cancellation (free disposal) model.
\end{itemize}
For values of buyback factor $\buybackcost$ between these two extremes, the system effectively operates under a two-tier capacity structure: a primary tier offering full value (up to capacity one per resource), and a secondary tier that accommodates overflow at a reduced value. This setting admits two interpretations:
(i) a third party with unlimited capacity absorbs overflow orders, but only at a discounted reward of $-\buybackcost \cdot \weight$; (ii) the system accesses a secondary, unreliable capacity--analogous to a spot market, where additional capacity is available with probability $-\buybackcost$.

It is important to note that the timing of cancellations (whether executed immediately upon arrival or deferred until the end) does not affect the final outcome of the problem. As a result, the algorithm focuses on retaining the highest-value allocations within the primary capacity and assigns lower-value allocations to the overflow capacity.

Moreover, the concept of buyback takes on a different meaning in this extended setting. Unlike in the main model, where buybacks represented a cost to the decision maker, here they function as an additional source of revenue. Accordingly, the \revcolorCMTwo{offline optimum} benchmark is also permitted to exploit this opportunity.

To capture this in our formulation, we introduce a new variable $t_{ij}$ in the linear program, where $t_{ij}$ denotes the amount of demand from online node $i$ allocated to the substitute (overflow) channel associated with offline node $j$. With this modification, the primal and dual programs are extended as follows:
\begin{align*}
\tag{$\tilde{\mathcal{P}}_{\texttt{OPT}}$}
\label{eq:LP-max-weight-Negative-f}
\arraycolsep=1.4pt\def\arraystretch{1}
\begin{array}{llllllll}
\max  &\displaystyle\sum_{i\in\onlinenodes}
\displaystyle\sum_{j\in \offlinenodes}\left(
\edgeallocij -\buybackcost t_{ij}\right) \weightij&~~\text{s.t.}&
& \quad\quad\quad\quad \text{min} &\displaystyle\sum_{i\in\onlinenodes }{\onlineduali}+\displaystyle\sum_{j\in \offlinenodes}{\offlinedualj}&~~\text{s.t.} \\[1.4em]
 &\displaystyle\sum_{j\in \offlinenodes}{\left(\edgeallocij+t_{ij}\right)}\leq1 &  i\in\onlinenodes~,& 
& &\onlineduali+\offlinedualj\geq \weightij& i\in\onlinenodes,j\in \offlinenodes~,\\[1.4em]
 &\displaystyle\sum_{i\in \onlinenodes}{\edgeallocij}\leq 1 &j\in \offlinenodes~, &
& &\onlineduali \geq -\buybackcost\weightij& i\in\onlinenodes,j\in \offlinenodes~, \\
 &\edgeallocij,t_{ij} \geq 0 &i\in\onlinenodes,j\in\offlinenodes~.
 \qquad \qquad &
& &\offlinedualj \geq 0 , \onlineduali \geq 0  &i\in\onlinenodes,j\in\offlinenodes~. 
\end{array}
\end{align*}
Notice that when $\buybackcost \geq 0$, the optimal solution satisfies $t_{ij} = 0$ for all $(i,j)$, and program~\ref{eq:LP-max-weight-Negative-f} reduces to the original formulation in program~\ref{eq:LP-max-weight}. Using this pair of primal and linear programs described---which in some sense is a generalization of our offline LP for the case with non-negative $f$---we extend our main results of \Cref{sec:matching} and \Cref{sec:deterministic} to the case with negative values of the parameter $f$. 

First, we extend both \Cref{alg:primal dual matching} and \Cref{thm:competitive ratio exponential penalty function} (in particular \Cref{coro:optimal competitive ratio small f}) to accommodate the case where $\buybackcost < 0$. We formalize this extended result in \Cref{thm:competitive ratio exponential penalty function negative F}, and describe the details of the modified version of our previous primal-dual fractional algorithm for the matching environment (under non-negative $f$) that can now achieve this result in \Cref{alg:primal dual matching negative F}.

\begin{algorithm}
\oldcolor{
\caption{Primal-dual fractional online algorithm for 
 matching with $f \leq 0$}
\label{alg:primal dual matching negative F}
    \SetKwInOut{Input}{input}
    \SetKwInOut{Output}{output}
 \Input{Penalty function $\pen$
 }
 
 \vspace{2mm}
 
 Initialize $\offlinedual_j \gets 0$, and $\forall \weight \in \mathbb{R}_+:\allocj(\weight) \gets \diracdeltafunction_0(\weight), \callocj(\weight) \gets \indicator{\weight = 0}$ for every $j\in\offlinenodes$.
 
 
 
 \vspace{1mm}
 
{\color{royalazure} \tcc{$\diracdeltafunction_w'(\cdot)$
 is the Dirac delta function centered at $w'$.}}
 
 \vspace{1mm}
 

 
 \For{each online node $i\in\onlinenodes$}
 {
 
 \vspace{1mm}
 
 \While{capacity of online node $i$ is not exhausted}
 {
 \vspace{2mm}
 
 Let $j^* \gets \argmax_{j\in\offlinenodes}\ 
 \weightij - \offlinedualj$ 

 \vspace{1mm}

 Let $\tilde{j} \gets \argmax_{j\in\offlinenodes}\ 
 \weightij$ 
 
 \vspace{1mm}

 \If{$-\frac{\buybackcost e}{e-(1+f)}w_{i\tilde{j}} \geq w_{ij^*} - \beta_{j^*}$}{
    \vspace{2mm}
 Allocate the remaining capacity to the \revcolorCMTwo{secondary tier} of offline node $\tilde{j}$.
 }

 \Else{

 \vspace{1mm}
 
 Buyback $dx$ fraction of offline node $j^*$ from 
 the smallest allocated weight $\buybackweight_{j^*}$,
 i.e., $\alloc_{j^*}(\weight) \gets
 \alloc_{j^*}(\weight) - dx\cdot\delta_{\buybackweight_{j^*}}(\weight)$
 {\color{royalazure}\tcc{Formally, $\buybackweight_{j^*} = \inf\{\weight'\in\reals_+:
 \calloc_{j^*}(\weight') < 1\}$}}
 
 \vspace{1mm}
 Allocate $dx$ fraction of offline node $j^*$
 to online node $i$,
 i.e., $\alloc_{j^*}(\weight) \gets \alloc_{j^*}(\weight) + dx\cdot\delta_{\weight_{ij^*}}(\weight)$
 
 \vspace{1mm}
 Update allocation quantile function 
$\forall \weight\in[\buybackweight_{j^*},\weight_{ij^*}]:~\calloc_{j^*}(\weight) \gets \int_\weight^{\infty}
\alloc_{j^*}(t)\,dt$

 \vspace{1mm}
Update the dual assignment $\offlinedual_{j^*}\gets 
\int_{0}^{\infty}
\pen_{j^*}(\calloc_{j^*}(\weight))\,d\weight$}
}
}}
\end{algorithm}
\begin{theorem}[{Optimal competitive ratio of fractional algorithms for $\boldsymbol{f\in[-1,0]}$}]
\label{thm:competitive ratio exponential penalty function negative F}
For every $f \leq 0$ and by setting $\expfuncparamone = e$
and 
$\expfuncparamtwo = \frac{1 + \buybackcost}{\expfuncparamone - (1+f)}$,
\Cref{alg:primal dual matching negative F} with generalized exponential penalty function 
$\pen(\calloc) = (1+f)\frac{e^\calloc - 1}{e-(1+f)}$
has competitive ratio at most $\frac{e}{e-(1+\buybackcost)}$. 
\end{theorem}
\begin{proof}{\emph{Proof.}}
Consider the linear program \ref{eq:LP-max-weight-Negative-f}. We construct a dual assignment as follows: for each online node $i$, if $j^*$ was selected (i.e., ``else case'' in line 8 of \Cref{alg:primal dual matching negative F}), we update the dual variables similar to the update in the proof of \Cref{thm:competitive ratio exponential penalty function}:
\begin{align*}
    \onlineduali \gets \onlineduali + 
    \left(
    \weight_{ij^*} - \offlinedual_{j^*}
    \right)\, dx
    \qquad\textrm{and} \qquad
    \offlinedual_{j^*} \gets 
    \offlinedual_{j^*} + 
    \left(\displaystyle\int_{\weight_{i'j^*}}^{\weight_{ij^*}}
    \penderivative(\calloc_{j^*}(\weight))\,d\weight
    \right)dx~~,
\end{align*}
where $i'<i$ is the previous online node whose value $\weight_{i'j^*} = \buybackweight$ has been bought back in the algorithm.
Otherwise (if $\tilde{j}$ was selected, i.e., ``if case'' in line 6 of \Cref{alg:primal dual matching negative F}), we update the dual variables in the following way: 
\begin{align*}
    \onlineduali \gets \onlineduali -\frac{\buybackcost e}{e-(1+f)}\weight_{i\tilde{j}}\,dx.
\end{align*}
The rest of the proof is done in two similar steps:

\noindent
[\emph{Step i}] \emph{Checking the feasibility of the new dual constraint.} By construction, whenever the algorithm allocates $dx$ units of the main channel of an offline node $j$ to an online node $i$, we have $w_{ij} - \beta_j \ge 0$. Hence, $\onlineduali$ is nondecreasing and positive for all $i\in\onlinenodes$. Also, $\offlinedualj$ is updated only if $w_{ij^*} - \beta_{j^*} \ge -\frac{\buybackcost e}{e-(1+f)}w_{ij^*}$ which is equivalent to: 
\begin{align*}
\frac{(1+f)(e-1)}{e-(1+f)}w_{ij^*} \geq \beta_{j^*} \geq \pen(1)w_{i'j^*} = (1+f)\frac{e - 1}{e-(1+f)}w_{i'j^*}.    
\end{align*}
Therefore, $w_{ij^*} \geq w_{i'j^*}$ and $\offlinedualj$ is nondecreasing during the execution of the algorithm, and in particular $\offlinedualj \geq 0$ for all $j\in\offlinenodes$.

\noindent Finally, with each infinitesimal fraction $dx$ allocated to an offline node, we have 
$$d \onlineduali \geq \max\left\{\weightij-\offlinedualj,-\frac{\buybackcost e}{e-(1+f)}w_{ij}\right\}dx \geq \max\left\{\weightij-\offlinedualj,-\buybackcost w_{ij}\right\}dx.$$ 
Thus, the dual variable $\onlineduali$ can be lower bounded as follows:
   \begin{align*}
        \onlineduali \geq \displaystyle\int_{0}^{1}
        \max\left\{\weightij - \offlinedualj^{(i, x)},-fw_{ij}\right\}dx
        \geq
        \max\left\{\weightij - \offlinedualj^{(i, 1)},-fw_{ij}\right\}
        \geq 
        \max\left\{\weightij-\offlinedualj,-fw_{ij}\right\},
    \end{align*}
    where $\offlinedualj^{(i,x)}$ is the value of dual variable $\offlinedualj$ after a fraction $x$ of online node $i$ is matched with offline nodes, and the last inequality follows since $\offlinedualj$ is non-decreasing throughout $\ALG$'s execution.

\smallskip
\noindent
[\emph{Step ii}] \emph{Comparing objective values in primal and dual.} We will show the approximation holds true after each infinitesimal allocation of online node $i$.
In the case of choosing $j^*$ (i.e., ``else case'' in line 8 of \Cref{alg:primal dual matching negative F}), 
following the same argument in \Cref{thm:competitive ratio exponential penalty function}, it can be shown that 
\begin{align*}
    \frac{\Delta(\text{Dual})}{\Delta(\text{Primal})}= \frac{
    (\expfuncparamtwo + 1)\log(\expfuncparamone)\weightij
    -\expfuncparamtwo\expfuncparamone\log(\expfuncparamone)
    \weightipj
    }{
    {
    \weightij
    -
    (1+\buybackcost)\weightipj
    }} = \frac{
    (\expfuncparamtwo + 1)\weightij
    -\expfuncparamtwo\expfuncparamone
    \weightipj
    }{
    {
    \weightij
    -
    (1+\buybackcost)\weightipj
    }} = \frac{e}{e-(1+f)}.
\end{align*}
The new case is when $\tilde{j}$ is chosen (i.e., ``if case'' in line 6 of \Cref{alg:primal dual matching negative F}) in which the changes in the primal and dual can be easily written as:
\begin{align*}
    \frac{\Delta(\text{Dual})}{\Delta(\text{Primal})} = \frac{
    -\frac{fe}{e-(1+f)}\weightij}{-fw_{ij}}= \frac{e}{e-(1+f)}.
\end{align*}
Hence, by summing $\Delta(\text{Dual})$ and $\Delta(\text{Primal})$ over all allocations and buyback decisions throughout the horizon, we obtain:
$$
\textrm{total-profit}(\text{\Cref{alg:primal dual matching negative F}})\triangleq  \text{Primal}\geq \frac{1}{ \approxratioexp(\buybackcost) }\cdot\text{Dual}
$$
still holds and the results follows from weak-duality.
\hfill\halmos
\end{proof}
\begin{remark}[\revcolorCMTwo{Tightness of the competitive ratio for $f<0$}]
It is not hard to see that the competitive ratio upper-bound  $\frac{e}{e-(1+f)}$ in \Cref{thm:competitive ratio exponential penalty function negative F} for fractional online algorithms under $f\in[-1,0]$ is tight---which means that no online (fractional) algorithm can obtain a better competitive ratio. The hard instance to show this result is the classic ``half graph'' (i.e., special case of \Cref{example:lower bound matching} when $f=0$). As an informal argument, consider an instance with $n$ online nodes and $n$ offline nodes (indexed by $1,2,\ldots,n$), each with capacity $1$.  Suppose that each online node $i\in[n]$ has an edge (with weight $1$) to all offline nodes $j\in\{i,\dots,n\}$. Since the graph admits a perfect matching, the offline optimum does not use any secondary channels and obtains $n$ in the objective. At the same time, the best fractional online algorithm, due to the symmetry of the instance and the fact that $-f\leq 1$, would be a greedy algorithm that (i)~uniformly allocates the arriving online node $i$ to available offline nodes until they reach their capacity of $1$, and then (ii)~will use the second channel to get an additional reward at the rate of $-f$ for the remaining online nodes. By similar calculations to those in \Cref{example:lower bound matching} (or by using the analysis of the lower bound for the classic online unweighted matching problem in \citealp{KVV-90}), we conclude that this algorithm will fully allocate the first $\lfloor\frac{n(e-1)}{e}\rfloor$ online nodes, and obtains a reward of at most $-f\cdot\frac{n}{e}+o(n)$ by using the secondary channel from the matching of the remaining online nodes. Therefore, the competitive ratio of this algorithm is at least $\frac{n}{n(\frac{e-1}{e})-f\frac{n}{e}}+\Omega(\frac{1}{n})$, which converges to $\frac{e}{e-(1+f)}$ as $n$ goes to $+\infty$. The formal argument to cover randomized algorithms is based on using Yao's lemma and picking a distribution over graph isomorphisms of the half-graph, where the nodes are indexed by a uniform random permutation $\Pi$. See details of this formal argument for the special case of $f=0$ in \Cref{example:lower bound matching}, \Cref{sec:lower-bound}, and a more formal analysis in \Cref{lem:lower bound matching optimum online}). We omit these details here for brevity. 
\end{remark}

Next, we shift our focus to the case of deterministic integral algorithms. In particular, we extend both \Cref{alg:opt deterministic matching} and \Cref{thm:competitive ratio deterministic integral} (in particular, \Cref{coro:optimal deterministic competitive ratio small f}) to accommodate the case where $\buybackcost < 0$. We formalize this extended result \Cref{thm:competitive ratio determinsitic integral negative F}, and explain the modified version of our previous deterministic integral primal-dual algorithm for the matching  environment (under non-negative $f$) that can achieve this result in  \Cref{alg:primal dual matching determinstic negative F}.
\begin{algorithm}
\oldcolor{
\caption{Primal-dual deterministic integral online algorithm for matching with $f\leq 0$}
\label{alg:primal dual matching determinstic negative F}
    \SetKwInOut{Input}{input}
    \SetKwInOut{Output}{output}
 \Input{penalty scalar $\penscalar$
 }
 
 \vspace{2mm}
 
 Initialize $\buybackweightj \gets 0$
 for every offline node $j\in\offlinenodes$.
 
 \vspace{1mm}
 
 \For{each online node $i\in\onlinenodes$}
 {
 
 \vspace{1mm}

  Let $j^* \gets \argmax_{j\in\offlinenodes}\ 
 \weightij - \penscalar\cdot \buybackweightj$

  Let $\tilde{j} \gets \argmax_{j\in\offlinenodes}\ 
 \weightij$
    
\If{$(1-\tau)w_{i\tilde{j}} \geq w_{ij^*} -  \penscalar\cdot \buybackweight_{j^*}$}{
    \vspace{2mm}
 Allocate online node $i$ to the \revcolorCMTwo{secondary tier} of offline node $\tilde{j}$.
 }

 \Else{

 \vspace{1mm}
 
 Buy back offline node $j^*$ from 
 the previously allocated online node $i'$ with $w_{i'j^*}\equiv\buybackweight_{j^*}$
 
 \vspace{1mm}
 Allocate offline node $j^*$ to online node $i$ 
 
  \vspace{1mm}
 Update $\buybackweight_{j^*}\leftarrow \weight_{ij^*}$
}
}}
\end{algorithm}

\begin{theorem}[{Optimal competitive ratio of deterministic integral algorithms for $\boldsymbol{f\in[-1,0]}$}]
\label{thm:competitive ratio determinsitic integral negative F}
For every buyback factor $\buybackcost \leq 0$,
\Cref{alg:primal dual matching determinstic negative F}
with $\penscalar = \frac{1+f}{1-f}$
has competitive ratio at most $\frac{2}{1 - \buybackcost}$.
\end{theorem}
\begin{proof}{\emph{Proof.}}
Similar to above, we construct a dual assignment for the linear program \ref{eq:LP-max-weight-Negative-f} as follows: for each online node $i$, if $j^*$ was selected (i.e., ``else case'' in line 7 of \Cref{alg:primal dual matching determinstic negative F}) we update similar to the update in the proof of \Cref{thm:competitive ratio deterministic integral}:
\begin{align*}
    \onlineduali \gets w_{ij^*} - \penscalar\cdot \buybackweight_{j^*}
    \qquad\textrm{and} \qquad
    \offlinedual_{j^*} \gets 
    \offlinedual_{j^*} + 
    \penscalar \left(w_{ij^*} - \buybackweight_{j^*}\right)
\end{align*}
(It is easy to see before the allocation $\offlinedual_{j^*} = 
    \penscalar \buybackweight_{j^*}$ and after the allocation $\offlinedual_{j^*} = 
    \penscalar w_{ij^*}$.) On the other hand if $\tilde{j}$ was selected (i.e., ``if case'' in line 5 of \Cref{alg:primal dual matching determinstic negative F}), we update the dual variables as 
\begin{align*}
    \onlineduali \gets(1-\tau) \weight_{i\tilde{j}}
\end{align*}
The rest of the proof is done in two similar steps:

\noindent[\emph{Step i}] \emph{Checking the feasibility of the new dual constraint.} By construction, whenever the algorithm allocates the main channel of an offline node $j$ to an online node $i$, we have $w_{ij} - \beta_j \ge 0$. Also as $1-\tau \geq 0$, then $\onlineduali$ is nondecreasing and positive for all $i\in\onlinenodes$. \\
For $\offlinedualj$, observe that upon update we have $w_{ij^*} - \penscalar\cdot \buybackweight_{j^*} \ge (1-\penscalar)w_{ij^*}$ which is equivalent to $w_{ij^*} \geq w_{i'j^*}$.  
Therefore, $\offlinedualj$ is nondecreasing during the execution of the algorithm, and in particular $\offlinedualj \geq 0$ for all $j\in\offlinenodes$.\\
Finally, after the allocation of online node $i$, for all $j$, we have $$\onlineduali \geq \max\left\{\weightij-\offlinedualj,(1-\tau)w_{ij}\right\} \geq \max\left\{\weightij-\offlinedualj,-fw_{ij}\right\}.$$
The last inequality comes from the fact that $\tau = \frac{1+f}{1-f} \leq 1+f$.

\smallskip
\noindent
[\emph{Step ii}] \emph{Comparing objective values in primal and dual.}
Here we show that the total profit of $\ALG$ 
is a $\frac{2}{1-f}$-approximation 
of the objective value of 
the above dual assignment.\\
First suppose allocation is made to the main channel (``else case'' in line 7 of \Cref{alg:primal dual matching determinstic negative F}), therefore, $\ALG$ buys back offline node $j$ 
from online node $i'$ (with weight $\weightipj\equiv\buybackweightj$)
and then re-allocates it to online node $i$.
The change in the profit (i.e., the net change in the primal objective \emph{after} we incorporate buyback cost) is $\Delta(\text{Primal}) = 
    \weightij - (1+\buybackcost)\buybackweightj,$
and the change in the dual objective is $
    \Delta(\text{Dual}) =
    \weightij - \penscalar\cdot \buybackweightj
    +
    \penscalar(\weightij - \buybackweightj)=
    (\penscalar + 1)\weightij
    -
    2\penscalar\buybackweightj$. Combining with the fact that
$\weightij \geq \buybackweightj \geq (1+f) \buybackweightj$,
we have 
\begin{align*}
    \frac{\Delta(\text{Dual})}{\Delta(\text{Primal})} = \frac{
    (\penscalar + 1)\weightij
    -
    2\penscalar\buybackweightj}{\weightij - (1+\buybackcost)\buybackweightj}=
    \frac{
    \frac{2}{1-f}\cdot
\frac{\weightij}{\buybackweightj}
    -
    2\frac{1+f}{1-f}}{
\frac{\weightij}{ \buybackweightj} - (1+\buybackcost)} = \frac{2}{1-f}.
\end{align*}
Second, consider the case where the allocation is made to the secondary tier  (``if case'' in line 5 of \Cref{alg:primal dual matching determinstic negative F}), therefore:
\begin{align*}
    \frac{\Delta(\text{Dual})}{\Delta(\text{Primal})} = \frac{
    (1-\tau)\weightij}{-fw_{ij}}= \frac{2}{1-f}.
\end{align*}
By summing $\Delta(\text{Dual})$ and $\Delta(\text{Primal})$ over the entire horizon, we obtain:
$$
\textrm{total-profit}(\ALG)\triangleq  \text{Primal}\geq \frac{1}{ \approxratiodet(\buybackcost) }\cdot\text{Dual}
$$
Finally, by weak duality of the linear program, $\text{Dual}\geq \textrm{profit}(\OPT)$, which  finishes the proof.
\hfill\halmos
\end{proof}
\begin{remark}[\revcolorCMTwo{Tightness of the deterministic competitive ratio for $f<0$}] The competitive ratio upper bound of $\frac{2}{1-f}$ in \Cref{thm:competitive ratio determinsitic integral negative F} for deterministic integral algorithms under $f\in[-1,0]$ is also tight, meaning that no deterministic integral online algorithm can obtain an improved competitive ratio. To see this lower bound, consider a simple $2 \times 2$ instance, where $w_{i_1j_1} = w_{i_1j_2} = 1$. Suppose the online algorithm matches online node $i_1$ to offline node $j_1$. Then the adversary introduces a second online node with $w_{i_2j_1} = 1$ (see \Cref{fig:badExample}). The offline optimum selects the edges ${(i_1,j_2), (i_2,j_1)}$, achieving a total value of $2$. In contrast, the best online algorithm obtains at most $1-f$ by using the secondary channel of $j_1$ at the time of arrival of $i_2$, which implies that the competitive ratio is greater than $\tfrac{2}{1-f}$.
\end{remark}
\begin{figure}[h!]
\oldcolor{
\centering
\begin{tikzpicture}[x=2.8cm,y=1.6cm]
  \tikzset{
    vertex/.style={circle, draw, minimum size=8mm, inner sep=0pt}
  }

  \node[vertex] (j1) at (2,  1) {$j_1$};
  \node[vertex] (j2) at (2, -1) {$j_2$};
  \node[vertex] (i1) at (0,  1) {$i_1$};
  \node[vertex] (i2) at (0, -1) {$i_2$};

  \draw (j1) -- (i1);
  \draw (j2) -- (i1);
  \draw (j1) -- (i2);

\end{tikzpicture}
\caption{\centering The offline optimum matches $\{(i_1,j_2),(i_2,j_1)\}$ while the online algorithm matches $\{(i_1,j_1),(i_2,j_1)\}$ where $(i_2,j_1)$ uses the secondary channel.}
}
\label{fig:badExample}
\end{figure}

\begin{remark}[\revcolorCMTwo{Degenerate case and greedy optimality at $f = -1$}]\label{remark:degenerateCasef1}
    When $f=-1$, the restricted main channel provides no benefit compared to the \revcolorCMTwo{secondary} channel. Consequently, the competitive ratio of \Cref{thm:competitive ratio exponential penalty function negative F} simplifies to $\tfrac{e}{e-(1+f)}=1$, also the competitive ratio of \Cref{thm:competitive ratio determinsitic integral negative F} simplifies to $\tfrac{2}{1-f}=1$. In both cases, the algorithm reduces to a straightforward greedy rule that assigns each online node to the offline neighbor with maximum weight, i.e., $\max_j \weight_{ij}$.
\end{remark}
}
\end{document}